\documentclass[format=acmsmall, screen=true, nonacm,pdftex,svgnames]{acmart}
\setcopyright{none}
\usepackage{booktabs}   %% For formal tables:
\usepackage{subcaption} %% For complex figures with subfigures/subcaptions
\usepackage{bussproofs}
\usepackage[cal=boondoxo]{mathalfa}
\DeclareMathAlphabet{\mathpzc}{OT1}{pzc}{m}{it}
\newtheorem{theorem}{Theorem}[section]

\usepackage{color}
\usepackage{xspace}

\usepackage{listings}
\newcommand{\ignore}[1]{{}}

\newcommand{\syntaxDef}[3]{\rulebox{%
\syntaxKeyword$#1\mathrel{::=}{#2}$ \ifthenelse{\equal{#3}{}}{}{[#3]}%
}%
}

\makeatletter
\newcommand{\shorteq}{%
  \settowidth{\@tempdima}{-}% Width of hyphen
  \resizebox{\@tempdima}{\height}{=}%
}
\makeatother

\usepackage{xcolor}
\usepackage{colortbl}

\usepackage[utf8]{inputenc}

\usepackage{amsmath}
\usepackage{hyperref}

\usepackage{subcaption}
\usepackage{wrapfig}

\usepackage{xspace}
\usepackage{float}
\usepackage{balance}

\usepackage{textcomp} %just for a vertical quote

\usepackage{changepage}
\usepackage{array,etoolbox}

\usepackage{algorithm}
\usepackage{algorithmic}

\preto\tabular{\setcounter{magicrownumbers}{0}}
\newcounter{magicrownumbers}

\usepackage{cleveref} % must be loaded after hyperref and amsmath

\usepackage{wrapfig}
\usepackage{caption}

\usepackage{longtable}
\usepackage{tabu}

\usepackage{wasysym}
\usepackage{mathpartir}
\usepackage{mathtools} %psmallmatrix
\usepackage{xfrac} % fractions

\usepackage[qm,braket]{qcircuit}

\colorlet{kwd}{black!80!green}
\definecolor{spec1}{RGB}{78, 131, 162}
\definecolor{spec0}{RGB}{66, 102, 136}
\definecolor{lespec}{RGB}{30, 80, 180}
\colorlet{spec}{lespec}
\colorlet{auto}{lespec!35!lightgray}
\colorlet{stack}{magenta}

\newcommand{\sqir}{SQIR\xspace}

\newcommand{\qwire}{{Qwire}\xspace}
\newcommand{\qbricks}{{QBricks}\xspace}

\newcommand{\voqc}{\textsc{voqc}\xspace}

\newcommand{\myparagraph}[1]{\noindent\paragraph{\textbf{#1}}}

\usepackage{tikz}

\newcommand{\mapp}[2]{#1\circ #2}
\newcommand{\uapp}[2]{#1\, #2}
\newcommand{\one}{\ensuremath{1}}

\newcommand{\slen}[1]{|#1|}
\newcommand{\pau}{\mathpzc{P}}
\newcommand{\CC}{\ensuremath{\mathbb{C}}\xspace}
\usepackage{pgf}

\usepackage{tikz} % circuit diagrams
\usetikzlibrary{%
  arrows,%
  shapes.misc,% wg. rounded rectangle
  shapes.geometric, % diamonds
  shapes.arrows,%
  shapes.callouts,
  shapes.gates.logic.US,
  chains,%
  matrix,%
  positioning,% wg. " of "
  scopes,%
  decorations.pathmorphing,% /pgf/decoration/random steps | erste Graphik
  decorations.text,
  decorations.pathreplacing, % braces
  shadows,%
  automata,
  fit, calc, arrows.meta
}

\tikzset{ machine/.style={
    rectangle,
    minimum width=25mm,
    minimum height=18mm,
    text width=24mm,
    align=center,
    very thick,
    draw=black,
    color=black,
    fill=white,
  }
}

\newcommand\wideparen[1]{%
\tikz[baseline=(wideArcAnchor.base)]{
    \node[inner sep=0] (wideArcAnchor) {$#1$}; 
    \coordinate (wideArcAnchorA) at ($0.9*(wideArcAnchor.north west) + 0.1*(wideArcAnchor.north east)+(0.0em,0.75ex)$);
    \coordinate (wideArcAnchorB) at ($0.1*(wideArcAnchor.north west) + 0.9*(wideArcAnchor.north east)+(0.0em,0.75ex)$);
    \draw[line width=0.1ex,line cap=round] 
        ($(wideArcAnchor.north west)+(0.0em,0.1ex)$) 
            .. controls (wideArcAnchorA) and (wideArcAnchorB) ..
        ($(wideArcAnchor.north east)+(0.0em,0.1ex)$)        
    ;
}}
\newcommand{\cmsg}[1]{\wideparen{#1}}

\DeclarePairedDelimiter\abs{\lvert}{\rvert}
\DeclarePairedDelimiter\norm{\lVert}{\rVert}

\makeatletter
\let\oldabs\abs
\def\abs{\@ifstar{\oldabs}{\oldabs*}}
\let\oldnorm\norm
\def\norm{\@ifstar{\oldnorm}{\oldnorm*}}
\makeatother

\makeatletter
\DeclareRobustCommand{\vardivision}{%
  \mathbin{\mathpalette\@vardivision\relax}% 
}
\newcommand{\@vardivision}[2]{%
  \reflectbox{$\m@th\smallsetminus$}%
}
\makeatother

\newcommand{\ie}{\emph{i.e.,}\xspace}

\usepackage{booktabs}

\usepackage{alltt}
\usepackage{listings,lstcoq}
\definecolor{ltblue}{rgb}{0,0.4,0.4}
\definecolor{dkblue}{rgb}{0,0.1,0.6}
\definecolor{dkgreen}{rgb}{0,0.35,0}
\definecolor{dkviolet}{rgb}{0.3,0,0.5}
\definecolor{dkred}{rgb}{0.5,0,0}
\usepackage[export]{adjustbox}

\newcommand{\code}[1]{{\small\texttt{#1}}}
 \newcommand{\rulelab}[1]{{\small \textsc{#1}}}

\newcommand{\sifb}[3]{\texttt{if}~{(#1)}~{#2}~\texttt{else}~{#3}}

\DeclareMathOperator*{\Motimes}{\text{\raisebox{0.25ex}{\scalebox{0.8}{$\bigotimes$}}}}

\newcommand{\hspc}[2]{\bigotimes^{#1}\mathpzc{H}^{#2}}
\newcommand{\lambdae}[3]{\lambda #1 : #2 .\, #3}
\newcommand{\tjudge}[3]{#1 \vdash #2 \triangleright #3}

\newcommand{\teq}[2]{ #1 \equiv #2}

\newcommand{\ttimes}{\,\textcolor{spec}{\otimes}\,}
\newcommand{\bigttimes}[1]{\textcolor{spec}{\bigotimes}^{#1}\,}

\newcommand{\sapp}[2]{#1\circ #2}
\newcommand{\smu}[3]{\mu #1:#2.\,#3}
\newcommand{\stype}[2]{{#1}\,{\textcolor{spec}{: #2}}}

\newcommand{\zero}{\mathpzc{O}}

\newcommand{\dualb}[3]{#1\ket{#2}_{0}\ket{#3}_{1}}

\newcommand{\sdag}[1]{#1^{\dag}}

\newcommand{\pmx}{\cn{p}}
\newcommand{\hmx}{\cn{h}}
\newcommand{\umx}{\cn{u}}

\newcommand{\rocq}{\textsc{Rocq}\xspace}

\newcommand{\sminus}{\texttt{-}}
\newcommand{\splus}{\texttt{+}}
\newcommand{\eexp}[1]{\cn{exp}(\sminus i #1)}
\newcommand{\pexp}[1]{\cn{exp}(i #1)}
\newcommand{\quan}[3]{\mathpzc{#1}^{#2}(#3)}

\newcommand{\funsa}[3]{\mathpzc{#1}(#2)(#3)}
\newcommand{\elog}[1]{\sminus i\cn{log}(#1)}

\newcommand{\iseq}[2]{{#1}\,{#2}}

\newcommand{\qsnd}{$\textsc{QBlue}$\xspace}
\newcommand{\Hilb}{\mathcal H}
\newcommand{\subcap}[1]{_{\textsc{#1}}}
\newcommand{\complex}{\mathbb C}

\newcommand{\cn}[1]{\texttt{#1}}

\newcommand{\denote}[1]{\llbracket #1 \rrbracket\xspace}
\newcommand{\dabs}[1]{|\!| #1 |\!|}

\newcommand{\tob}[1]{[\!\!( #1 )\!\!]}

\usepackage{newunicodechar}
\let\Alpha=A
\let\Beta=B
\let\Epsilon=E
\let\Zeta=Z
\let\Eta=H
\let\Iota=I
\let\Kappa=K
\let\Mu=M
\let\Nu=N
\let\Omicron=O
\let\omicron=o
\let\Rho=P
\let\Tau=T
\let\Chi=X

\newunicodechar{Α}{\ensuremath{\Alpha}}
\newunicodechar{α}{\ensuremath{\alpha}}
\newunicodechar{Β}{\ensuremath{\Beta}}
\newunicodechar{β}{\ensuremath{\beta}}
\newunicodechar{Γ}{\ensuremath{\Gamma}}
\newunicodechar{γ}{\ensuremath{\gamma}}
\newunicodechar{Δ}{\ensuremath{\Delta}}
\newunicodechar{δ}{\ensuremath{\delta}}
\newunicodechar{Ε}{\ensuremath{\Epsilon}}
\newunicodechar{ε}{\ensuremath{\epsilon}}
\newunicodechar{ϵ}{\ensuremath{\varepsilon}}
\newunicodechar{Ζ}{\ensuremath{\Zeta}}
\newunicodechar{ζ}{\ensuremath{\zeta}}
\newunicodechar{Η}{\ensuremath{\Eta}}
\newunicodechar{η}{\ensuremath{\eta}}
\newunicodechar{Θ}{\ensuremath{\Theta}}
\newunicodechar{θ}{\ensuremath{\theta}}
\newunicodechar{ϑ}{\ensuremath{\vartheta}}
\newunicodechar{Ι}{\ensuremath{\Iota}}
\newunicodechar{ι}{\ensuremath{\iota}}
\newunicodechar{Κ}{\ensuremath{\Kappa}}
\newunicodechar{κ}{\ensuremath{\kappa}}
\newunicodechar{Λ}{\ensuremath{\Lambda}}
\newunicodechar{λ}{\ensuremath{\lambda}}
\newunicodechar{Μ}{\ensuremath{\Mu}}
\newunicodechar{μ}{\ensuremath{\mu}}
\newunicodechar{Ν}{\ensuremath{\Nu}}
\newunicodechar{ν}{\ensuremath{\nu}}
\newunicodechar{Ξ}{\ensuremath{\Xi}}
\newunicodechar{ξ}{\ensuremath{\xi}}
\newunicodechar{Ο}{\ensuremath{\Omicron}}
\newunicodechar{ο}{\ensuremath{\omicron}}
\newunicodechar{Π}{\ensuremath{\Pi}}
\newunicodechar{π}{\ensuremath{\pi}}
\newunicodechar{ϖ}{\ensuremath{\varpi}}
\newunicodechar{Ρ}{\ensuremath{\Rho}}
\newunicodechar{ρ}{\ensuremath{\rho}}
\newunicodechar{ϱ}{\ensuremath{\varrho}}
\newunicodechar{Σ}{\ensuremath{\Sigma}}
\newunicodechar{σ}{\ensuremath{\sigma}}
\newunicodechar{ς}{\ensuremath{\varsigma}}
\newunicodechar{Τ}{\ensuremath{\Tau}}
\newunicodechar{τ}{\ensuremath{\tau}}
\newunicodechar{Υ}{\ensuremath{\Upsilon}}
\newunicodechar{υ}{\ensuremath{\upsilon}}
\newunicodechar{Φ}{\ensuremath{\Phi}}
\newunicodechar{φ}{\ensuremath{\phi}}
\newunicodechar{ϕ}{\ensuremath{\varphi}}
\newunicodechar{Χ}{\ensuremath{\Chi}}
\newunicodechar{χ}{\ensuremath{\chi}}
\newunicodechar{Ψ}{\ensuremath{\Psi}}
\newunicodechar{ψ}{\ensuremath{\psi}}
\newunicodechar{Ω}{\ensuremath{\Omega}}
\newunicodechar{ω}{\ensuremath{\omega}}

\newunicodechar{ℕ}{\ensuremath{\mathbb{N}}}
\newunicodechar{∅}{\ensuremath{\emptyset}}

\newunicodechar{∙}{\ensuremath{\bullet}}
\newunicodechar{≈}{\ensuremath{\approx}}
\newunicodechar{≅}{\ensuremath{\cong}}
\newunicodechar{≡}{\ensuremath{\equiv}}
\newunicodechar{≤}{\ensuremath{\le}}
\newunicodechar{≥}{\ensuremath{\ge}}
\newunicodechar{≠}{\ensuremath{\neq}}
\newunicodechar{∀}{\ensuremath{\forall}}
\newunicodechar{∃}{\ensuremath{\exists}}
\newunicodechar{±}{\ensuremath{\pm}}
\newunicodechar{∓}{\ensuremath{\pm}}
\newunicodechar{·}{\ensuremath{\cdot}}
\newunicodechar{⋯}{\ensuremath{\cdots}}
\newunicodechar{…}{\ensuremath{\ldots}}
\newunicodechar{∷}{~\mathrel{:\!\!\!:}~}
\newunicodechar{×}{\ensuremath{\times}}
\newunicodechar{∞}{\ensuremath{\infty}}
\newunicodechar{→}{\ensuremath{\to}}
\newunicodechar{←}{\ensuremath{\leftarrow}}
\newunicodechar{⇒}{\ensuremath{\Rightarrow}}
\newunicodechar{↦}{\ensuremath{\mapsto}}
\newunicodechar{↝}{\ensuremath{\leadsto}}
\newunicodechar{∨}{\ensuremath{\vee}}
\newunicodechar{∧}{\ensuremath{\wedge}}
\newunicodechar{⊢}{\ensuremath{\vdash}}
\newunicodechar{⊣}{\ensuremath{\dashv}}
\newunicodechar{∣}{\ensuremath{\mid}}
\newunicodechar{∈}{\ensuremath{\in}}
\newunicodechar{⊆}{\ensuremath{\subseteq}}
\newunicodechar{⊂}{\ensuremath{\subset}}
\newunicodechar{∪}{\ensuremath{\cup}}
\newunicodechar{⋓}{\ensuremath{\Cup}}
\newunicodechar{∉}{\ensuremath{\not\in}}
\newunicodechar{√}{\ensuremath{\sqrt}}

\newunicodechar{⊸}{\ensuremath{\multimap}}
\newunicodechar{⊗}{\ensuremath{\otimes}}
\newunicodechar{⨂}{\ensuremath{\bigotimes}}
\newunicodechar{⊕}{\ensuremath{\oplus}}
\newunicodechar{〈}{\ensuremath{\langle}}
\newunicodechar{⟨}{\ensuremath{\langle}}
\newunicodechar{⟩}{\ensuremath{\rangle}}
\newunicodechar{〉}{\ensuremath{\rangle}}
\newunicodechar{¡}{\ensuremath{\upsidedownbang}}
\newunicodechar{∘}{\ensuremath{\circ}}
\newunicodechar{†}{\ensuremath{\dagger}}
\newunicodechar{⊤}{\ensuremath{\top}}
\newunicodechar{⊥}{\ensuremath{\bot}}

\newunicodechar{〚}{\ensuremath{\llbracket}}
\newunicodechar{〛}{\ensuremath{\rrbracket}}

\usepackage{etoolbox} % replaces ifthen package
\newtoggle{comments}
\toggletrue{comments}
  \usepackage[normalem]{ulem}
  \usepackage{minted}

\iftoggle{comments}{
  \newcommand{\fixme}[1]{\textbf{\textcolor{red}{[ Fixme: #1]}}}
  \newcommand{\todo}[1]{\textbf{\textcolor{green}{[ TODO: #1 ]}}}
  \newcommand{\mwh}[1]{\textbf{\textcolor{red}{[ Mike: #1 ]}}}
  
  \newcommand{\khh}[1]{\textbf{\textcolor{orange}{[ Kesha: #1 ]}}}
  \newcommand{\shh}[1]{\textbf{\textcolor{purple}{[ Shih-Han: #1 ]}}}
  \newcommand{\liyi}[1]{\textbf{\textcolor{orange}{[ Liyi: #1 ]}}}
  
 \newcommand{\chand}[1]{\textbf{\textcolor{blue}{[ Chandeepa: #1 ]}}}
  \newcommand{\oth}[2]{\textbf{\textcolor{red}{[ #1: #2 ]}}}
  \newcommand{\xwu}[1]{\textbf{\textcolor{purple}{[ Xiaodi: #1 ]}}}
  
  \newcommand{\ynote}[1]{\textbf{\textcolor{magenta}{[ Yi: #1 ]}}}

  \colorlet{MZ}{violet!80!pink}

  \newcommand{\mzr}[1]{{\color{MZ}{#1}}}
  \newcommand{\was}[1]{}
  
  \NewCommandCopy{\Creff}{\Cref}
  \renewcommand{\Cref}[1]{\mbox{\Creff{#1}}}

  \usepackage[inline]{enumitem}
  \colorlet{LC}{cyan!31!teal}

  \usepackage[inline]{enumitem}
}{
  \newcommand{\fixme}[1]{}
  \newcommand{\todo}[1]{}
  \newcommand{\rnr}[1]{}
  \newcommand{\mwh}[1]{}  
  \newcommand{\khh}[1]{}
  \newcommand{\liyi}[1]{}
  \newcommand{\shh}[1]{}
  \newcommand{\xwu}[1]{}
  \newcommand{\oth}[2]{}
  \newcommand{\mzr}[1]{}
  \newcommand{\chand}[1]{}

  \newcommand{\ynote}[1]{}
}

\newtoggle{submission}
\toggletrue{submission}
\iftoggle{submission}{
  
}{
  
}

\usepackage{bbold}
\usepackage{pifont}% http://ctan.org/pkg/pifont
\usepackage{multicol,tabularx,capt-of}
\usepackage{multirow}

\newcommand{\tferm}{\ensuremath{t^\aleph(2)}}
\begin{document}

%% Title information
%\title{The Quantum Superposition Virtual Machine}         %% [Short Title] is optional;
                                        %% when present, will be used in
                                        %% header instead of Full Title.
%\titlenote{with title note}             %% \titlenote is optional;
                                        %% can be repeated if necessary;
                                        %% contents suppressed with 'anonymous'
\title{A Verified Compiler for Quantum Simulation}                      %% \subtitle is optional
%\subtitlenote{Extended Version}       %% \subtitlenote is optional;
                                        %% can be repeated if necessary;
                                        %% contents suppressed with 'anonymous'

%% Author information
%% Contents and number of authors suppressed with 'anonymous'.
%% Each author should be introduced by \author, followed by
%% \authornote (optional), \orcid (optional), \affiliation, and
%% \email.
%% An author may have multiple affiliations and/or emails; repeat the
%% appropriate command.
%% Many elements are not rendered, but should be provided for metadata
%% extraction tools.

%% Author with single affiliation.

\def\titlerunning{A Verified Compiler for Quantum Simulation}
\def\authorrunning{L. Li, F. Zahariev, C. Dissanayake, J. Swanepoel, A. Sabry, M. Gordon}

\author{Liyi Li}
\affiliation{
 \institution{Iowa State University}
 \country{USA} 
}
\email{liyili2@umd.edu }

\author{Fenfen An}
%\authornote{ORCID: 0000-0002-3675-5670}
\affiliation{
 \institution{Iowa State University}
 \country{USA} 
}
\email{anff@iastate.edu}% 0000-0002-3675-5670

\author{Federico Zahariev}
%\authornote{ORCID: 0000-0002-3223-3576}
\affiliation{
 \institution{Iowa State University}
 \country{USA} 
} 
\email{fzahari@iastate.edu}

%https://orcid.org/0000-0002-3223-3576

\author{Zhi Xiang Chong}
%\authornote{ORCID: 0009-0004-7921-6898}
\affiliation{
 \institution{Iowa State University}
 \country{USA} 
}
\email{ianchong@iastate.edu}% 0009-0004-7921-6898

\author{Amr Sabry}
\affiliation{
  \institution{Indiana University}
  \country{USA}
}
\email{sabry@iu.edu}

\author{Mark S. Gordon}
%\authornote{ORCID: 0000-0001-6893-553X}
\affiliation{
  \institution{Iowa State University}
  \country{USA}
}
\email{mgordon@iastate.edu}%https://orcid.org/0000-0001-6893-553X

%% Abstract
%% Note: \begin{abstract}...\end{abstract} environment must come
%% before \maketitle command
\begin{abstract}
 %Quantum simulations are designed to model quantum systems, and many compilation frameworks have been developed for executing such simulations on quantum computers. 
%These systems are typically Pauli operation-based and lack clear programmability to show the utility of Hamiltonian simulation. In addition, the compilers were fragmented, focusing on a piece of the compilation pipeline, and were not verified.
%This paper proposed \qsnd, the first verified compilation framework for second-quantization-based Hamiltonian simulation.
%We intend to lift the Pauli-based Hamiltonians to constraints based on second quantization and examine the programmability of programming Hamiltonian constraints.
%We then propose a type system to generalize a qubit-based quantum system to quantum particle systems and provide a certified compilation of Hamiltonian simulation for quantum circuits, both digital and analog. All the works are mechanized in Rocq.

Hamiltonian simulation is a central application of quantum computing, with significant potential in modeling physical systems and solving complex optimization problems. Existing compilers for such simulations typically focus on low-level representations based on Pauli operators, limiting programmability and offering no formal guarantees of correctness across the compilation pipeline. We introduce \qsnd, a high-level, formally verified framework for compiling Hamiltonian simulations. \qsnd is based on the formalism of second quantization, which provides a natural and expressive way to describe quantum particle systems using creation and annihilation operators. To ensure safety and correctness, \qsnd includes a type system that tracks particle types and enforces Hermitian structure. The framework supports compilation to both digital and analog quantum circuits and captures multiple layers of semantics, from static constraints to dynamic evolution. All components of \qsnd, including its language design, type system, and compilation correctness, are fully mechanized in the \rocq proof framework, making it the first end-to-end verified compiler for second-quantized Hamiltonian simulation.

%clearly define states across different quantum systems and treat quantum computers as quantum particle systems of specific types. The type system is compatible with the compilation of quantum simulations expressed in \qsnd for digital and analog quantum computers. With \qsnd, users can specify the desired quantum particle system and model the system on quantum computers.

%quantum computers and all user programs as descriptions of particle systems, and the compilation from user programs to quantum computers, definable in \qsnd, can be viewed as the transformation from one particle system to the other.
%We develop a type system in \qsnd to type different particle states to simplify the compilation procedures.

%Circuit-based quantum models suffer from the issue that neither the users nor the target quantum computer machines are related to them; both users and target machines are related to the Hamiltonian concept in quantum mechanics.

%Therefore, an end-to-end compilation from a user program to a quantum machine might suffer unnecessary complications due to the extra circuit layer.

%On the other hand, analog quantum computing views quantum computers as a Hamiltonian simulator and tries to bypass quantum circuit compilation, but it is typically defined in terms of Pauli group operations, restricted to quantum computing settings, and it is very unintuitive.

%We find that user programs and quantum computers can both be viewed as particle systems, describable in second quantization.

\end{abstract}

%% 2012 ACM Computing Classification System (CSS) concepts
%% Generate at 'http://dl.acm.org/ccs/ccs.cfm'.
% <concept_id>10011007.10011006.10011008</concept_id>
% <concept_desc>Software and its engineering~General programming languages</concept_desc>
% <concept_id>10003456.10003457.10003521.10003525</concept_id>
% <concept_desc>Social and professional topics~History of programming languages</concept_desc>
% <concept_significance>300</concept_significance>

% \ccsdesc[500]{Software and its engineering~General programming languages}
% \ccsdesc[300]{Social and professional topics~History of programming languages}
%% End of generated code

%% Keywords
%% comma separated list
%\keywords{keyword1, keyword2, keyword3}  %% \keywords are mandatory in final camera-ready submission

%% \maketitle
%% Note: \maketitle command must come after title commands, author
%% commands, abstract environment, Computing Classification System
%% environment and commands, and keywords command.
\maketitle

%\textcolor{red}{Pointing out quantum conditional proofs somewhere. The essense of qafny is to permit inter-op automated proof.}

\section{Introduction}\label{sec:intro}

Quantum simulation is a flagship application of quantum computing, promising to model physical systems~\cite{Zhang_2021,osti_1782924,du2023multinucleon}, chemical reactions~\cite{Lanyon2010,Cao2019,Evangelista2023,RevModPhys.92.015003}, and materials~\cite{Emani2021,Fedorov2021,wcms.1481} with precision beyond classical capabilities~\cite{feynman1982simulating}. Yet programming such simulations remains precarious. Current compilers for Hamiltonian simulation operate at low-level abstractions, often directly manipulating Pauli strings or untyped tensor expressions, and provide no formal guarantees about the correctness of the resulting circuits~\cite{10.1145/3470496.3527394,10.1145/3632923,cao2024marqsimreconcilingdeterminismrandomness,10.1145/3503222.3507715}. This disconnect between high-level physical intent and low-level implementation risks silent errors, miscompiled models, and untrustworthy scientific outcomes.

We present \qsnd, the first \emph{formally verified compiler} for quantum simulation based on the formalism of \emph{second quantization}. \qsnd allows users to express Hamiltonians in terms of creation and annihilation operators, enforcing Hermiticity and physical validity through a high-level type system. Programs written in \qsnd are compiled into executable quantum circuits or analog schedules with \emph{end-to-end formal guarantees}: every compilation path is verified using a mechanized proof system, and all generated artifacts preserve the intended quantum dynamics within proven error bounds.

To support this, \qsnd employs a three-tiered semantic architecture:
\begin{itemize}
  \item \textbf{Constraint semantics} defines symbolic, well-typed Hamiltonians using second-quantized expressions subject to constraints such as Hermiticity and particle type;
  \item \textbf{Dynamic semantics} interprets the simulation operator $\eexp{\uapp{r}{e}}$ as the time evolution of the Hamiltonian $e$ over a duration $r$;
  \item \textbf{Application semantics} compiles the simulation operator into backend-specific execution artifacts, such as gate-based circuits via Trotterization or QDrift, or weighted interaction graphs for analog Ising platforms.
\end{itemize}

%compile these constraint-level expressions into Pauli-based operators acting on qubit registers, via verified transformations such as the Jordan-Wigner map

These layers are connected via a \emph{type-indexed compilation pipeline}, verified in the \rocq proof framework \cite{bertot2013interactive}. Each transformation is guided by types that track not just dimensions and modes (e.g., bosonic  \footnote{More precisely, we model finite ``boson-like" systems as representing a real boson needs an infinite Hilbert space state.} vs. fermionic  \footnote{For simplicity, we restrict fermions in one system to be of the same kind and the same spin.}), but also matrix kinds (plain, Hermitian, unitary), ensuring that the semantics of quantum simulation is preserved throughout. The type system plays a central role, as it enforces anti-commutation for fermions, validates that matrix expressions are well-formed and Hermitian where required, and ensures that exponentiated Hamiltonians yield unitary operators. Equational rules, canonicalization procedures, and symbolic rewrites are all type-preserving and fully mechanized.

\qsnd supports \emph{multiple backends}, each with distinct compilation paths and correctness theorems. For digital platforms, \qsnd emits standard gate sequences with Trotterized simulation of Pauli terms. For stochastic simulation, it supports QDrift-based sampling. For analog devices, it produces native Ising Hamiltonians for direct scheduling. All paths are type-preserving and formally linked to the semantics of the original program.

We prove a central \emph{compilation correctness theorem}: for any Hermitian-typed Hamiltonian expression $e$ and simulation time $r$, the compiled program $U$ simulates $\exp(-i r e)$ within a formally derived error bound~$\epsilon$. These bounds are computed symbolically as part of the compiler and are backend-specific, capturing, for instance, differences between first-order Trotterization and QDrift. Backend-specific optimizations (e.g., analog graph emission) are verified to preserve semantics modulo these error bounds. 

\myparagraph{Summary of Contributions}
This paper contributes:
\begin{itemize}
  \item A typed second-quantization quantum programming language supporting bosonic and fermionic modes, symbolic operators, and physically meaningful type-indexed semantics;
  \item A verified compilation pipeline that spans constraint, dynamic, and application semantics targeting multiple backends (gate-based, stochastic, and analog) with certified error bounds;
  \item A mechanized metatheory, including type soundness, anti-commutation preservation, equational correctness, and end-to-end compilation theorems;
  \item A complete example of typed, verified Hamiltonian simulation, illustrating how symbolic physics programs can be compiled and executed with formal confidence.
\end{itemize}

\myparagraph{Paper Roadmap}
Section~\ref{sec:hubbard-example} introduces the two-site Hubbard model as a running example, showcasing how \qsnd encodes physical systems via typed creation and annihilation operators.  
Section~\ref{sec:formalism} reviews mathematical preliminaries from quantum simulation and second quantization.  
Section~\ref{sec:formal} presents the core formal system of \qsnd, including syntax, type system, and equational theory.  
Section~\ref{sec:compilation} describes the compilation pipeline in detail, including canonicalization, backend-specific transformations, and certified error bounds.  
Section~\ref{sec:related} discusses related work across quantum simulation languages and verified compilers.  
Finally, Section~\ref{sec:conclusion} summarizes contributions and outlines future directions in verified quantum programming.

The \rocq proofs mentioned in the paper are available on Zenodo \url{https://zenodo.org/records/15852130}, anonymously. We additionally include eight appendices: Appendix~\ref{app:paulimap} provides in-depth background; Appendix~\ref{appx:trotterthm} collects lemmas on Trotterization error bounds; Appendix~\ref{appx:gadget} analyzes the perturbative gadget construction; Appendix~\ref{sec:jw-trans} surveys other particle transformation methods; Appendix~\ref{sec:hubbard} presents a Hubbard model for hydrogen chains; and Appendix~\ref{sec:othermodels} outlines other physical models definable in QBlue.

\ignore{
A singleton \qsnd operation can be understood as creating (a creator) or eliminating (an annihilator) electrons appearing in particle sites.
Modeling \qsnd as a particle system is a trivial but useful generalization of the the Hamiltonian simulator view of quantum computers,
because it coincides with the program views of the main quantum computer users from the scientific computing fields, such as in physics \cite{Zhang_2021,osti_1782924,du2023multinucleon}, chemistry \cite{Lanyon2010,Cao2019,Evangelista2023,RevModPhys.92.015003}, and computational biology \cite{Emani2021,Fedorov2021,wcms.1481}.

\begin{figure}[h]
    \includegraphics[width=0.6\textwidth]{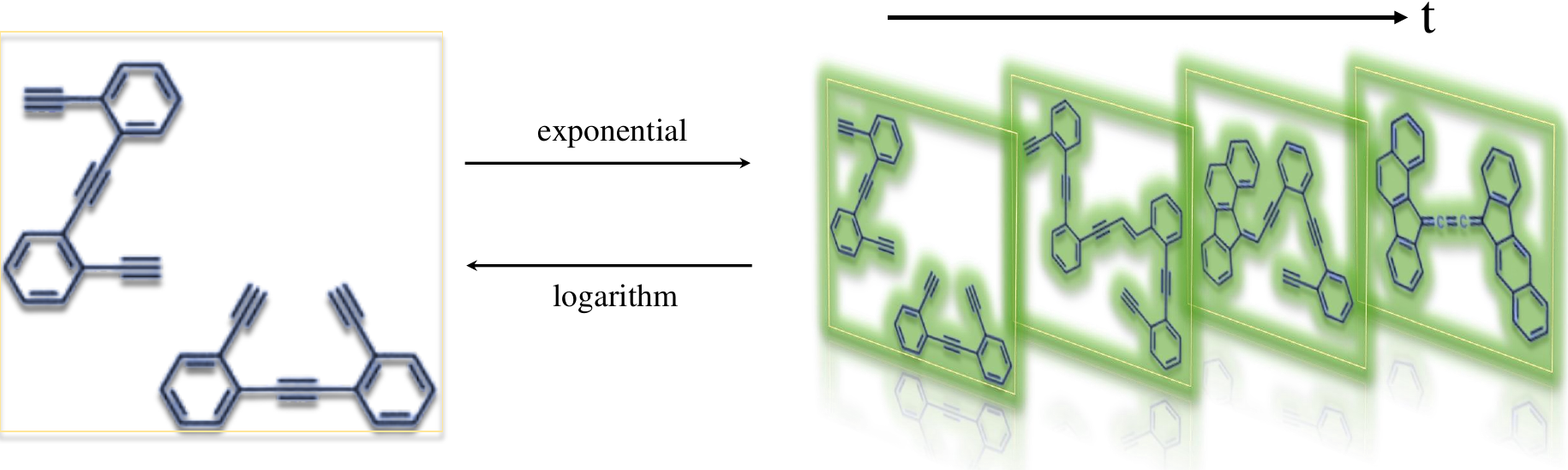}
            \caption{Hamiltonian as a snapshot (left) and quantum simulation (unitary) as movie generation (right); $\cn{exp}$ and $\cn{log}$ perform switching between the two modes.}\label{fig:ham-sim}   
\end{figure}

The particle system program view can be thought of as abstract classes or high order functions in traditional programming fields.
The physical model systems are essentially high order abstract classes physicists defined to represent a general pattern, and when dealing with a specific problem, such as hydrogens, they instantiate the abstract class with specific functions. Such programming style can be extended to solving classical computation problems, by using these physical abstract classes, as the examples shown in \Cref{sec:overview}.

%where they define physical systems in terms of Hamiltonians with applications, such as performing ground energy state computation and Hamiltonian simulation. 
%Lattice-based Hamiltonian matrices (\Cref{sec:background}) describe particle interactions by locating particles in a fixed-site structure formed as a lattice, and Hamiltonian simulation simulates the Hamiltonian's transition behavior for a quantum state over a period of time.
%In this program view, the quantum computer behavior is a procedure of a lattice-based particle system Hamiltonian simulation.
%A trivial generalization is to generalize the Hamiltonian simulator view of quantum computers as a particle system, describable by second quantization, a standard physical formalism for describing quantum particle movements. 

Generalizing the circuit-based model to \qsnd allows the possibility of connecting different groups together to provide better optimization.
Previously, there are rich optimization works related to the circuit model,
but few of them are linked to users and machines, so many optimizations are performed without context,
e.g., the compilation in \Cref{sec:sim-gates} shows some circuit-level compilation, and optimization might be counterproductive if a target machine is provided.

The \qsnd formalism also provides type information to support programming in the user fields.
The current second quantization formalisms are based on Fork space (\Cref{sec:background}),
referring to that a particle state can be described by tensor products of arbitrary dimensional Hilbert spaces, without types.
A main utility of the particle system view is to transform between different systems;
thus the untypedness caused confusions in such transformation, i.e., transforming fermion systems to quantum computers \cite{Cervera_Lierta_2018,Yang_2020} are unintuitive
and some steps might be unnecessarily complicated.
The \qsnd type system intends to classify different particle systems by providing particle types,
so users can rely on the type information to understand and perform computation.
}
\ignore{

Third, different physical systems have different characteristics and second-quantization formalisms.
High-order functions and typing in \qsnd is the way to classify these characteristics.
Fock space (\Cref{sec:background}) is the quantum state formalism in second quantization and suggests one can use tensor products of arbitrary Hilbert spaces to describe a particle system, i.e., it is untyped; such untypedness did not provide a good intuition in understanding quantum systems, e.g., the interpretation of a particle state having more than one spin states is unintuitive.
In addition, fermionic and bosonic systems might have different behaviors.
\qsnd provides a type system to classify all these systems. Under \qsnd, a quantum computer can be typed as a particle system with all its particles having two-dimensional Hilbert space with only one spin state.

Programming in second quantization requires the switching between the two modes, which can be viewed as function applications in a functional computation model, such as $\lambda$-calculus.
In \qsnd, we develop operations representing matrix exponential and logarithm to capture the mode-switching behaviors, model similar operations succinctly through high-order functions, and distinguish differences through typing.
For example, \qsnd models matrix multiplications and energy computation operations as function applications,
with the use of typing relations to distinguish the input and output types.

%that desperately need the quantum computing power, from the scientific computing fields and try to solve hard problems in physics \cite{Zhang_2021,osti_1782924,du2023multinucleon}, chemistry \cite{Lanyon2010,Cao2019,Evangelista2023,RevModPhys.92.015003}, and computational biology \cite{Emani2021,Fedorov2021,wcms.1481};
%such problems are typically defined in terms of Hamiltonians with applications on these Hamiltonians, such as performing ground energy state computation and Hamiltonian simulation.

We propose \qsnd, a simple typed $\lambda$-calculus style computation model based on lattice-based Hamiltonian described by second quantization operations,
a standard physical formalism for describing quantum particle movements.

These facts indicate the development of a general quantum computing model extending the current circuit-based model.
%The key observations for the model development is that circuit-based quantum computers and analog quantum simulators are essentially the same, where the control pulses, the layer directly connecting to quantum computers, describable by a kind of target machine level lattice-based Hamiltonians with particle sites in the lattice being qubits.
%
In the general compilation procedure, the black path in \Cref{fig:process}, only the intermediate layer, is described by quantum unitary circuits, with all others represented by Hamiltonians.
The analog simulator compilation, the blue path, tries to map a physical system Hamiltonian to target machine Hamiltonians directly; thus, it removes many intermediate steps of compiling to an intermediate quantum circuit.
From the user perspective,  the users that desperately need the quantum computing power might be from scientific computing fields and try to solve hard problems in physics \cite{Zhang_2021,osti_1782924,du2023multinucleon}, chemistry \cite{Lanyon2010,Cao2019,Evangelista2023,RevModPhys.92.015003}, and computational biology \cite{Emani2021,Fedorov2021,wcms.1481};
such problems are typically defined in terms of Hamiltonians with applications on these Hamiltonians, such as performing ground energy state computation and Hamiltonian simulation.

We propose \qsnd, a simple typed $\lambda$-calculus style computation model based on lattice-based Hamiltonian described by second quantization operations,
a standard physical formalism for describing quantum particle movements.
In \qsnd, particles can be thought of as living in some lattice sites, expressed as matrix entries in a Hamiltonian.
A singleton \qsnd operation can be understood as creating (a creator) or eliminating (an annihilator) electrons appearing in particle sites.

%Many works \cite{VOQC,10.1145/3519939.3523433,oracleoopsla,ccx-adder,quilc,Fagan2018,reverC,scaffCCnew,Nam2018,quantumssa,ripple-carry-mod,qft-adder} are proposed to optimize quantum circuits extensively to run quantum algorithms in NISQ computers, though NISQ computers are still far from executing quantum algorithms showing quantum advantage.
%Recently, researchers \cite{Daley2022,10.1145/3632923} investigated the possibility of analog quantum computing, i.e., they viewed a quantum computer a lattice-based analog Hamiltonian simulator, such as the blue path in \Cref{fig:process}, to simulate a physical model system \footnote{It is called physical models in many literature, to avoid confusion, we call it a physical model system or physical system.} Hamiltonian by directly mapping the Hamiltonian to the kinds of Hamiltonians the machine supports, which shows some performance advantages; they have yet to establish a general compilation strategy to perform analog computing.

These facts indicate the development of a general quantum computing model extending the current circuit-based model.
The key observations for the model development is that circuit-based quantum computers and analog quantum simulators are essentially the same, where the control pulses, the layer directly connecting to quantum computers, describable by a kind of target machine level lattice-based Hamiltonians with particle sites in the lattice being qubits.
In the general compilation procedure, the black path in \Cref{fig:process}, only the intermediate layer, is described by quantum unitary circuits, with all others represented by Hamiltonians.
The analog simulator compilation, the blue path, tries to map a physical system Hamiltonian to target machine Hamiltonians directly; thus, it removes many intermediate steps of compiling to an intermediate quantum circuit.
From the user perspective,  the users that desperately need the quantum computing power might be from scientific computing fields and try to solve hard problems in physics \cite{Zhang_2021,osti_1782924,du2023multinucleon}, chemistry \cite{Lanyon2010,Cao2019,Evangelista2023,RevModPhys.92.015003}, and computational biology \cite{Emani2021,Fedorov2021,wcms.1481};
such problems are typically defined in terms of Hamiltonians with applications on these Hamiltonians, such as performing ground energy state computation and Hamiltonian simulation.

We propose \qsnd, a simple typed $\lambda$-calculus style computation model based on lattice-based Hamiltonian described by second quantization operations,
a standard physical formalism for describing quantum particle movements.
In \qsnd, particles can be thought of as living in some lattice sites, expressed as matrix entries in a Hamiltonian.
A singleton \qsnd operation can be understood as creating (a creator) or eliminating (an annihilator) electrons appearing in particle sites.

\qsnd intends to generalize the quantum computing circuit model to the one viewing applications as particle movements, including quantum computer applications, and expressions in \qsnd can typed in two different modes.
We first examine how second quantization users think of physical system programs.

The functional and typing views have several benefits.
First, we provide a unified platform to connect different user groups and communicate using the same standard.
Many previous systems are restricted in a specific subdomain, such as quantum circuits (the works above), pulse and machine level programming \cite{10.1145/3632923,Li_2022}, and user level systems \cite{Zahariev2023,10.1063/5.0004997}; the fact leads to possibly unnecessary transformations between these systems, such as the circuit model compilation issues above, and misinterpretation, e.g., some works \cite{Cervera_Lierta_2018,Yang_2020} implemented a physical system in quantum circuits, based on a limited setting, e.g., modeling spin states directly as qubits ($\ket{\uparrow}$  as $\ket{1}$ and $\ket{\downarrow}$ as $\ket{0}$) only works for Ising systems, and causing other works \cite{Yang_2020} based on the limited setting have unnecessary complications in transformation. Having a uniform model to show the difference is the first step in overcoming the complications and confusion in the field.

}

\ignore{
Previously, the former world \cite{vandenBerg2020circuitoptimization,Smith2019a,Berry_2015b,10.5555/2481569.2481570} typically shows that the ability of implementing Hamiltonian simulation via quantum computer gates without showing the meaning of the gate combination,
while the latter world (the works above) represents problems in terms of Hamiltonians and provided some computation on the Hamiltonians, without indicating if the computation can be performed in a quantum computer. \qsnd provides the platform for the two worlds to communicate together.
Especially, the linking between \qsnd and $\lambda$-calculus permits authors to understand quantum systems in terms of intuitive functional language,
such that Hamiltonian operations can be understood as function applications.
By doing so, we also lift the quantum algorithm finding, previously restricted to quantum circuit computing, to a broader high-performance computing (HPC) field.
There is a popular critization \cite{JackKrupansky,MattSwayne,RichardWaters} that quantum computers lack of applications.
The designs of new quantum algorithms should be in the lifted field instead of restricting to quantum computing fields,
since many algorithms showing quantum advantage act as subroutines of a large application.
For example, Energy computation in quantum chemistry \cite{Lee2023} might not have quantum advantage in executing it in a quantum computer,
but some subroutine computation might experience quantum advantage \cite{BLOMBERG2006969,Kovyrshin2023}.

Regarding the user perspective, the users that desperately need the quantum computing power might be from scientific computing fields and try to solve hard problems in physics \cite{Zhang_2021,osti_1782924,du2023multinucleon}, chemistry \cite{Lanyon2010,Cao2019,Evangelista2023,RevModPhys.92.015003}, and computational biology \cite{Emani2021,Fedorov2021,wcms.1481};
such problems are typically defined in terms of Hamiltonians with applications on these Hamiltonians, such as performing ground energy state computation and Hamiltonian simulation.
If both the user and the underlying machine levels of quantum computing are performing applications related to Hamiltonians, it indicates strongly the needs of a new Hamiltonian-based intermediate representation language for expressing quantum computing algorithms.

1. machine reality ---> quantum computers are simulators. analog quantum computing ---> directly implement problems in quantum computers.

2. user prosepctive, many applications are scintific comptuing problems, such as ....

3. if users and machines are both hamiltonians, why using unitary in the middle? % might not want to say this.

4. provide a chance to examine the necciecity of linking the two.
Whether or not we should use Hamiltonian or unitary, it is time to develop a user (programmer) friendly framework to provide a generalized view of quantum computing in the context of quantum mechanics.

reasons. 

a. time to connect users and machines. Many current works only provide one or two example cases, tend to merge everything together by mapping an example system in quantum circuits and states, probably with optimizations to make a case. Then, different way of mapping for optimization causes confusion on learning about original system, and how to compile systems to quantum computers in general, e.g., the descriptions of the physical systems are misunderstood due to the understand of the compiled version (CITE).

realize what quantum computers actually are t(1) type bosonic system.

b. correct optimizations in the machine and user context. current optimizations have no context, with the user and machine level details. 
enabling the analysis of viewing quantum computers as simulators, bypassing the circuit models.

c. provide typing for long existing problems in describing physical systems. fork states are not good for describing systems. a formalism of second quantizations. differet verions discring different systems, without having the context hard to understand, so we need a type system.

Having a functional computation model for second quantization, high-order functions for including both fermions and bosons. context switching between snapshot and simulation mode.

There are two observations based on machine reality and user prospective.
In the quantum machine level, each elementary quantum gate is implemented as controlled pulses, which is equivalent to perform a Hamiltonian simulation, shown in \Cref{fig:process}.
Essentially, every quantum computer has a fixed form of recognizable Hamiltonian.
To implement a quantum gate, we perform an eigendecomposition of the gate, construct the gate Hamiltonian that is recognizable by the quantum computer,
and simulate a certain time on the Hamiltonian as the simulation of the gate in the quantum computer.
This procedure represents the controlled pulse level programming in a quantum computer and it unveils that a quantum computer, in its machine level, is a Hamiltonian simulator that simulates gate behaviors.

}

%There are several reasons for the \qsnd development.
%The primary reason is to connect the quantum computing and the user worlds.
 
\ignore{
In quantum field theory, it is known as canonical quantization, in which the fields (typically as the wave functions of matter) are thought of as field operators, in a manner similar to how the physical quantities (position, momentum, etc.) are thought of as operators in first quantization. The key ideas of this method were introduced in 1927 by Paul Dirac \cite{Dirac1988}, and were later developed, most notably, by Pascual Jordan \cite{Jordan1928} and Vladimir Fock \cite{Fock1932,Reed:1975uy}. In this approach, the quantum many-body states are represented in the Fock state basis, which are constructed by filling up each single-particle state with a certain number of identical particles \cite{Becchi:2010}. The second quantization formalism introduces the creation and annihilation operators to construct and handle the Fock states, providing useful tools to the study of the quantum many-body theory.

We propose a 

Second quantization, also referred to as occupation number representation, is a formalism used to describe and analyze quantum many-body systems. In quantum field theory, it is known as canonical quantization, in which the fields (typically as the wave functions of matter) are thought of as field operators, in a manner similar to how the physical quantities (position, momentum, etc.) are thought of as operators in first quantization. The key ideas of this method were introduced in 1927 by Paul Dirac \cite{Dirac1988}, and were later developed, most notably, by Pascual Jordan \cite{Jordan1928} and Vladimir Fock \cite{Fock1932,Reed:1975uy}. In this approach, the quantum many-body states are represented in the Fock state basis, which are constructed by filling up each single-particle state with a certain number of identical particles \cite{Becchi:2010}. The second quantization formalism introduces the creation and annihilation operators to construct and handle the Fock states, providing useful tools to the study of the quantum many-body theory.

In this paper, we develop a programming language for second quantization, \qsnd, based on the combination of the second qauntization formalism \cite{Dirac1988} and simple typed $\lambda$-calculus \cite{Church1940,Church1956}.

Developing more and more comprehensive quantum programs and algorithms is essential for the continued practical development of quantum computing \cite{JackKrupansky,MattSwayne}.
Unfortunately, because quantum systems are inherently probabilistic and also must obey laws of quantum physics, traditional validation techniques based on run-time testing are virtually impossible to develop for large quantum algorithms,
This leaves \emph{formal methods} as a viable alternative for program checking, and yet these typically require a great effort; for example,
four experienced researchers needed two years to formally verify Shor's algorithm \cite{shorsprove}.

To alleviate the effort required for formal verification, many frameworks have been proposed to verify quantum algorithms \cite{qhoreusage,qhoare,qbricks,qsepa,qseplocal,VOQC}
using interactive theorem provers, such as Isabelle, \rocq, and Why3, by building quantum semantic interpretations and libraries.
Some attempts towards proof automation have been made by creating new proof systems for quantum data structures such as Hilbert spaces; however,
building and verifying quantum algorithms in these frameworks are still time-consuming and require great human effort.
Meanwhile, automated verification is an active research field in classical computation with many proposed frameworks
\cite{HoareLogic,separationlogic,nat-proof-fun,nat-proof,nat-proof-frame,10.1145/3453483.3454087,arxiv.1609.00919,martioliet00rewriting,rosu-stefanescu-2011-nier-icse,rosu-stefanescu-ciobaca-moore-2013-lics,10.1007/978-3-642-17511-4_20,10.1007/978-3-642-03359-9_2}
showing strong results in reducing programmer effort when verifying classical programs. 
%Can classical automated verification frameworks be utilized for verifying quantum programs?
None of the existing quantum verification frameworks utilize these classical verification infrastructures, however.

Second quantization, also referred to as occupation number representation, is a formalism used to describe and analyze quantum many-body systems. In quantum field theory, it is known as canonical quantization, in which the fields (typically as the wave functions of matter) are thought of as field operators, in a manner similar to how the physical quantities (position, momentum, etc.) are thought of as operators in first quantization. The key ideas of this method were introduced in 1927 by Paul Dirac \cite{Dirac1988}, and were later developed, most notably, by Pascual Jordan \cite{Jordan1928} and Vladimir Fock \cite{Fock1932,Reed:1975uy}. In this approach, the quantum many-body states are represented in the Fock state basis, which are constructed by filling up each single-particle state with a certain number of identical particles \cite{Becchi:2010}. The second quantization formalism introduces the creation and annihilation operators to construct and handle the Fock states, providing useful tools to the study of the quantum many-body theory.

In this paper, we develop a programming language for second quantization, \qsnd, based on the combination of the second qauntization formalism \cite{Dirac1988} and simple typed $\lambda$-calculus \cite{Church1940,Church1956}.

\liyi{1. Quantum Computation limitation. 2. describe the situation: 1) the main customers of quantum computers are scientific computing, 2) machine is implemented in terms of Ham simulation (see overview, describe a little). --> why we ever need quantum circuit? 3. main point is to create a language for scientific computing (marry the world of physics and CS). 4. (maybe not be good) Remove the circuit level and create a language for directly linking customers and mechine basis. }

\liyi{Describe what Ham can do. How physics view the world --> snapshot of Ham, and unitary for simulation for movies. }
}
                   % intro
\section{A Guided Example: \qsnd on a Two-Site Hubbard Model}
\label{sec:hubbard-example}

To illustrate how \qsnd expresses, compiles, and verifies quantum models, we present a complete example: the fermionic Hubbard model on two lattice sites. This minimal system captures key features of many-body quantum mechanics, including particle motion, local interactions, and coherent evolution, while remaining compact enough to explore in a single section. Conceptually, it models a simple chemical reaction in which electrons tunnel between orbitals and repel each other when co-located, offering a physically grounded yet abstractly expressive setting for quantum programming and formal verification.

\myparagraph{Model Setup and Physical Operators}
The fermionic Hubbard model captures essential features of strongly correlated systems, such as quantum tunneling and on-site repulsion. Although originally formulated for extended lattices, even a two-site version reveals the core phenomena of many-body dynamics. We consider a one-dimensional lattice with two fermionic sites, labeled site~$0$ and site~$1$. Each site may be either \emph{unoccupied} or \emph{occupied} by a single fermion. The occupation number at each site is given by a value in $\{0,1\}$ (for bosons, the occupation number can be bigger than $1$, see \Cref{sec:formal}), where $0$ means no particle and $1$ means a particle is present. These values describe physical quantum states, not Boolean truth values, and are distinct from the numeric indices labeling the sites.

The total Hilbert space is $\mathbb{C}^2 \otimes \mathbb{C}^2$, with computational basis states $\ket{m}\ket{n}$, where $m$ and $n$ indicate the occupation of site~$0$ and site~$1$, respectively. Three standard fermionic operators define the system’s dynamics. The \emph{creation operator}~$\sdag{a}(j)$ adds a particle to site~$j$ if it is unoccupied: for instance, $\sdag{a}(0)\ket{0}\ket{0} = \ket{1}\ket{0}$, and $\sdag{a}(1)\ket{m}\ket{0} = (-1)^m\ket{m}\ket{1}$. Applying $\sdag{a}(j)$ to an occupied site yields the vacuous state~$\zero$. The \emph{annihilation operator}~$a(j)$ removes a particle if one is present: for example, $a(0)\ket{1}\ket{0} = \ket{0}\ket{0}$ and $a(1)\ket{m}\ket{1} = (-1)^m\ket{m}\ket{0}$. Again, acting on an unoccupied site yields $\zero$. The \emph{number operator}~$\mathbb{1}(j) = \mapp{\sdag{a}(j)}{a(j)}$ (and its complement $\mathbb{0}(j) = \mapp{a(j)}{\sdag{a}(j)}$) projects onto the subspace where site~$j$ is occupied (or unoccupied, respectively).\footnote{In traditional literature, the number operator is denoted \( n(j) = \sdag{a}(j) a(j) \); here, we write it as \( \mathbb{1}(j) \) to emphasize its role as a projector onto the occupied subspace.} This composition first annihilates a particle if present and then recreates it, preserving the state. For instance, $\mathbb{1}(0)\ket{1}\ket{0} = \ket{1}\ket{0}$, while $\mathbb{1}(0)\ket{0}\ket{0} = \zero$.

\myparagraph{Constraint Semantics: Hamiltonian as a Typed Expression}
The physical behavior of the system is captured as a constraint-level expression in \qsnd:
\[
\hat{H} =
  - \left( \sapp{\sdag{a}(0)}{a(1)} + \sapp{\sdag{a}(1)}{a(0)} \right)
  + 2 \cdot \sapp{\mathbb{1}(0)}{\mathbb{1}(1)}.
\]
The first component models coherent \emph{tunneling} of a fermion between the two sites. The term $\mapp{\sdag{a}(0)}{a(1)}$ removes a particle from site~1 and creates it at site~0, and vice versa for the mirror term. Together, they ensure symmetric transition amplitudes and reversible evolution. The second term models \emph{on-site repulsion}: the state $\ket{1}\ket{1}$, where both sites are occupied, incurs an energy cost of~2.

Typing ensures this expression is well-formed. The type system classifies particle operators by the types of the sites  ($\iota$) they act on. Here, each site has type~$\tferm$, representing a spinless fermionic mode. The creation and annihilation operators act on single $\tferm$-typed sites, and compositions like $\mapp{\sdag{a}(0)}{a(1)}$ act over $\tferm \ttimes \tferm$, two sites each typed as $\tferm$ ($\ttimes$ is marked \textcolor{spec}{blue} when referring to type operators). The entire Hamiltonian is typed as $\quan{F}{\hmx}{\tferm \ttimes \tferm}$, meaning it is a Hermitian constraint-level expression over fermionic types.
This forms the \emph{constraint semantics} of the program: a typed specification of allowable physical interactions. 

Note that a single \qsnd operator, $\sdag{a}$ or $a$, is not Hermitian. \qsnd utilizes a \emph{kind flag} to identify expressions as Hermitian ($\hmx$) or non-Hermitian ($\pmx$), e.g., $\sdag{a}$ is typed as $\quan{F}{\pmx}{\tferm}$. We then utilize type-guided equational theory to identify Hermitian constraint-level expressions. Via the \qsnd type soundness theorem, we unveil the following. 

\begin{proposition}[Unitary Existence of Hermitian Constraints]\label{thm:type-good-simp}\rm 
Given a Hamiltonian $e$, typed as $\tjudge{\iota}{e}{\quan{F}{\hmx}{\iota}}$, for all $r$, the $\eexp{\uapp{r}{e}}$ semantics is properly defined in \qsnd.
\end{proposition}

The lemmas are given in \Cref{sec:theorem}. The fact ensures that every $\hmx$-kind expression can be simulated via our Hamiltonian simulation operator, explained shortly below. 

\myparagraph{Dynamic Semantics: Mapping to Pauli Operators}
To simulate the model, \qsnd compiles constraint-level expressions into qubit-level operators using the Jordan-Wigner transformation. This yields the \emph{dynamic semantics}, a.k.a., Hamiltonian simulation ($\eexp{\uapp{r}{e}}$), which governs unitary evolution on qubit registers.

The number operator maps to $\mathbb{1}(j) = \frac{I - Z(j)}{2}$ \footnote{We use a single $I$ without any indices to indicate a tensor of identity operations applied to all sites.}, where $Z(j)$ is the Pauli-$Z$ matrix. The tunneling term becomes $\frac{1}{2}(X(0) \circ X(1) + Y(0) \circ Y(1))$, where $X(j)$ and $Y(j)$ are Pauli-$X$ and $Y$ matrices. Substituting yields:
\[
\hat{H} = -\frac{1}{2}(X(0) \circ  X(1) + Y(0) \circ Y(1))
  + 2 \cdot \left(\frac{I - Z(0)}{2}\right)\circ\left(\frac{I - Z(1)}{2}\right),
\]
which simplifies to:
\[
\hat{H} = -\frac{1}{2}(X(0)\circ X(1) + Y(0) \circ Y(1))
  + \frac{1}{2}(I - Z(0) - Z(1) + Z(0) \circ Z(1)).
      \tag{2.1}
  \label{eq:ex1}
\]
Dropping constant and local-field terms yields the effective Hamiltonian:
\[
\hat{H}' = -\frac{1}{2}(X(0) \circ X(1) + Y(0) \circ Y(1)) + \frac{1}{2} Z(0) \circ Z(1).
    \tag{2.2}
  \label{eq:ex}
\]

Transforming $\hat{H}$ to $\hat{H}'$ introduces some system errors, explained shortly below.
This form is used for digital simulation or direct execution on analog backends. In both cases, type-correctness ensures that the transformation preserves the intended physical behavior.
Extending the result from our type soundness, we show that every $\hmx$-kind constraint expression $e$ can be rewritten into a Pauli-based expression with a canonicalized form, as follows.

\begin{proposition}[Always Canonicalizible]\label{thm:type-good-simp}\rm 
Given a $n$-site Hamiltonian $e$, typed as $\tjudge{\iota}{e}{\quan{F}{\hmx}{\iota}}$, $e$ can be canonicalized to a Pauli-string-based Hamiltonian $\hat{H}=\sum_j r_j \cdot (\bigotimes^{n\sminus 1}_{k=0} \pau_{j,k})$ with $\pau\in\{X, Y, Z, I\}$.
\end{proposition}

The lemmas are shown in \Cref{sec:transformation,sec:compilecanonical}. We show that the canonicalized Pauli-string-based Hamiltonian can always be in the form of a sum of pairs of a real-number amplitude with a Pauli string of single-Pauli operators.

\myparagraph{Application Semantics: Target-Specific Execution}
The final step assigns meaning to the Hamiltonian through backend-specific execution---the \emph{application semantics}. On gate-based quantum devices, \qsnd applies Trotterization to approximate time evolution. For $r = \frac{\pi}{4}$, we obtain \footnote{Hamiltonian simulation $\eexp{r \hat{H}}$ typically requires $r$ to be positive, while negative time simulation can be done by inverting the circuit. The restriction is on the analog schedule generation, as we require $r$ to be positive in every step.}:
\[
\eexp{r \hat{H}'} \approx 
  \eexp{\frac{-\pi}{8} X(0) \circ X(1)}
  \circ
  \eexp{\frac{-\pi}{8} Y(0) \circ Y(1)}
  \circ
  \eexp{\frac{\pi}{8} Z(0) \circ Z(1)}.
      \tag{2.3}
  \label{eq:ex_final}
\]
Each exponential is implemented using standard gate patterns:
\begin{itemize}
  \item $X(0) \circ X(1)$: Hadamard basis change, \cn{CNOT}–$\cn{Rz}$–\cn{CNOT}, undo basis.
  \item $Y(0) \circ Y(1)$: rotate via $S^\dagger$ and Hadamards, then same as above.
  \item $Z(0) \circ Z(1)$: directly synthesized using \cn{CNOT}–$\cn{Rz}$–\cn{CNOT}.
\end{itemize}

\ignore{
\begin{figure}[h]
\centering
\[
\Qcircuit @C=0.8em @R=0.8em {
  \lstick{\ket{q_0}} & \gate{\cn{H}} & \ctrlo{1} & \gate{\cn{Rz}(-\pi/4)} & \ctrlo{1} & \gate{\cn{H}} 
                    & \gate{\cn{S}^\dagger} & \gate{\cn{H}} & \ctrlo{1} & \gate{\cn{Rz}(-\pi/4)} & \ctrlo{1} & \gate{\cn{H}} & \gate{\cn{S}} 
                    & \ctrlo{1} & \gate{\cn{Rz}(\pi/4)} & \ctrlo{1} & \qw \\
  \lstick{\ket{q_1}} & \gate{\cn{H}} & \control \qw & \qw & \control \qw & \gate{\cn{H}} 
                    & \gate{\cn{S}^\dagger} & \gate{\cn{H}} & \control \qw & \qw & \control \qw & \gate{\cn{H}} & \gate{\cn{S}} 
                    & \control \qw & \qw & \control \qw & \qw
}
\]
\vspace{-1.5em}
\caption{Quantum circuit for one Trotter step of evolution under $\hat{H}'$, with $\theta = \pi/8$. Terms are applied in the order $X(0)\circ X(1)$, $Y(0) \circ Y(1)$, $Z(0) \circ Z(1)$.}
\label{fig:hubbard-trotter-circuit}
\end{figure}
}

\begin{figure}[h]
\centering
\[
\Qcircuit @C=0.8em @R=0.8em {
  \lstick{\ket{q_0}} & \ctrlo{1} & \gate{\cn{Rz}(\pi/4)} & \ctrlo{1}  
                    & \gate{\cn{S}} & \gate{\cn{H}} & \ctrlo{1} & \gate{\cn{Rz}(-\pi/4)} & \ctrlo{1} & \gate{\cn{H}} &  \gate{\cn{S}^{\dagger}} & \gate{\cn{H}} 
                    & \ctrlo{1} & \gate{\cn{Rz}(-\pi/4)} & \ctrlo{1} & \gate{\cn{H}} &\qw \\
  \lstick{\ket{q_1}} & \control \qw & \qw & \control \qw 
                    & \gate{\cn{S}} & \gate{\cn{H}} & \control \qw & \qw & \control \qw & \gate{\cn{H}} & \gate{\cn{S}^{\dagger}} & \gate{\cn{H}}
                    & \control \qw & \qw & \control \qw & \gate{\cn{H}} & \qw
}
\]
\vspace{-1.5em}
\caption{Quantum circuit for one Trotter step of evolution under $\hat{H}'$, with $\theta = \pi/8$. Terms are applied in the order $Z(0) \circ Z(1)$, $Y(0) \circ Y(1)$, $X(0)\circ X(1)$.}
\label{fig:hubbard-trotter-circuit}
\end{figure}
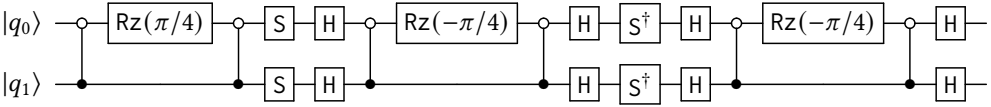

In \qsnd, we permit different compilation paths, e.g., users can select a different Trotterization algorithm, QDrift, to compile the Hamiltonian simulation.
We also permit analog-based compilation. For analog platforms such as the IBM~\cite{Alexander_2020} and Indiana~\cite{richerme2025multimodeglobaldrivingtrapped} machines, the compiler targets native Ising models. For the IBM machine model, these are of the form $\sum_j r_j \cdot X(j) + \sum_k r_k \cdot Z(k) + \sum_{j<k} \theta_{jk} \cdot Z(j) \circ Z(k)$ \cite{10.1145/3632923}. The interaction term $Z(0) \circ Z(1)$ is directly compatible with the native couplers. In this case, \qsnd emits a weighted interaction graph suitable for analog scheduling.

\myparagraph{Correctness Across the Stack}
At each stage---constraint, dynamic, and application---\qsnd guarantees correctness through types. The constraint semantics ensure well-formed expressions; the dynamic semantics preserve physical behavior; and the application semantics respect hardware constraints. These guarantees are enforced and verified in \rocq, the logical backend that ensures semantic alignment between typed programs and their denotations. In this way, \qsnd provides an end-to-end framework for writing, reasoning about, and executing quantum programs grounded in both physical intuition and formal correctness.

Our compilation correctness ensures that the compiled circuit, simulating a Hamiltonian $\hat{H}$ with time $r$, is within error bound  $\epsilon$ of the dynamic semantics $\eexp{r\hat{H}}$. For example, when simulating the above two-site Hubbard model, in the trotterization step, the error of compiling $\hat{H}$ in Eq. \ref{eq:ex} to $\eexp{r\hat{H}'}$ in Eq. \ref{eq:ex_final} has been calculated by our certified compiler.
In the standard Trotterization method, the error bound is $\pi^2/16$ following Eq. \ref{eq:trottererr_std}. The \qsnd compiler implements different compilation algorithms for different use cases. By selecting QDrift-based Trotterization, our compiler produces an error bound of $3\pi^2/16$ following Eq. \ref{eq:trottererr_qdrift}.
%$\frac{r^2}{2N}$ with $N$ being a selective step value.
Note that the error introduced by dropping the constant and local-field terms in Eq. \ref{eq:ex1} has been naturally included in the calculation.
We list our main compilation theorem (simplified version) below.

\begin{proposition}[Compilation Correctness]\label{thm:compile-good-simp}\rm 
Given a Hamiltonian $e$, typed as $\tjudge{\iota}{e}{\quan{F}{\hmx}{\iota}}$, and a time period $r$, the simulation $\eexp{\uapp{r}{e}}$ is correctly compiled to quantum circuit, via the pipeline $\quan{F}{\hmx}{\iota} \vdash (e, r) \gg U : \quan{F}{\umx}{\iota'}$, with a verified error bound  $\epsilon$ between $U$ and $\eexp{\uapp{r}{e}}$.
\end{proposition}

\section{Mathematical Background}
\label{sec:formalism}

This section reviews the mathematical foundations used throughout the paper, focusing on the representation of quantum states, operators, and their algebraic manipulation. These form the basis for our compilation and simulation framework, particularly as it relates to encoding Hamiltonians, managing particle types, and enforcing structural constraints during circuit generation.

\myparagraph{Quantum States and Superposition}
In quantum computing, the state of a system is described by a unit vector in a Hilbert space. We adopt Dirac (bra-ket) notation~\cite{Dirac1939} to express such vectors compactly. For a system of dimension $m$, a state in the Hilbert space $\Hilb_m$ %s written as $z_0\ket{0} + \cdots + %z_{m-1}\ket{m-1}$, where each $\ket{j}$ is a %computational basis vector such as $\ket{0} = %\%begin{psmallmatrix} 1 \\ 0 \\ \cdots %\%end{psmallmatrix}$ and $\ket{1} = %\begin{psmallmatrix} 0 \\ 1 \\ \cdots %%\end{psmallmatrix}$. 
is written as
$z_0\ket{0} + \cdots + z_{m-1}\ket{m-1}$,
where each $\ket{j}$ is a computational basis vector, such as
$\ket{0} = [1 \hspace{0.2em} 0 \hspace{0.1em} \cdots]^\mathsf{T}$ 
and
$\ket{1} = [0 \hspace{0.2em} 1 \hspace{0.1em} \cdots]^\mathsf{T}$,
and $z_j$ is a complex-number amplitude.
When more than one amplitude $z_j$ is nonzero, the state is said to be in a \emph{superposition}~\cite{mike-and-ike}, such as the state $\frac{1}{\sqrt{2}}(\ket{0} + \ket{1})$. Quantum systems can be composed via the tensor product. For instance, the two-qubit state $\ket{0} \otimes \ket{1}$ (also written as $\ket{0}\ket{1}$) corresponds to the vector 
$[0 \hspace{0.2em} 1 \hspace{0.2em} 0  \hspace{0.2em} 0]^\mathsf{T}$.
%$\begin{psmallmatrix} 0 \\ 1 \\ 0 \\ 0 \end{psmallmatrix}$. 
However, not all joint states can be decomposed into tensor products. Such states exhibit non-classical correlations between particles and are called \emph{entangled}. A canonical example is the Bell pair $\frac{1}{\sqrt{2}}(\ket{0}\ket{0} + \ket{1}\ket{1})$.

\myparagraph{Hamiltonian Simulation and Pauli Strings}
To simulate physical systems, quantum computers approximate the unitary time evolution $\eexp{r \hat{H}}$ generated by a Hamiltonian $\hat{H}$. Here, $\hat{H}$ is a Hermitian operator ($\hat{H} = \sdag{\hat{H}}$), and $r$ is a time or interaction parameter. A standard compilation strategy involves Trotterization, which rewrites $\hat{H}$ as a weighted sum of local terms: $\hat{H} = \sum_j z_j \hat{P}_j$, where each $\hat{P}_j$ is a tensor product of operators and $z_j$ are typically real scalars.\footnote{Any Hermitian Hamiltonian can be rewritten in a Pauli string form with real coefficients; see \Cref{sec:compilecanonical}.} In many cases, the $\hat{P}_j$ terms are drawn from the Pauli group, forming what are called \emph{Pauli strings}. For an $n$-site system, a Pauli string takes the form $\hat{P}_j = \pau_n \otimes \cdots \otimes \pau_1$, where each $\pau_k \in \{I, X, Y, Z\}$ acts on a qubit Hilbert space $\Hilb_2$.

\myparagraph{Second Quantization and Tensor Strings}
More expressive Hamiltonians (particularly those modeling fermionic or bosonic systems) are naturally represented using the formalism of second quantization. Instead of Pauli operators, the building blocks here are creation ($\sdag{a}$) and annihilation ($a$) operators. A typical Hamiltonian term is expressed as a composition $\overline{\alpha} = \alpha_1 \circ \cdots \circ \alpha_j$ with $\alpha_j \in \{a, \sdag{a}\}$ and $\circ$ denoting matrix multiplication. To represent spatial structure, such operator sequences are tensorized across sites, forming what we call \emph{tensor strings}: $\overline{\alpha}_n \otimes \cdots \otimes \overline{\alpha}_1$ \footnote{When it occurs as part of a tensor structure, the identity $I$ refers to a single site.}. For convenience, we write indexed notation such as $a(j)$ or $\sdag{a}(j)$ to insert an operator at site $j$ with identities elsewhere. For example, $a(j) = \Motimes_{k=0}^{j-1} I \otimes a \otimes \Motimes_{k=j+1}^{n-1} I$. 
%Compositions such as $\sapp{\sdag{a}(j)}{a(j+1)}$ insert both operators at adjacent sites.

\myparagraph{Fermionic Ladder Semantics}
To reason formally about tensor strings in second quantization, it helps to instantiate the algebra of creation and annihilation operators to a concrete setting. We focus on a simplified, two-dimensional model that behaves like a boson system but preserves key properties shared with fermions (particularly the notion of occupancy constraints). In this model, the ladder operators act on a two-level system with basis $\{\ket{0}, \ket{1}\}$, corresponding to the absence or presence of a particle. The annihilation operator is defined as $a = \begin{psmallmatrix} 0 & 1 \\ 0 & 0 \end{psmallmatrix}$ and the creation operator as $\sdag{a} = \begin{psmallmatrix} 0 & 0 \\ 1 & 0 \end{psmallmatrix}$. These matrices satisfy the following behavior 
$\sdag{a}\ket{0} = \ket{1}$, $\sdag{a}\ket{1} = \zero$, $a\ket{1} = \ket{0}$, and $a\ket{0} = \zero$. From these, we define two compound operations: $\mathbb{1} = \sapp{\sdag{a}}{a}$, which acts as a projector onto the occupied state $\ket{1}$, and $\mathbb{0} = \sapp{a}{\sdag{a}}$, which projects onto the unoccupied state $\ket{0}$. These projections are not unitary but are useful for expressing static constraints on particle presence within a state. Such projectors can be applied across multiple sites. For instance, the operator $\mathbb{1} \otimes \mathbb{0}$ preserves the state $\ket{1}\ket{0}$ and maps all other basis states to $\zero$. This mechanism enables local constraints to be embedded compositionally within a larger quantum system. For example, 
${\mapp{\sdag{a}}{a}} \otimes {\mapp{a}{\sdag{a}}} \ket{1}\ket{0} = \ket{1}\ket{0}$ but vanishes on other basis states. This semantics serve as the foundation for the constraint-based interpretation of programs in \qsnd, distinct from their dynamic (unitary) semantics.

\myparagraph{Pauli Operators as Ladder Composites}
The creation and annihilation operators not only express physical processes like particle addition and removal, but they also provide a basis for reconstructing familiar gate-level abstractions. In particular, the standard Pauli matrices can be algebraically derived as composites of these ladder operations: $I = \mathbb{0} + \mathbb{1}$, $X = \sdag{a} + a$, $Y = i a - i \sdag{a}$, and $Z = \mathbb{0} - \mathbb{1}$. 
%Here, $\mathbb{0} = a \circ \sdag{a}$ and $\mathbb{1} = \sdag{a} \circ a$ are projections onto the $\ket{0}$ and $\ket{1}$ subspaces, respectively. 
These definitions not only reproduce the expected matrix representations of the Pauli group but also clarify how Pauli strings (as used in quantum simulation) can be viewed as special cases of more general tensor strings involving creators and annihilators. This relationship reinforces the idea that the Pauli operators can be interpreted as expressive but coarse-grained encodings of particle-level dynamics. In the \qsnd framework, this makes it possible to derive constraint-oriented operations from lower-level semantics while retaining compatibility with gate-based formulations. However, it is important to emphasize that the resulting operators are interpreted as symbolic constraints, not necessarily as executable quantum gates.

\myparagraph{Typing and Compilation in \qsnd}
In \qsnd, quantum programs are enriched with a type system that classifies each site according to particle type. For example, $t^{\aleph}(2)$ denotes a two-dimensional fermionic state, while $t(2)$ is used for standard qubit or two-dimensional bosonic sites. These types are essential for interpreting the semantics of Hamiltonians correctly and guiding their compilation into circuits. Typing becomes particularly useful when translating systems across representations. For example, a two-site system typed as $t(m) \ttimes t(m)$ can be compiled to a qubit system by expressing each state $j \in [0, m)$ as a bitstring of length $N = \ulcorner \log(m) \urcorner$, yielding a system with $2N$ qubits of type $t(2)$. Such translations are captured by higher-order functions 
$\textcolor{spec}{(\Motimes^2 t(m) \to \Motimes^2 t(m)) \to (\Motimes^{2N} t(2) \to \Motimes^{2N} t(2))}$. This flexible typing and compilation infrastructure supports our goal of modeling diverse quantum systems within a unified symbolic and operational framework.

\section{Formalism}\label{sec:formal}

This section discusses \qsnd's syntax, type system, semantics, and equivalence relations on typed programs.

\begin{figure}[t]
%\vspace*{-0.5em}
{\small
  \[
  \begin{array}{c}
  {\hspace{-1em}
  \begin{array}{l@{\;\;\;}l@{\;}c@{\;}l @{\qquad} l@{\;\;\;}l@{\;}c@{\;}l @{\qquad} l@{\;\;\;}l@{\;}c@{\;}l @{\qquad} l@{\;\;\;}l@{\;\;}l}
  \text{Nat. Num} & m, n, j, k,g & \in & \mathbb{N}
&
\text{Real Number} & \theta, r & \in & {\mathbb{R}}
&
\text{Complex Number} & z & \in & {\mathbb{C}}
&
      \text{Variable} & x,y,f
\end{array}
}
\\[0.5em]
  \begin{array}{l@{\quad}l@{\;\;}c@{\;\;}l@{}c@{\;\;\;}l} 
\text{Single-Site Basis Vector in } \Hilb_{m} & \eta & ::= & {\displaystyle \ket{k}} && k\in[0,m) \\[0.5em]
n\text{-Site Basis-ket State} & w & ::= &  {\displaystyle \Motimes_{k=0}^{n-1} \eta_{k} }
%\text{Single Particle Type With }$n$\text{ Spins} & t(n,m) & ::= &  {\bigotimes^{n}\textcolor{spec}{\mathpzc{H}^{m}}}
%\\[1em]
%\text{Single Particle State} & v & ::= & {\displaystyle{\sum_j} z_j\eta_j }
\\[1.5em]
\text{Quantum State Type} & \iota & ::= & t^{[\aleph]}(m) &\mid& \iota \ttimes \iota
\\[0.5em]
\text{Quantum State} & \psi & ::= & {\displaystyle {\sum_j} z_j w_j  } &\mid& \zero
    \end{array}
        \end{array}
  \]
}
%\vspace*{-1.5em}
  \caption{\qsnd state structure for particles. A $t(m)$-typed basis vector $\ket{j}$ can have $j \in [0,m)$.  We abbreviate $z \cdot \Motimes_{k=0}^{0} \eta_{k} \equiv z\, \eta_0$ and $z_1\eta_1 \otimes z_2\eta_2\equiv (z_1 * z_2)\eta_1 \eta_2$.  $[\aleph]$ means an operational $\aleph$ flag.}
  \label{fig:data}
 % \vspace*{-1.2em}
\end{figure}

\subsection{Quantum Particle System State Representation}\label{sec:staterep}

In \qsnd, we model quantum systems as collections of \emph{sites}, where each site is associated with a local Hilbert space. In typical qubit-based systems, each site corresponds to a two-dimensional Hilbert space~$\Hilb_2$. More generally, however, sites may have different dimensionalities to account for distinct particle types (such as bosons and fermions), which require different local state spaces. For example, a three-site system might assign $\Hilb_3$ to the first site, $\Hilb_4$ to the second, and $\Hilb_5$ to the third, yielding a composite Hilbert space $\Hilb_{60} = \Hilb_3 \otimes \Hilb_4 \otimes \Hilb_5$. In this tensor product, each factor encodes the quantum state of an individual site. 

\Cref{fig:data} defines the quantum state syntax used in \qsnd, which formalizes such systems using \emph{computational basis vectors and basis kets}. A single-site basis vector of type $t^{[\aleph]}(m)$ \footnote{$m$ in the site type can be conceptually understood as the maximum number of particles that can occupy a site.} is written as the ket $\eta = \ket{k}$, where $k \in [0, m)$ identifies one of the $m$ local basis vectors. A composite $n$-site basis-ket state is then expressed as $w = \Motimes_{k=0}^{n-1} \eta_k$. For example, two sites of type $t(3)$ in basis vectors $\ket{2}$ and $\ket{1}$ compose into the joint state $\ket{2} \otimes \ket{1}$, which we abbreviate as $\ket{2}\ket{1}$. 

A single-site basis vector $\ket{k}$ of type $t(m)$ can be scaled by a complex amplitude $z$, forming a single-site state $z \ket{k}$. For an $n$-site system, we write basis-ket states as $w = \Motimes_{k=0}^{n-1} \eta_k$, and associate amplitudes component-wise. For example, the state $z_1 \eta_1 \otimes z_2 \eta_2$ can be abbreviated as $(z_1 * z_2) \eta_1 \eta_2$. In general, a quantum superposition state is written canonically as $\psi = \sum_j z_j w_j$, where each $w_j$ is a basis-ket term and $z_j \in \mathbb{C}$. We also include the zero vector $\zero$ to represent the null state, which plays a role in the semantics of second quantization.

\subsection{Program Syntax}\label{sec:syntax}

\Cref{fig:syntax} defines the abstract syntax of \qsnd programs. A constraint expression $e$ (or $\hat{H}$ when explicitly Hermitian) denotes a symbolic Hamiltonian imposing user-defined constraints to a system, while a unitary expression $U=\eexp{\hat{H}}$ represents compiled time evolution. The type constructor $\quan{F}{\zeta}{\iota}$ classifies program fragments over a typed state space $\iota$ by kind $\zeta$: plain ($\pmx$), Hermitian ($\hmx$) or unitary ($\umx$). Hermitian terms correspond to physical Hamiltonians, while unitary terms represent compiled simulations executable on quantum hardware.

\begin{figure}[t]
%\vspace*{-0.5em}
{\small
  \[\begin{array}{l@{\quad}l@{\;\;}c@{\;\;}l@{\;\;}} 
%\text{Single Particle Op} & \alpha & ::= & z a_{\sigma}
%\\
%\text{Arith Type} & \xi & ::= & \cn{nat} \mid \mathbb{R} \mid \mathbb{C} \mid \xi \to \xi
%\\
\text{DataType Kind} & \zeta & ::= & \pmx \mid \hmx \mid \umx
\\[0.2em]
\text{Quantum Operation Type} & \tau & ::= & \quan{F}{\zeta}{\iota}
%\\
%\text{Type} & \omega & ::= & \xi \mid \tau \mid \omega \to \omega
\\[0.2em]
       \text{Constraint Expression} & \hat{H},e & ::= &\stype{z\,{a}}{t^{[\aleph]}(m)} \mid \stype{I }{t^{[\aleph]}(m)} \mid \sdag{e} \mid e \otimes e \mid {e}+{e} \mid \mapp{e}{e}
       \\[0.2em]
       \text{Unitary Expression} & U & ::= & \eexp{\hat{H}} \mid \sapp{U}{U}
    \end{array}
  \]
}
 % \vspace*{-1.2em}
  \caption{\qsnd syntax. Both $e$ and $\hat{H}$ range over programs; we reserve $\hat{H}$ to Hamiltonian programs. \textcolor{spec}{Blue} type annotations are optional in the syntax. } 
  \label{fig:syntax}
 %   \vspace*{-1.5em}
\end{figure}

The term $\stype{z a}{t^{[\aleph]}(m)}$ represents an annihilation operator, scaled by an amplitude $z$, acting on a site of type $t^{[\aleph]}(m)$; the corresponding identity operator is  $\stype{I}{t^{[\aleph]}(m)}$. The type annotation $\aleph$ clarifies particle type (fermionic if $\aleph$ is set, bosonic otherwise), but it is not part of the term syntax by itself; types are inferred and checked by the system. The conjugate transpose operator $\sdag{e}$ denotes the adjoint of $e$. In particular, $\sdag{(\stype{z a}{t^{[\aleph]}(m)})}$ desugars to $\stype{\sdag{z} \sdag{a}}{t^{[\aleph]}(m)}$, representing a creator. In \qsnd, expressions $\stype{z \alpha}{t^{[\aleph]}(m)}$ and $\stype{z’ \alpha}{t^{[\aleph]}(m)}$ are syntactically equal if $z = z’$ for $\alpha \in \{a, \sdag{a}\}$.

Expressions include tensor compositions $e \otimes e$ to represent operations on distinct sites, additive combinations $e + e$ for linear combinations of terms (see \Cref{appx:linearsum}), and sequential compositions $\mapp{e}{e}$ for operator application. At the unitary level, $\eexp{\hat{H}}$ performs Hamiltonian simulation and $\sapp{U}{U}$ represents unitary composition. Both sequential compositions $\mapp{e}{e}$  and  $\sapp{U}{U}$  are matrix multiplications but they operate at distinct semantic levels, constraint and dynamic, respectively.

Programs often use indexed notation as syntactic sugar for tensor expressions. For a system with $n$ sites, an indexed operator such as $\sdag{a}(j)$ denotes the tensor product of identity operators at all sites except for site $j$, where the operator $\sdag{a}$ is applied. For instance, the Hamiltonian of the two-site Hubbard model includes these terms: $\textstyle -z_t \sum_j \left( \mapp{\sdag{a}(j)}{a(j+1)} + \mapp{\sdag{a}(j+1)}{a(j)} \right)$. Here, each indexed term implicitly expands into a tensor product over the full system. Note that the identity operator $I$ is type-sensitive; for a $t(2)$ typed site, $I = \mapp{\sdag{a}}{a} + \mapp{a}{\sdag{a}}$, capturing both occupied and unoccupied projections.

\subsection{Typing Rules}

We now introduce the \qsnd typing system. Every well-typed \qsnd program, whether a Hamiltonian $\hat{H}$ or a unitary $U$, is assigned a type of the form $\quan{F}{\zeta}{\iota}$, where $\zeta$ denotes the matrix kind and $\iota$ specifies the type of the quantum state to which the program applies. For instance, a single-site operation might have type $\quan{F}{\zeta}{t(m)}$, meaning that it acts on a bosonic state space of dimension $m$. Such types can be viewed as abstractly denoting maps $t(m) \to t(m)$ on single-site states. The kind $\zeta$ tracks the nature of the matrix represented by a program: ordinary ($\pmx$), Hermitian ($\hmx$), or unitary ($\umx$). A Hamiltonian program $\hat{H}$ must be Hermitian (i.e., $\hat{H} = \sdag{\hat{H}}$), and a unitary program $U$ must arise as a matrix exponential $U = \eexp{\hat{H}}$. This structure ensures that only unitary operators are considered executable on quantum hardware. To model multi-site systems, \qsnd includes tensor types of the form $\quan{F}{\zeta}{\iota_1 \ttimes \iota_2}$, used to separate sites (see \Cref{fig:data}). These types track how a program composes across sites with heterogeneous particle types.

\begin{figure*}[t]
  %\vspace*{-0.5em}
{\footnotesize

\begin{flushleft}\textcolor{blue}{Expression Equivalence:}\end{flushleft}

\begin{mathpar}
 %   \inferrule[]{}
   %             {\teq{\sapp{(e_1 \otimes e_2)}{(e_3 \otimes e_4)}}{(\sapp{e_1}{e_3}) \otimes (\sapp{e_2}{e_4})} }
    \inferrule[]{}
                {\vdash \teq{\sapp{(e_1 + e_2)}{e}}{(\sapp{e_1}{e}) + (\sapp{e_2}{e})} }
        
    \inferrule[]{}
               {\vdash \teq{\sapp{e}{(e_1+e_2)}}{(\sapp{e}{e_1}) + (\sapp{e}{e_2})} }

    \inferrule[]{}
                {\vdash \teq{{(e_1 + e_2)}\otimes{e}}{{e_1}\otimes{e} + {e_2}\otimes {e}} }
                        
    \inferrule[]{}
               { \vdash\teq{{e}\otimes{(e_1+e_2)}}{{e}\otimes{e_1} + {e}\otimes{e_2}} }
    
    \inferrule[]{}
        { \vdash\teq{\sdag{(e_1 \otimes e_2)}}{\sdag{e_1} \otimes \sdag{e_2}} }

    \inferrule[]{}
        { \vdash\teq{\sdag{(e_1 + e_2)}}{\sdag{e_1} + \sdag{e_2}} }
        
    \inferrule[]{}
        { \vdash\teq{ \sdag{(\sapp{e_1}{e_2})}}{\sapp{\sdag{e_2}}{\sdag{e_1}}}}          

    \inferrule[]{}
        {\vdash \teq{ \sdag{(\sdag{e})}}{e}}        

           \inferrule[]{}
        {\vdash \teq{ \sdag{I}}{I}}   

    \inferrule[]{}
               { \vdash\teq{\sapp{I}{e}}{e} }

    \inferrule[]{}
               { \vdash\teq{\sapp{e}{I}}{e} }

    \inferrule[E-Ten]{}{\iota^{\cn{b}}\vdash\teq{\sapp{(e_1 \otimes e_2)}{(e_3 \otimes e_4)}}{(\sapp{e_1}{e_3}) \otimes (\sapp{e_2}{e_4})} }

  \end{mathpar}

\ignore{
\begin{mathpar}
\hspace*{-1em}
    \inferrule[]{}
                { (\tau \otimes \tau_1) \to (\tau \otimes \tau_2) \equiv (\tau \to \tau ) \otimes (\tau_1 \to \tau_2)}

    \inferrule[]{}
                {\tau \otimes (\tau_1 \otimes \tau_2) \equiv (\tau \otimes \tau_1) \otimes \tau_2}

     % \inferrule[T-Var]{}
     %   {\tjudge{\Gamma}{x}{\Gamma(x)}}
        
    %\inferrule[T-Vec]{\forall j\in[0,n)\,.\,m_j < k}
    % {\tjudge{\Gamma}{\duala{j=0}{n}{m_j}{j}:t(n,k)}{(t(n,k))}}

  %  \inferrule[T-Lambda]{\tjudge{\Gamma[x\mapsto \omega]}{e}{\omega'}}
    %            {\tjudge{\Gamma}{\lambdae{x}{\omega}{e}}{\omega \to \omega'}}

  % \inferrule[T-Mu]{\tjudge{\Gamma[f\mapsto \omega \to \omega]}{e}{\omega}}
   %             {\tjudge{\Gamma}{\smu{f}{\omega}{e}}{\omega\to\omega}}
      % \inferrule[T-Dag]{\tjudge{\Gamma}{e}{t(\iota)}}
   %             {\tjudge{\Gamma}{\sdag{e}}{\sdag{t}(\iota)}}
   
      % \inferrule[T-Nor]{\tjudge{\Gamma}{e}{{t}^{[\dag]}(\iota)}}
    %            {\tjudge{\Gamma}{\cn{nor}(e)}{{t}^{[\dag]}(\iota)}}
                
    %\inferrule[T-Seq]{\tjudge{\Gamma}{e}{\quan{F}{\zeta}{\iota}}\quad\;\;\tjudge{\Gamma}{e'}{\quan{F}{\zeta'}{\iota}}}
    %            {\tjudge{\Gamma}{\sapp{e}{e'}}{\quan{F}{\zeta\sqcup \zeta'}{\iota}}}
                                
    %\inferrule[T-Mat]{\tjudge{\Gamma}{e}{\quan{F}{\zeta}{\iota}}\quad\;\;\tjudge{\Gamma}{e'}{\iota}}
    %            {\tjudge{\Gamma}{\sapp{e}{e'}}{\iota}}
                                
    %\inferrule[T-Inner]{\tjudge{\Gamma}{e}{\sdag{t}(\iota)}\\\tjudge{\Gamma}{e'}{t(\iota)}}
    %            {\tjudge{\Gamma}{\sapp{e}{e'}}{\mathbb{C}}}

   % \inferrule[T-Exp]{\tjudge{\Gamma}{e}{\quan{F}{\hmx}{\iota}}}
   %            {\tjudge{\Gamma}{\eexp{e}}{\quan{F}{\umx}{\iota}}}
                
    % \inferrule[T-Log]{\tjudge{\Gamma}{e}{\quan{F}{\umx}{\iota}}}
     %           {\tjudge{\Gamma}{\elog{e}}{\quan{F}{\hmx}{\iota}}}
           \inferrule[T-Par]{e \equiv e' \\ \iota \vdash e' \triangleright \tau}{ \iota \vdash e \triangleright \tau }
  \end{mathpar}
  }

  \vspace{0.5em}
  \begin{flushleft}\textcolor{blue}{Typing Rules:}\end{flushleft}
  
  \begin{mathpar}
    \inferrule[T-OpF]{}
                {\tjudge{t^{\aleph}(2)}{z\,{a}}{{\quan{F}{\pmx}{t^{\aleph}(2)}}}}
                
    \inferrule[T-OpB]{}
                {\tjudge{t(m)}{z\,{a}}{{\quan{F}{\pmx}{t(m)}}}}

    \inferrule[T-ID]{}
                {\tjudge{t^{[\aleph]}(m)}{I}{{\quan{F}{\hmx}{t^{[\aleph]}(m)}}}}

    \inferrule[T-Dag]{\tjudge{\iota}{e}{\quan{F}{\zeta}{\iota}}}
                {\tjudge{\iota}{\sdag{e}}{\quan{F}{\zeta}{\iota}}}
                
    \inferrule[T-Her]{\iota \vdash\teq{\sdag{e}}{e} \quad \tjudge{\iota}{e}{\quan{F}{\pmx}{\iota}}}
                {\tjudge{\iota}{e}{\quan{F}{\hmx}{\iota}}}

    \inferrule[T-Ten]{\tjudge{\iota}{e}{\quan{F}{\zeta}{\iota}}\\ \tjudge{\iota'}{e'}{\quan{F}{\zeta}{\iota'}}}
                {\tjudge{\iota \ttimes \iota'}{e \otimes e'}{\quan{F}{\zeta}{\iota \ttimes \iota'}}}

    \inferrule[T-App]{\tjudge{\iota}{e}{\quan{F}{\zeta}{\iota}}\\ \tjudge{\iota}{e'}{\quan{F}{\zeta'}{\iota}}}
                {\tjudge{\iota}{\sapp{e}{e'}}{\quan{F}{\zeta \sqcup \zeta'}{\iota}}}
                
    \inferrule[T-Sim]{\tjudge{\iota}{e}{\quan{F}{\hmx}{\iota}}}
                {\tjudge{\iota}{\eexp{e}}{\quan{F}{\umx}{\iota}}}
                
    \inferrule[T-Plus]{\tjudge{\iota}{e}{\quan{F}{\zeta}{\iota}}\\ \tjudge{\iota}{e'}{\quan{F}{\zeta}{\iota}}}
                {\tjudge{\iota}{e + e'}{\quan{F}{\zeta}{\iota}}}
                
    \inferrule[T-UApp]{\tjudge{\iota}{U}{\quan{F}{\umx}{\iota}}\\\tjudge{\iota}{U'}{\quan{F}{\umx}{\iota}}}
                {\tjudge{\iota}{\sapp{U}{U'}}{\quan{F}{\umx}{\iota}}}
  \end{mathpar}
}
 % \vspace*{-0.8em}
\caption{Typing and equational rules. $\equiv$ is the equivalence relation, assuming associativity and commutativity for $+$. $\hmx\sqcup \pmx=\pmx$ and $\umx\sqcup \pmx=\pmx$ and type error otherwise. $\iota^{\cn{b}}$ means $\iota$ has no fermionic type subterms ($t^{\aleph}(2)$).}
\label{fig:exp-proofsystem-1}
 % \vspace*{-1em}
\end{figure*}

The typing judgment $\iota \vdash e \triangleright \tau$ asserts that under input state type $\iota$, the expression $e$ has type $\tau$. \Cref{fig:exp-proofsystem-1} defines the typing rules and equivalence principles of \qsnd. The type system enforces two critical properties:
\begin{itemize}
\item Correct typing for each operator with respect to its local particle type.
\item Kind soundness: Hermitian and unitary kinds are preserved and checked via rewrite equivalences.
\end{itemize}
\qsnd models particle systems with various state types $t^{[\aleph]}(m)$, where the dimensionality $m$ corresponds to the Hilbert space of a site, and the flag $[\aleph]$ distinguishes bosons (flag off) from fermions (flag on). The type $\iota$ denotes the full system’s shape, constructed using tensor products $\ttimes$ across sites.

Rules \rulelab{T-OpF} and \rulelab{T-OpB} type a single annihilation operator applied to a fermionic or bosonic state, respectively. We assume fermions are modeled using a two-dimensional Hilbert space $t^{\aleph}(2)$~\footnote{Due to the anti-commutation property, a fermionic site can have at most one particle.}, while bosons may use arbitrary dimensions $t(m)$. Multi-site systems are typed using rule \rulelab{T-Ten}, which constructs a composite operation over a system with type $\iota = \iota_1 \ttimes \iota_2$. For example, a two-site system $\eta_1 \otimes \eta_2$ has type $t(m) \ttimes t(j)$, and a composite operator $a \otimes a$ is well-typed if its components are typed as $\quan{F}{\zeta}{t(m)}$ and $\quan{F}{\zeta}{t(j)}$. The resulting type is $\quan{F}{\zeta}{t(m) \ttimes t(j)}$, making particle types explicit and analyzable in the type system. The typing rules also ensure well-formedness of other operations. For instance, rule \rulelab{T-Plus} types additive combinations $e + e’$, which represent quantum superpositions. The operands must have matching types and kinds; notably, additive composition is not permitted for unitary kinds, as quantum hardware does not support coherent superpositions of distinct unitaries.

Since \qsnd programs represent matrices, the kind system tracks how program fragments compose. Every physically realizable (i.e., executable) program must ultimately have unitary kind $\umx$, and this must arise through simulation of a Hamiltonian of kind $\hmx$ using $\eexp{\hat{H}}$. In second quantization, basic operators such as $a$ and $\sdag{a}$ are not Hermitian. To promote an expression to kind $\hmx$, rule \rulelab{T-Her} requires that the expression be equivalent to its adjoint under the equational rules shown in \Cref{fig:exp-proofsystem-1}. For example, the Hubbard Hamiltonian in \Cref{sec:hubbard-example} includes the term $\sapp{\sdag{a}(j)}{ a(j+1)} + \sapp{\sdag{a}(j+1)}{ a(j)}$, which is Hermitian because it is equal to its own conjugate transpose under syntactic equivalence. Rule \rulelab{T-Sim} enforces that the matrix exponential $\eexp{e}$ is only allowed when $e$ is Hermitian, producing a unitary. Rule \rulelab{T-App} governs composition: the resulting kind is computed via a join operation $\zeta \sqcup \zeta’$, defined over the subtyping lattice $\hmx \sqsubseteq \pmx$ and $\umx \sqsubseteq \pmx$. This permits combining Hermitian or unitary operations with ordinary ones (defaulting to $\pmx$), but prohibits mixing $\hmx$ and $\umx$ directly, which would result in a kinding error.

\subsection{Equational Theory}\label{sec:eqthm}

\qsnd includes a system of equational rules that define algebraic equivalences between programs. These equivalences are interpreted modulo commutativity for the $+$ operation and associativity, which we assume implicitly throughout. The equational theory is primarily used to relate well-typed programs. We define this notion of equivalence formally below.

\begin{definition}[Well-Typed Equational Rewrites]\label{def:dag-form}\rm
We write $\iota \vdash e \equiv e’$ to mean that two well-typed \qsnd programs $e$ and $e’$ are equationally equivalent. That is, $e \equiv e’$ can be derived using the equational rules, and both programs have the same type $\tau$ under the same system shape $\iota$: $\tjudge{\iota}{e}{\tau}$ and $\tjudge{\iota}{e’}{\tau}$.
We denote $\vdash e \equiv e'$ as $\forall \iota\,.\,\iota \vdash e \equiv e'$.
\end{definition}

Most of the equational rules in \Cref{fig:exp-proofsystem-1} are valid for arbitrary types, with the exception of rule \rulelab{E-Ten}, which is only valid for bosons. Specifically, the subterm $\iota^{\cn{b}}$ must not contain any component typed with $t^{\alpha}(2)$, as tensor operations involving fermions are not independent due to anti-commutation. This restriction is illustrated by the equation 
$(I \otimes \sdag{a}) \circ (\sdag{a} \otimes I) = - (\sdag{a} \otimes I) \circ (I \otimes \sdag{a})$ which exhibits fermionic anti-commutation. If \rulelab{E-Ten} were allowed, the left-hand side could be rewritten as $(I \circ \sdag{a}) \otimes (\sdag{a} \circ I) = \sdag{a} \otimes \sdag{a}$ while the right-hand side would become $- (\sdag{a} \circ I) \otimes (I \circ \sdag{a}) = - \sdag{a} \otimes \sdag{a}$. Clearly, $\sdag{a} \otimes \sdag{a}$ and $-\sdag{a} \otimes \sdag{a}$ are not equal but \rulelab{E-Ten} would incorrectly equate them. This demonstrates the necessity of type sensitivity in our equational system.

Equational rewrites allow us to normalize programs into canonical forms, which serve as a foundation for both the semantics of \qsnd and the compilation procedure in \Cref{sec:compilation}. In particular, they allow adjoint operators ($\dag$) to be pushed inward to atomic terms and linear sum ($+$) outward to the program top level, simplifying the structure of expressions and enabling structural comparison. This motivates the following notion of normal form.

\begin{definition}[$\dag$-Canonical Form]\label{def:dag-form-canon}\rm
A well-typed program $e$ is in \emph{$\dag$-canonical form} if:

\begin{itemize}
\item $e$ is written as a linear combination $\sum_j \uapp{z_j}{e_j}$, where each $z_j$ is a scalar amplitude and each $e_j$ contains no further sum operators;
\item in every subterm of the form $\sdag{e_1}$, $e_1$ is of the form $z\,{a}:\iota$, meaning the adjoint is applied only to atomic particle operators.
\end{itemize}
\end{definition}

\begin{lemma}[Canonicalization via Rewriting]\label{lem:dag-can}\rm
For every Hamiltonian program $e$, there exists a $\dag$-canonical program $e’$ such that $\iota \vdash e \equiv e’$.
\end{lemma}

\begin{proof}
By structural induction on the typing derivation of $e$, using the equational rewrite rules in \Cref{fig:exp-proofsystem-1}. All cases have been fully mechanized.
\end{proof}

This canonical form underlies the Hamiltonian constraint semantics developed in the next section. In \Cref{sec:compilecanonical}, we introduce a second, more general canonicalization procedure (also defined via the equational theory of \Cref{fig:exp-proofsystem-1}) that applies to Hamiltonians expressed not in terms of second-quantized creation and annihilation operators, but using Pauli operators.

\begin{figure*}[t]
{\footnotesize
\begin{flushleft}\textcolor{blue}{Single Ket Semantics:}\end{flushleft}\vspace{0.3em}

{
\begin{center}
$\hspace{-0.5em}
\denote{\sdag{a}}(m)\ket{k} := \sqrt{k\splus 1}\ket{k\splus 1}\;\;\cn{if}\;k\neq m
\quad\;\;
\denote{\sdag{a}}(m)\ket{m} := \zero
\quad\;\;
\denote{a}(m)\ket{k} := \sqrt{k}\ket{k\sminus 1}\;\;\cn{if}\;k\neq 0
\quad\;\;
\denote{a}(m)\ket{0} := \zero
$
\\[0.7em]
\end{center}
}
}
{\small 
\begin{flushleft}\textcolor{blue}{Semantics Rules:}\end{flushleft}

       \begin{mathpar}    
        \inferrule[S-Top]{}
        { \denote{\tjudge{\iota}{e}{\tau}}_g \sum_j z_j w_j := \textcolor{purple}{\sum_j} z_j (\denote{\tjudge{\iota}{e}{\tau}}_g w_j)}
        
        \inferrule[S-OpB]{}
        { \denote{\tjudge{t(m)}{z a^{[\dag]}}{\tau}}_g\ket{j} := z (\denote{a^{[\dag]}}(m)\ket{j})}
        
        \inferrule[S-OpF]{}
        { \denote{\tjudge{t^{\aleph}(2)}{z a^{[\dag]}}{\tau}}_g\ket{j} := ((-1)^g * z) (\denote{a^{[\dag]}}(2)\ket{j})}

        \inferrule[S-ID]{}{ \denote{\tjudge{t^{[\aleph]}(m)}{I}{\tau}}_g\ket{j} := \ket{j}}
        
        \inferrule[S-Ten]{}
        { \denote{\inferrule[]{\tjudge{\iota}{e}{\quan{F}{\zeta}{\iota}}\\ \tjudge{\iota'}{e'}{\quan{F}{\zeta}{\iota'}}}{\tjudge{\iota \ttimes \iota'}{e \otimes e'}{\quan{F}{\zeta}{\iota \ttimes \iota'}}}}_g (w_1 \otimes w_2) := \denote{\tjudge{\iota}{e}{\quan{F}{\zeta}{\iota}}}_{g}{w_1}\,\textcolor{purple}{\otimes}\,\denote{\tjudge{\iota'}{e'}{\quan{F}{\zeta}{\iota'}}}_{(g + \textcolor{spec}{\funsa{S}{\iota}{w_1}})}{w_2} }
                
        \inferrule[S-App]{}
        { \denote{\inferrule[]{\tjudge{\iota}{e}{\quan{F}{\zeta}{\iota}}\\ \tjudge{\iota}{e'}{\quan{F}{\zeta'}{\iota}}}
                {\tjudge{\iota}{\sapp{e}{e'}}{\quan{F}{\zeta \sqcup \zeta'}{\iota}}}}_g := \denote{{\tjudge{\iota}{e}{\quan{F}{\zeta}{\iota}}}}_{g}(\denote{\tjudge{\iota}{e'}{\quan{F}{\zeta'}{\iota}}}_g) }

        \inferrule[S-Plus]{}
        { \denote{\inferrule[]{\tjudge{\iota}{e}{\quan{F}{\zeta}{\iota}}\\ \tjudge{\iota}{e'}{\quan{F}{\zeta}{\iota}}}
                {\tjudge{\iota}{e + e'}{\quan{F}{\zeta}{\iota}}}}_g :=
                \denote{\tjudge{\iota}{e}{\quan{F}{\zeta}{\iota}}}_{g}\, \textcolor{purple}{+}\, \denote{\tjudge{\iota}{e'}{\quan{F}{\zeta}{\iota}}}_g }
                
    \inferrule[S-Sim]{}
        {\denote{\inferrule[]{\tjudge{\iota}{e}{\quan{F}{\hmx}{\iota}}}
                {\tjudge{\iota}{\eexp{e}}{\quan{F}{\umx}{\iota}}}}_g := \textcolor{purple}{\sum_{n=0}^{\infty}}\frac{\textcolor{purple}{(}\sminus i \denote{\tjudge{\iota}{e}{\quan{F}{\hmx}{\iota}}}_g\textcolor{purple}{)^n}}{n!} }
                
        \inferrule[S-UApp]{}
        { \denote{\inferrule[]{\tjudge{\iota}{U}{\quan{F}{\umx}{\iota}}\\\tjudge{\iota}{U'}{\quan{F}{\umx}{\iota}}}
                {\tjudge{\iota}{\sapp{U}{U'}}{\quan{F}{\umx}{\iota}}}}_g := \denote{{\tjudge{\iota}{U}{\quan{F}{\umx}{\iota}}}}_{g}(\denote{\tjudge{\iota}{U'}{\quan{F}{\umx}{\iota}}}_g) }
  \end{mathpar}
}
{\small
$
{
\textcolor{spec}{\funsa{S}{t^{\aleph}(m)}{\ket{j}} = j}
\qquad
\textcolor{spec}
{\funsa{S}{t(m)}{\ket{j}} = 0}
\qquad
\textcolor{spec}
{\funsa{S}{\iota \otimes \iota'}{\eta \otimes \eta'} = \funsa{S}{\iota}{\eta}} + \textcolor{spec}{\funsa{S}{\iota'}{\eta'}}
}
$
}
  \vspace*{-0.5em}
\caption{\qsnd semantics for $\dag$-canonicalized program. $a^{[\dag]}$ means either a creator $\sdag{a}$ or an annihilator $a$. Terms inside $\denote{-}$ are QBlue operators. For terms outside $\denote{-}$, black $*$ and $+$ are scalar multiplication and addition, \textcolor{purple}{purple} items are quantum state vector operations, and $\denote{-}(\denote{-})$ is matrix multiplication.} %$e^n = \underbrace{\sapp{e}{\sapp{e}{\sapp{\cdots}}{e}}}}_{n}$.}
\label{fig:equiv}
  \vspace*{-0.5em}
\end{figure*}

\subsection{Semantics}\label{sec:sem}

The semantics of \qsnd is defined over $\dag$-canonicalized programs (see \Cref{def:dag-form}) as presented in \Cref{fig:equiv}. These programs describe finite $n$-site systems whose site states are typed by a shape $\iota$, with $\slen{\iota} = n$. To interpret such programs, we begin with a typed quantum state $\iota \vdash \psi$, defined in terms of the Dirac notation introduced in \Cref{fig:data}.

\begin{definition}[\qsnd Typed State]\label{def:typed-state}\rm
Let $\psi = \sum_j z_j w_j$ be a quantum state. We write $\iota \vdash \psi$ to mean that for every $j$, term $w_j$ is a well-typed basis ket under $\iota$: that is, $\iota \vdash w_j$:

\begin{itemize}
\item $\iota \ttimes \iota’ \vdash w \otimes w’$, if $\iota \vdash w$ and $\iota’ \vdash w’$;
\item $t(m) \vdash \ket{j}$, if $j < m$;
\item $t^{\aleph}(2) \vdash \ket{j}$, if $j \in  \{0,1\}$.
\end{itemize}
\end{definition}

Given a typed state, the semantics of \qsnd is defined denotationally as $\denote{\tjudge{\iota}{\Box}{\tau}}_g$, where $\Box$ is a well-typed $\dag$-canonicalized program (either a Hamiltonian $e$ or a unitary $U$), and $g \in \mathbb{N}$ is a context parameter. As explained in \Cref{sec:hubbard-example}, each \qsnd program simultaneously defines both constraint semantics (via Hamiltonians) and dynamic semantics (via unitary evolution).

The interpretation $\denote{\tjudge{\iota}{\Box}{\tau}}_0$ (with $g=0$) gives the core meaning of a Hamiltonian $\hat{H}$ as a Hermitian operator, and a unitary program $U$ as a unitary matrix, corresponding to the constraint and dynamic semantics, respectively. The context $g$ is used to account for fermionic anti-commutation effects and is initialized to $0$ during evaluation. More generally, $\denote{\tjudge{\iota}{\Box}{\tau}}_g$ defines a transformation on $\iota$-typed quantum states, i.e., a function $\iota \to \iota$.

\Cref{fig:equiv} presents the semantic rules using two complementary views. The rules \rulelab{S-Top}, \rulelab{S-OpB}, \rulelab{S-OpF}, and \rulelab{S-Ten} adopt the function view, describing how programs act on specific basis kets or superpositions thereof. The rules \rulelab{S-Plus}, \rulelab{S-App}, \rulelab{S-Sim}, and \rulelab{S-UApp} adopt the matrix view, interpreting programs as linear operators and expressing their composition explicitly. The context $g$ plays a role in enforcing global constraints, particularly related to fermionic parity signs, as described next.

\myparagraph{Constraint Semantics.}
The constraint-level semantics, $\denote{\tjudge{\iota}{e}{\tau}}_g$, begins with rule \rulelab{S-Top}, which lifts the semantics of $e$ pointwise over a superposition $\sum_j z_j w_j$. The rules \rulelab{S-OpB} and \rulelab{S-OpF} define the effect of bosonic and fermionic single-site operators, using the primitive definitions at the top of \Cref{fig:equiv}. Each operator is scaled by $z$, or $(-1)^g * z$ in the fermionic case.

In rule \rulelab{S-OpF}, the natural-numbered context $g$ enforces fermionic anti-commutation: applying a fermionic operator introduces a sign $(-1)^g$, where $g$ is the number of occupied fermion sites (i.e., $\ket{1}$ basis kets) before the target site. We assume a fixed total ordering of the sites in shape $\iota$. For example, applying an annihilation operator $a$ to the third site in the state $\ket{1}\ket{1}\ket{1}\ket{0}$ (where the first is bosonic and the second fermionic) yields $- \ket{1}\ket{1}\ket{0}\ket{0}$, due to one prior fermionic excitation. The sign correction is propagated via the accumulated context $g$ using the function $\mathpzc{S}$ defined in rule \rulelab{S-Ten}.
Rule \rulelab{S-Ten} defines the semantics of tensor composition. Given two subterms and a tensor product input state $w_1 \otimes w_2$, the semantics applies the two subterms recursively and combines the results using $\textcolor{purple}{\otimes}$. Note that zero vectors propagate through tensor products: $\zero \,\textcolor{purple}{\otimes}\, w = \zero$ and $w \,\textcolor{purple}{\otimes}\, \zero = \zero$. A zero-valued basis ket anywhere in the tensor string collapses the entire state.
Rules \rulelab{S-Plus} and \rulelab{S-App} define semantic interpretations of sum and application. A sum corresponds to a linear combination of operators, while a sequential composition corresponds to matrix multiplication or function composition.

\myparagraph{Dynamic Semantics.}
The rules \rulelab{S-Sim} and \rulelab{S-UApp} define the dynamic semantics of \qsnd. Rule \rulelab{S-Sim} interprets the expression $\eexp{e}$ as the matrix exponential, where $e$ is Hermitian. The exponential is defined via its standard power series expansion, which converges to a unitary operator when $e$ is Hermitian. We interpret this construct as the least fixed point of the matrix exponential series, approximated computationally as needed.
Rule \rulelab{S-UApp} interprets sequential composition of unitaries as standard matrix multiplication. Thus, dynamic programs in \qsnd are built by composing unitary evolutions, each derived from a Hamiltonian constraint via simulation.

\ignore{
\myparagraph{Permitting Equivalence Relations.}
The type system enables the equivalence relations, shown in \Cref{fig:equiv}, that support the semantic definitions in \Cref{fig:sem}.
The first line performs $\alpha$ conversions in a $\lambda$ and a $\mu$ term, provided that the substituted variable does not cause variable-capturing issues.

    %\inferrule[]{}
    %            { \teq{\sapp{(e_1 \otimes e_3)}{(e_2 \otimes e_4)}}{(\sapp{e_1}{e_2}) \otimes (\sapp{e_3}{e_4})} }

    %\inferrule[]{} maybe big mistakes to include the rules that are commented
    %            { \teq{}{\lambdae{x}{\tau_1 \otimes \tau_2 \to \tau_1 \otimes \tau_2}{e_1 \otimes e_2}}{ 
    %             (\lambdae{x}{\tau_1\to \tau_1}{e_1}) \otimes (\lambdae{x}{\tau_2\to \tau_2}{e_2})} }
    %\inferrule[]{}
    %            { \smu{f}{\tau_1 \otimes \tau_2}{e_1 \otimes e_2} \equiv 
    %             (\smu{f}{\tau_1}{e_1}) \otimes (\smu{x}{\tau_2}{e_2}) }
        %\inferrule[]{}
    %            { \lambdae{x}{\tau}{e_1 + e_2} \equiv (\lambdae{x}{\tau}{e_1}) + (\lambdae{x}{\tau}{e_2})}

    %\inferrule[]{}
    %            { \smu{x}{\tau}{e_1 + e_2} \equiv (\smu{x}{\tau}{e_1}) + (\smu{x}{\tau}{e_2})}
    
The second line relates application operations with different particle sites, separated by tensor operations ($\otimes$).
The equivalence relation implements distributed laws. The left-hand rule deals with the case when $e_1$ and $e_3$ are row vectors, and the right-hand rule happens when $e_1$ and $e_3$ are matrix operations ($\quan{F}{\zeta}{\kappa(n)}$).
The result pairs up $e_1$ and $e_2$ as well as $e_3$ and $e_4$, with the connection of a tensor operation.
In the two rules, the left-hand side can be rewritten to the right-hand side if $e_1$ and $e_2$ are typed to have the same dimensionality.

The third line deals with sum operations. When dealing with a sum of kets, we can move the operation $e$ to apply to the kets directly.
The right-hand-side rule distributes the superposition sum operations $e_1$ and $e_2$, so the rewrite results in a sum of two applications of applying $e_1$ and $e_2$ to the state $e_3$.
The fourth line deals with transpose operations ($\dag$), where the operation moves to the inner level when it is combined with a tensor ($\otimes$) and sum ($+$) operation.
When combined with an application, the order of the application is swapped if $e_1$ is typed as a row vector or a matrix operation, while the order is fixed if $e_1$ is typed as a number.
The last line defines the semantics of matrix exponential and logarithm operations as equivalence relations.
They are defined as a power series, which may converge to a fixed matrix operation. The exponential of a Hamiltonian should converge to a unitary, while the logarithm of a unitary should converge to a Hamiltonian. The \qsnd compilation implements the two operations as finite approximations.

\liyi{Maybe add some examples. }
}

\subsection{Metatheory}\label{sec:theorem}

We present the \qsnd theorems that connect the semantics and the type system. These results highlight how the type system enforces the anti-commutation property for fermions.

Lemma \ref{thm:equiv-sem} below connects Hamiltonians and unitaries, stating that a type guided equivalence Hamiltonian $\hat{H'}$ has the same quantum simulation behavior as the original Hamiltonian $\hat{H}$.

\begin{lemma}[Hamiltonian Equivalence Correctness]\label{thm:equiv-sem}\rm 
Given $\hat{H}$ and $\hat{H'}$, such that ${\iota}\vdash \teq{\hat{H}}{\hat{H'}}$, for every $r$, $\denote{\tjudge{\iota}{\eexp{\hat{H}}}{\quan{F}{\umx}{\iota}}}_g = \denote{\tjudge{\iota}{\eexp{\hat{H'}}}{\quan{F}{\umx}{\iota}}}_g$.
\end{lemma}

\begin{proof}
Fully mechanized proofs were done by induction on semantic rules, with the extra lemma that our equational theory is an equivalence relation (reflective, symmetric, and transitive).
\end{proof}

We prove the \qsnd type soundness theorem, which states that a well-typed \qsnd program, when applied to a well-typed quantum state, produces a result that is also well-typed.
%A well-typed state $\tau \vdash \psi$ means that each particle site described by $\psi$ respects the particle site types indicated in $\tau$.
%
%The result relies on the definition of a \emph{value} ($\nu$), which extends the value definition in $\lambda$ calculus (a variable $x$ or a lambda/mu abstraction) to include pure quantum states $\psi$, row vectors $\sdag{\psi}$, numbers subtyped to $\mathbb{C}$, and pure matrix operations ($\quan{F}{\xi}{\iota}$) consisting only of creators, annihilators, linear sums, and tensors.

\begin{theorem}[Type Soundness]\label{thm:type-progress-oqasm}\rm 
  If $\iota \vdash \Box \triangleright \quan{F}{\zeta}{\iota}$, with $\Box$ being $\hat{H}$ or $U$, and $\iota \vdash \psi$, then $\denote{\iota \vdash \Box \triangleright \quan{F}{\zeta}{\iota}}_g$ well-defines a $\zeta$-kind matrix, and for all $\psi'$, $\denote{\iota \vdash \Box \triangleright \tau}_g (\psi) = \psi'$, then $\iota \vdash \psi'$.
\end{theorem}

\begin{proof}
Fully mechanized proofs were done by induction on the \qsnd type system.
\end{proof}

One utility of the type soundness theorem is to guarantee the anti-commutation property of fermion operations in \qsnd, stated as follows.

\begin{corollary}[Guarantee Anti-Commutation]\label{thm:type-progress-oqasm}\rm 
  If $\bigttimes{n} t^{\aleph}(2) \vdash \hat{H} \triangleright \quan{F}{\hmx}{\bigttimes{n} t^{\aleph}(2)}$, $\hat{H}$'s interpretation $\denote{\bigttimes{n} t^{\aleph}(2) \vdash \hat{H} \triangleright \quan{F}{\hmx}{\bigttimes{n} t^{\aleph}(2)}}_0$ preserves fermions' anti-commutation property.
\end{corollary}

\begin{proof}
Fully mechanized proofs were done by induction on the \qsnd semantics for well-typed programs with the assumption that all sites are fermions ($t^{\aleph}(2)$).
\end{proof}

%Based on \Cref{def:hamsim}, the Hamiltonian simulation, $\eexp{{\hat{H}}{t}}$, of a Hermitian Hamiltonian $\hat{H}$ produces a unitary.
%Our type soundness theorem in \Cref{thm:type-progress-oqasm} shows that every Hermitian type ($\hmx$) program $e$ can be turned into a unitary through the matrix exponent operation ($\eexp{e}$).
%Additionally, Deutsch \textit{et al.} \cite{Deutsch1989QuantumCN,Deutsch1995} showed that any unitary matrix can be decomposed to elementary quantum gates, executable in quantum computers.
%Combined with our type soundness theorem, this fact results in the following compilation existence corollary.

The type soundness theorem indicates the existence of a compilation to quantum circuits and is used in our compilation correctness in \Cref{sec:compilation}.

\ignore{
{\small
\[
\label{equ1}
\sapp{(z_t \sdag{a} \otimes I)}{(I \otimes a)} + \sapp{(z_t I \otimes \sdag{a})}{(a \otimes I)}
\]
}
}

                   % formal developement

%\liyi{This section describes the compilation from QBlue to quantum program for Hamiltonian simulation. Three steps: JW/BK, Gadgetization, and Trotterization. Each with an optimized algorithm.}

\begin{figure}[t]
\includegraphics[width=\textwidth]{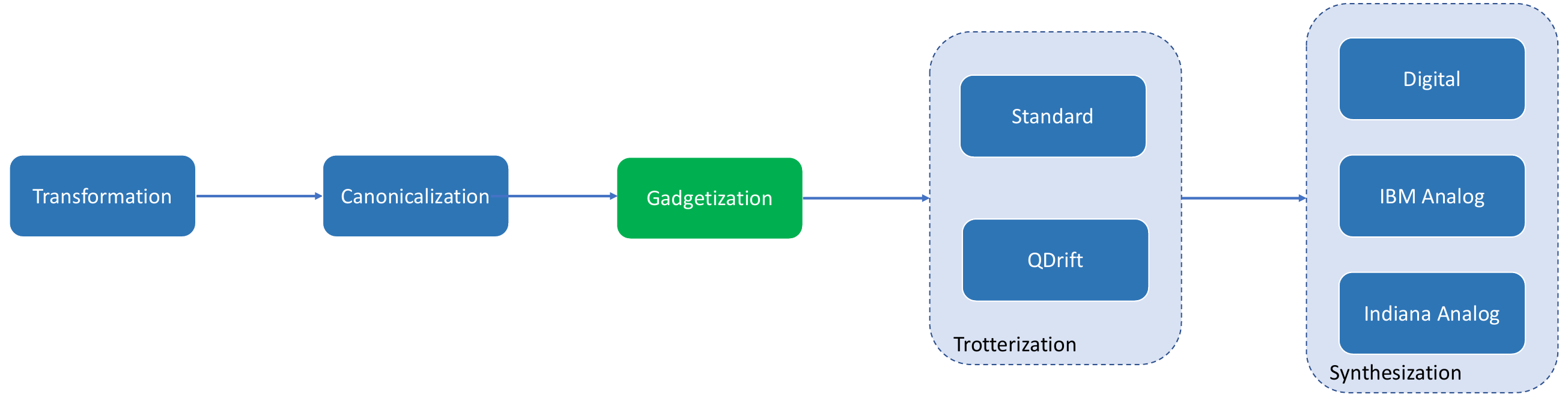}
  \caption{Hamiltonian Simulation Compilation Flow}
\label{fig:compilationprocess}
\end{figure}

\section{The QBlue Compilation Procedure}\label{sec:compilation}

This section presents the \qsnd compilation steps (\Cref{fig:compilationprocess}) along with correctness proofs for each step. The compiler transforms a Hamiltonian simulation expression $\eexp{\uapp{r}{e}}$, where $e$ is a \qsnd Hamiltonian program and $r$ is a time parameter, into a quantum circuit executable on quantum hardware. We model compilation as a judgment of the form $\quan{F}{\hmx}{\iota} \vdash (e, r) \gg U : \quan{F}{\umx}{\iota’}$, which extends the typing judgment $\iota \vdash e \triangleright \quan{F}{\hmx}{\iota}$ introduced in \Cref{sec:syntax}. The \qsnd compiler supports multiple compilation paths, allowing users to choose among different strategies. The perturbative gadget transformation (\Cref{sec:gadget}) is optional. We also provide two Trotterization algorithms and three different synthesis methods. The main theorem in \Cref{sec:mainthm} assumes a compilation path that omits the perturbative gadget, but it is general enough to handle both Trotterization variants. The correctness of compilation paths that include the perturbative gadget is addressed in \Cref{appx:gadget}. We begin by describing the individual compilation steps and then state the theorem covering the entire pipeline.

\ignore{
\subsection{Compilation to $\lambda$-Calculus}\label{sec:theory}

We show the encoding of \qsnd as a simply typed $\lambda$-calculus as our theory foundation.
It is well-known \cite{6278971} that one can encode natural numbers, pairs, lists, and conditionals by using Church or Mogensen–Scott encoding,
as well as approximate the computation of real and complex numbers by encoding them as some data structures.
With the assumption encoding for the few operations above, we show that every \qsnd operation can be encoded through Church encoding based on a simply typed $\lambda$-calculus.

We use $\textcolor{spec}{\cmsg{z}}$ to mean the church encoding of approximating complex number $z$, and $\textcolor{spec}{\cmsg{[,]}}$ refers to the encoding of a length $n$ array.
For representing a quantum state $v$, we utilize its vector representation, i.e., an $t(n,m)$ typed Hilbert space ($\hspc{n}{m}$) is encoded as a $2^{k}$ dimensional complex vector with $k=m*n$,
which is encoded as a $2^{k}$ length array $\textcolor{spec}{\cmsg{[\cmsg{z_1},...,\cmsg{z_{2^{k}}}]}}$;
$\zero$ can be interpreted as a zero vector.
The sum of two vectors ($\textcolor{spec}{\cmsg{+}}$) is encoded as the vector sum of complex numbers, as:

{\small
\begin{center}
$\textcolor{spec}{\cmsg{[\cmsg{z_1},...,\cmsg{z_{2^{k}}}]}\,\cmsg{+}\,\cmsg{[\cmsg{z'_1},...,\cmsg{z'_{2^{k}}}]}=\cmsg{[\cmsg{z_1\splus z'_1},...,\cmsg{z'_{2^{k}}\splus z'_{2^{k}}}]}}$
\end{center}
}

The interpretation of a tensor product of two vectors ($v \otimes v'$) can be interpreted as a Cartesian product of entries in the two complex arrays.
A \qsnd matrix operation applying to a $t(n,m)$ typed Hilbert space, having type $\quan{F}{\zeta}{t(n,m)}$,
can be encoded as a $2^{k} \times 2^{k}$ matrix with $k=m*n$, i.e., a two dimensional array $\textcolor{spec}{\cmsg{[[\cmsg{z^1_1},...,\cmsg{z^1_{2^{k}}}],...,[\cmsg{z^{2^{k}}_1},...,\cmsg{z^{2^{k}}_{2^{k}}}]]}}$.
The application of applying a $\quan{F}{\zeta}{\iota}$ typed operation to a vector $t(\iota)$ can be interpreted as a matrix multiplication, which can be encoded as:

{\small
\begin{center}
$\textcolor{spec}{\cmsg{[[\cmsg{z^1_1},...,\cmsg{z^1_{2^{k}}}],...,[\cmsg{z^{2^{k}}_1},...,\cmsg{z^{2^{k}}_{2^{k}}}]]}\,\cmsg{\times}\,\cmsg{[\cmsg{z'_1},...,\cmsg{z'_{2^{k}}}]}=
\cmsg{[\cmsg{z^1_1} \cmsg{z'_{1}}\splus ... \splus \cmsg{z^1_{2^{k}}} \cmsg{z'_{2^{k}}},..., \cmsg{z^{2^{k}}_1} \cmsg{z'_{1}}\splus ... \splus \cmsg{z^{2^{k}}_{2^{k}}} \cmsg{z'_{2^{k}}}]}}$
\end{center}
}

For the encoding of a $\quan{F}{\zeta}{\iota}$ typed operation, we can first rewrite the operation type to be $\quan{F}{\zeta}{t(n_1)} \otimes ... \otimes \quan{F}{\zeta}{t(n_m)}$, if $\iota = t(n_1) \otimes ... \otimes t(n_m)$. A tensor product ($\textcolor{spec}{\cmsg{\otimes}}$) can be interpreted as a $\lambda$-calculus term similar to the matrix operations above.
To interpret a transpose operation $\sdag{e}$, we first rewrite $\sdag{e}$ to its equivalence form through the equivalence relations in \Cref{fig:equiv}.
After rewrites, the only application is a row vector measurement operation, i.e., applying a row vector ($\sdag{v}$) to a vector ($v'$), as $\sapp{\sdag{v}}{v'}$,
which essentially performs an inner product as $\sdag{v}\cdot v'$, and an inner product can be interpreted as a $\lambda$-calculus term as a similar manner above.

To interpret $\eexp{e}$ and $\elog{e}$, we expand the terms by using the least fixed-point equations in \Cref{fig:equiv}, with possible approximation by finite truncation.
Each iteration includes scalar and matrix multiplications, whose interpretations are described above.
An $n$-step expansion for $\eexp{e}$ can be encoded as a fixed-point operator below.

{\small
\begin{center}
$\smu{f}{\quan{F}{u}{\iota}}{\lambdae{j}{\cn{nat}}{\sifb{j=0}{I}{\frac{-i}{j!}e^j}}}
$
\end{center}
}

Depending on the dimension of $\iota$, $I$ is an identity matrix having the same dimension as $\iota$, $e^k$ is a $k$ series of matrix multiplications,
and the equality $j=0$ checks if a natural number $j$ is zero, definable in a church encoding.
A matrix logarithm ($\elog{e}$) can be defined similarly.
Essentially, the encoding emits a compilation from \qsnd to $\lambda$-calculus. We show the soundness theorem below.

\begin{theorem}[Approximation Soundness]\label{thm:compile-sound}\rm 
The approximation of any \qsnd operation is sound with respect to simply typed $\lambda$-calculus plus a fixed-point combinator.
\end{theorem}

Since \qsnd is a model of second quantization, which is general enough in quantum mechanics to define computational behaviors of quantum systems, the soundness theorem indicates that the computation of any quantum system is approximable in $\lambda$-calculus.
}

\subsection{Particle Transformation}\label{sec:transformation}

\begin{figure*}[t]
 % \vspace*{-0.8em}
{\footnotesize
  \begin{mathpar}
      \inferrule[PT-Plus]{{\iota} \vdash_n e_1 \gg_{\cn{t}} \hat{H}_1 \triangleright {\iota'} \\{\iota} \vdash_n e_2 \gg_{\cn{t}} \hat{H}_2 \triangleright {\iota'} }{ {\iota} \vdash_n e_1+ e_2 \gg_{\cn{t}} \hat{H}_1 + \hat{H}_2 \triangleright {\iota'}}
    
        \inferrule[PT-App]{{\iota} \vdash_n e_1 \gg_{\cn{t}} \hat{H}_1 \triangleright {\iota'} \\{\iota} \vdash_{n} e_2 \gg_{\cn{t}} \hat{H}_2 \triangleright {\iota'} }{ {\iota} \vdash_n e_1 \circ e_2 \gg_{\cn{t}} \hat{H}_1 \circ \hat{H}_2 \triangleright {\iota'}}
        
    \inferrule[PT-Ten]{\slen{\iota_1}^{\aleph} = m \\{\iota_1} \vdash_n e_1 \gg_{\cn{t}} \hat{H}_1 \triangleright{\iota'_1} \\{\iota_2} \vdash_{n+m} e_2 \gg_{\cn{t}} \hat{H}_2\triangleright {\iota'_2} }{ {{\iota_1 \textcolor{spec}{\otimes} \iota_2}} \vdash_n e_1 \otimes e_2 \gg_{\cn{t}} \hat{H}_1 \otimes \hat{H}_2\triangleright {{\iota'_1 \textcolor{spec}{\otimes} \iota'_2}}}
                
    \inferrule[PT-FeLow]{}
                {{t^{\aleph}(2)} \vdash_n \iseq{z}{{a}} \gg_{\cn{t}} \iseq{z}{ (\bigotimes^n Z) }\otimes (X + iY) \triangleright {t(2)} }

    \inferrule[PT-FeUp]{}
                {{t^{\aleph}(2)} \vdash_n \sdag{(\iseq{z}{{a})}} \gg_{\cn{t}} \iseq{\sdag{z}}{ (\bigotimes^n Z)} \otimes (X - iY) \triangleright {t(2)} }
               
    \inferrule[PT-BoLow]{B(j) = \bigotimes_{k=0}^{\ulcorner \log(m) \urcorner\sminus 1}\mathpzc{A}^k(j,{j\sminus 1})}
                {{t(m)} \vdash_n \iseq{z}{{a}} \gg_{\cn{t}} \iseq{z}{ \sum_{j=1}^{m\sminus 1} \sqrt{j} \,B(j)} \triangleright {\textcolor{spec}{\bigotimes^{\ulcorner \log(m) \urcorner}} t(2)} }

    \inferrule[PT-BoUp]{B'(j) = \bigotimes_{k=0}^{\ulcorner \log(m) \urcorner\sminus 1}\mathpzc{A}^k({j},{j\splus 1})}
                {{t(m)} \vdash_n \sdag{(\iseq{z}{{a})}} \gg_{\cn{t}} \sdag{z}\, \sum_{j=0}^{m\sminus 2} \sqrt{j\splus 1} \, B'(j) \triangleright {\textcolor{spec}{\bigotimes^{\ulcorner \log(m) \urcorner}} t(2)} }
                
       \inferrule[PTS-Fe]{\iota \vdash w \gg_{\cn{t}} w' \triangleright \iota' }
                {t^{\aleph}(2) \textcolor{spec}{\otimes} \iota \vdash z \,\eta \otimes w \gg_{\cn{t}}  z \,\eta \otimes w' \triangleright t(2) \textcolor{spec}{\otimes} \iota' }

       \inferrule[PTS-Sum]{\forall j . \iota \vdash w_j \gg_{\cn{t}} w'_j \triangleright \iota'}
                {\iota \vdash \sum_j w_j \gg_{\cn{t}} \sum_j w'_j \triangleright \iota'}
                
       \inferrule[PTS-Bo]{\iota \vdash w \gg_{\cn{t}} w' \triangleright \iota' }
                {t(m) \textcolor{spec}{\otimes} \iota \vdash z \,\ket{j} \otimes w \gg_{\cn{t}} z \,\ket{{j}}_{{\ulcorner \log(m) \urcorner}} \otimes w' \triangleright \textcolor{spec}{\bigotimes^{\ulcorner \log(m) \urcorner}} t(2) \textcolor{spec}{\otimes} \iota'}
  \end{mathpar}
}
{\footnotesize
$
\begin{array}{c}
\textcolor{spec}{\ket{j}_n} = \textcolor{spec}{\ket{\tob{j}[0]}\ket{\tob{j}[1]}...\ket{\tob{j}[n\sminus 1]}}
\\[0.2em]
\begin{array}{l@{\;}c@{\;}l@{\;\;}c@{\;}l@{\qquad\quad}l@{\;}c@{\;}l@{\;\;}c@{\;}l}
\textcolor{spec}{\mathpzc{A}^k(j,m)}&=&\textcolor{spec}{\frac{1}{2}(X - i Y)}&\cn{if}&\textcolor{spec}{\tob{j}[k]=0 \wedge \tob{m}[k]=1}
&
\textcolor{spec}{\mathpzc{A}^k(j,m)}&=&\textcolor{spec}{\mathbb{1}}&\cn{if}&\textcolor{spec}{\tob{j}[k]=\tob{m}[k]=1} 
\\[0.2em]
\textcolor{spec}{\mathpzc{A}^k(j,m)}&=&\textcolor{spec}{\frac{1}{2}(X + i Y)} & \cn{if} & \textcolor{spec}{\tob{j}[k]=1 \wedge \tob{m}[k]=0}
&
\textcolor{spec}{\mathpzc{A}^k(j,m)}&=&\textcolor{spec}{\mathbb{0}}&\cn{if}&\textcolor{spec}{\tob{j}[k]=\tob{m}[k]=0} 
\end{array}
\end{array}
$
}
\vspace*{-0.3em}
\caption{Particle transformation rules. $\slen{\iota}^{\aleph}$ collects the \# of fermions ($t^{\aleph}(2)$) in $\iota$. $\tob{j}$ converts $j$ to bitstring. }
\label{fig:exp-transform}
\vspace*{-0.8em}
\end{figure*}

The first compilation step is particle transformation, transforming the second quantization program in \Cref{fig:syntax} to a $t(2)$-typed qubit-based Hamiltonian program $\hat{H}$ in a Pauli-string-based form. For clarity, we list below the Pauli string program syntax, as a syntactic sugar extension of the \qsnd syntax in \Cref{fig:syntax}, with all the sites being $t(2)$-typed.

{
\begin{center}
$
\mathpzc{P} ::= X \mid Y \mid Z \mid I\qquad\qquad
\hat{H} ::= \iseq{z}{\mathpzc{P}} \mid \hat{H} + \hat{H} \mid \hat{H} \otimes \hat{H} 
$
\end{center}
}

The particle transformation compilation can be expressed as the flow-sensitive judgment $\iota \vdash_n e \gg_{\cn{t}} \hat{H} \triangleright \iota'$, as an extension of the \qsnd typing rules. Here, $e$ is the input $\hmx$-kind Hamiltonian program in the second quantization form with the type $: \quan{F}{\hmx}{\iota}$, $\hat{H}$ is the Pauli string Hamiltonian (syntax defined above) with the type $\quan{F}{\hmx}{\iota'}$, and $n$ is the flow sensitive context. A Hamiltonian program describes constraints for a finite number of quantum particles. Without losing generality, for any program $e$, we assume that $e$ can be canonicalized as its operations ordered as a list structure describing a one dimensional list of applications to the sites, and the flag $n$ indicates that the current compilation is processing the $n\splus 1$-th element in the list as $n$ is the previous site adjacent to the $n\splus 1$-th site.

\Cref{fig:exp-transform} shows the transformation rules. The transformation for fermion particles is based on Jordan-Wigner transformation \cite{JordanAboutTP}, while the boson transformation is based on Holstein–Primakoff transformation \cite{qua.25176}. We show an alternative transformation in \Cref{sec:jw-trans}. 
Rules \rulelab{PT-Plus} and \rulelab{PT-App} deal with linear sum and sequence operations; such operation transformations depend on the subterm transformations.
In transforming tensor operations in \rulelab{PT-Ten}, we ensure that $e_1$'s type is $\quan{F}{\hmx}{\iota_1}$. Without losing generality, we also find the site application length (the site number of applied operations in $e_1$) to be $m$, as $\slen{\iota_1} = m$. Recall that we assume the sites of $e_1$ and $e_2$ applied to are canonicalized as an ordered list, so $e_2$'s flow-sensitive context is set to $n+m$, as $n+m+1$ represents the first site position $e_2$ applies to.

Rules \rulelab{PT-FeLow} and \rulelab{PT-FeUp} describe the transformation for fermion annihilators and creators.
As we mentioned, fermions are always $\textcolor{spec}{t^{\aleph}(2)}$ typed. To compile fermion particle simulations, we apply $Z$ operations to all the sites before the $n$-th site.
For the current site ($n+1$), we apply a lowering ($X + iY$) or a raising ($X - iY$) operation.
The resulting type is $\textcolor{spec}{\quan{F}{\hmx}{t(2)}}$, representing that the compiled result is a single qubit operation, i.e., a qubit is enough to describe a fermion site state.

Rules \rulelab{PT-BoLow} and \rulelab{PT-BoUp} describe the binary compilation scheme to transform boson annihilators and creators.
In compiling the particle simulation, we use $\ulcorner \log(m) \urcorner$ qubits to represent a $\textcolor{spec}{t(m)}$ typed boson state, as the compilation results in a $\textcolor{spec}{\quan{F}{\hmx}{\otimes^{\ulcorner \log(m) \urcorner} t(2)}}$ typed operation. 
Generally, applying a creator $\sdag{a}$ to $\ket{j}$ results in $\ket{j\splus 1}$ and applying an annihilator $a$ to $\ket{j}$ results in $\ket{j\sminus 1}$.
The compiled system produces a sum of all $j$-th Pauli strings turning $\ket{j}$ to $\ket{j\splus 1}$ (or $\ket{j\sminus 1}$ for $a$).
For each $j$, we examine $j$'s binary representation,  compare $j$ and $j\splus 1$ bitwise (or $j\sminus 1$ for $a$), and then apply a proper Pauli operation ($\mathpzc{A}^k(j,m)$).
%The compilation scheme is based on having exactly a single $\ket{1}$ state in am $m$ basis-ket, e.g., the $\ket{j}$ (a $t(m)$ ket with $j < m$) state is representing as $\otimes^{j\sminus 1} \ket{0} \otimes \ket{1} \otimes \otimes^{m\sminus j} \ket{0}$ (rule \rulelab{PT-SBo}). A $\quan{F}{\hmx}{t(m)}$ boson operation is compiled to $m+1$ $\quan{F}{\hmx}{t(2)}$ typed qubit operations, as the last qubit represents the overflow bit indicating that an out-of-bound boson state is reached.
%The compiled result is a linear sum of different cases of removing the $\ket{1}$ particle and putting it somewhere else.
For example, if we compiled a $t(5)$ boson state $\ket{3}$ as $\ket{1}\ket{1}\ket{0}\ket{0}$ (LSB form), applying a creator to the state results in $\ket{0}\ket{0}\ket{1}\ket{0}$ ($\ket{4}$);
Such behavior is captured by $B'(3)$ in \rulelab{PT-BoUp}, where we use a lowering operation ($X + iY$) to remove the two $\ket{1}$ bits and apply a raising operation on the third qubit.
Note that in the case both bits ($\tob{j}[k]$ and $\tob{m}[k]$) are the same, we need to apply $\mathbb{1}$ or $\mathbb{0}$ operations, corresponding to the Pauli operations $\frac{1}{2}(I-Z)$ and $\frac{1}{2}(I+Z)$, respectively.

Rules \rulelab{PTS-Sum}, \rulelab{PTS-Fe}, and \rulelab{PTS-Bo} describe the transformation of states for a particle system state to a state in the qubit system.
The transformation judgment is $\iota \vdash \psi \gg_{\cn{t}} \psi' \triangleright \iota'$, where we transform the typed $\iota$ state $\psi$ to type $\iota'$ state $\psi'$.
Rules \rulelab{PTS-Fe} and \rulelab{PTS-Bo} transform fermion and boson states, described above.

\begin{lemma}[Particle Transformation Correctness]\label{thm:particle-trans}\rm 
Given ${\iota} \vdash_0 e \gg_{\cn{t}} \hat{H} \triangleright {\iota'}$, for any $\psi_1$, such that $\iota \vdash \psi_1$, with the compiled state $\psi'_1$, denoted as $\iota \vdash \psi_1\gg_{\cn{t}} \psi'_1\triangleright \iota'$, $\psi_2=\denote{\tjudge{\iota}{e}{\quan{F}{\hmx}{\iota}}}_0(\psi_1)$ and $\psi'_2=\denote{\tjudge{\iota}{{\hat{H}}}{\quan{F}{\hmx}{\iota}}}_0(\psi'_1)$; thus, $\iota \vdash \psi_2 \gg_{\cn{t}} \psi'_2\triangleright \iota'$.
\end{lemma}

\begin{proof}
Fully mechanized proofs were done by induction on the particle transformation rules in \Cref{fig:exp-transform}.
\end{proof}

\myparagraph{Example Transformation.}\label{sec:bosonsystem}
We show the compilation of a bosonic particle system to quantum computers (not the \qsnd boson-like system).
Bosons are a special particle and typically require an infinite-dimensional Hilbert space to describe a particle site, e.g., a boson is typed $t(\infty)$.
A typical way of approximating a bosonic system \cite{qua.25176} utilizes $2^{n}$ dimensional Hilbert space, i.e., we use a length $n$ basis vector to track $2^n-2$ basis-vector state from $\ket{0}$ to $\ket{2^n-2}$.
The extra $n\splus 1$-th basis-vector $\ket{2^n-1}$ acts as an overflow bit, indicating if a maximal orbital is reached.

We now apply the transformation to a two-site bosonic system, with $m=4$. We use the Hubbard model $\hat{H}$ in \Cref{sec:hubbard-example} to describe the superconducting behaviors of bosons. We focus on transforming the sub-expression $\sapp{\sdag{a}(0)\textcolor{spec}{:t(4)}}{a(1)\textcolor{spec}{:t(4)}}+\sapp{a(0)\textcolor{spec}{:t(4)}}{\sdag{a}(1)\textcolor{spec}{:t(4)}}$ to Pauli-string-based Hamiltonian.
It is enough to show the transformation for a single creator and annihilator applying to different $\ket{j}$ basis-vectors, e.g., examining the cases of applying $\sdag{a}$ to $\ket{0}$ getting $\ket{1}$ and applying $a$ to $\ket{2}$ outputting $\ket{1}$.

\vspace{0.3em}
{\small
$
\begin{array}{|c|c|c|c|c|}
\hline
\multirow{3}{*}{$\sdag{a}$}
&
 \ket{j} \text{ transition} & \ket{0} \to \ket{1} & \ket{1} \to \ket{2} & \ket{2} \to \ket{3} 
\\[0.2em]
\cline{2-5}
&&&&\\[-1.1em]
&
 \text{Binary transition} & \ket{0}\ket{0} \to \ket{1}\ket{0} & \ket{1}\ket{0} \to \ket{0}\ket{1} & \ket{0}\ket{1} \to \ket{1}\ket{1} 
 \\[0.2em]
\cline{2-5}
&&&&\\[-1em]
& 
  \text{Pauli String} & \frac{1}{2}(X - iY) \otimes \frac{1}{2}(I + Z) & \frac{1}{2}(X + iY) \otimes \frac{1}{2}(X- iY)  & \frac{1}{2}(X - iY) \otimes \frac{1}{2}(I - Z)
  \\[0.2em]
  \hline
  \multirow{3}{*}{$a$}
&
 \ket{j} \text{ transition} & \ket{3} \to \ket{2} & \ket{2} \to \ket{1} & \ket{1} \to \ket{0} 
\\[0.2em]
\cline{2-5}
&&&&\\[-1.1em]
&
 \text{Binary transition} & \ket{1}\ket{1} \to \ket{0}\ket{1} & \ket{0}\ket{1} \to \ket{1}\ket{0} & \ket{1}\ket{0} \to \ket{0}\ket{0} 
 \\[0.2em]
\cline{2-5}
&&&&\\[-1em]
& 
  \text{Pauli String} & \frac{1}{2}(X + iY) \otimes \frac{1}{2}(I - Z) & \frac{1}{2}(X - iY) \otimes \frac{1}{2}(X + iY) & \frac{1}{2}(X + iY) \otimes \frac{1}{2}(I + Z)
  \\[0.2em]
  \hline
\end{array}
$
}
\vspace{0.3em}

The above table shows different cases for $\sdag{a}$ and $a$. In each case, we first show the desired basis-vector transitions, the binary-based basis-vector transitions, and then the Pauli string encoding of the transitions. The transformation for the $\textcolor{spec}{t(4)}$ typed creators and annihilators are the sum of the three cases listed above, e.g., $\sdag{a}\textcolor{spec}{:t(4)}$ is transformed to $\frac{1}{2}(X - iY) \otimes \frac{1}{2}(I + Z) + \frac{1}{2}(X + iY) \otimes \frac{1}{2}(X- iY)  + \frac{1}{2}(X - iY) \otimes \frac{1}{2}(I - Z)$. The transformation of the whole sub-expression ($\sapp{\sdag{a}(0)\textcolor{spec}{:t(4)}}{a(1)\textcolor{spec}{:t(4)}}+\sapp{a(0)\textcolor{spec}{:t(4)}}{\sdag{a}(1)\textcolor{spec}{:t(4)}}$) replaces each of the creators and annihilators with the transformed sum terms above.

\subsection{Canonicalization}\label{sec:compilecanonical}

The compiled Pauli-string-based Hamiltonian program $\hat{H}$ in \Cref{sec:transformation} can have a complicated structure, which can be canonicalized to a standard form via the use of the equational rules in \Cref{sec:sem}.
We define the canonical form as follows.

\ignore{
\begin{definition}[Pauli Canonical Form]\label{def:canon-form}\rm 
Let $\mathpzc{P}$ be a syntactic category for Pauli operations $\iseq{r}{P}$ (real number amplitudes), where $P \in \{I,X, Y, Z\}$, every $\hat{H}$ program, typed as $\tjudge{\textcolor{spec}{\bigotimes^n t(2)}}{\hat{H}}{\textcolor{spec}{\quan{F}{\hmx}{\bigotimes^n t(2)}}}$, can be rewritten as the form $\sum_j \bigotimes_{k=0}^{n\sminus 1} \mathpzc{P}_{j,k}$ via the equational rules in \Cref{sec:sem}, where

\begin{itemize}
\item $\hat{P}=\bigotimes_{k=0}^{n\sminus 1} \mathpzc{P}_k$ is defined as $\mathpzc{P}_0 \otimes ... \otimes \mathpzc{P}_{n\sminus 1}$.
\item $\sum_j \hat{P}_j$ is defined as there exists $m$, such that $\hat{P}_0 + ... + \hat{P}_m$.
\end{itemize}
\end{definition}

\begin{definition}[Pauli Canonical Form]\label{def:canon-form}\rm 
Let $\mathpzc{P}$ be a syntactic category for Pauli operations $\iseq{r}{P}$, where $P \in \{I,X, Y, Z\}$ and $r$ is a real amplitude, then a program $\hat{H}$ , typed as $\tjudge{\textcolor{spec}{\bigotimes^n t(2)}}{\hat{H}}{\textcolor{spec}{\quan{F}{\hmx}{\bigotimes^n t(2)}}}$, has Pauli Canonical if it is in the form of $\sum_j \bigotimes_{k=0}^{n\sminus 1} \mathpzc{P}_{j,k}$, where

\begin{itemize}
\item $\hat{P_j}=\bigotimes_{k=0}^{n\sminus 1} \mathpzc{P}_{j, k}$. %is defined as $\mathpzc{P}_{j, 0} \otimes ... \otimes \mathpzc{P}_{j, n\sminus 1}$.
\item It exists $m$, $\sum_{j=0}^m \hat{P}_j$. %= \hat{P}_0 + ... + \hat{P}_m$ supposing there  $m$ terms.
\end{itemize}
\end{definition}
}

\begin{definition}[Pauli Canonical Form]\label{def:canon-form}\rm 
Let $\mathpzc{P} \in \{I, X, Y, Z\}$, a program $\hat{H}$, typed as $\tjudge{\textcolor{spec}{\bigotimes^n t(2)}}{\hat{H}}{\textcolor{spec}{\quan{F}{\hmx}{\bigotimes^n t(2)}}}$, is called Pauli Canonical $\sum_j r_{j} \hat{P}_j$, where $\hat{P}_j=\bigotimes_{k=0}^{n\sminus 1} \mathpzc{P}_{j,k}$ ($r_j$ is a real number).
\end{definition}

          \begin{wrapfigure}{r}{3.8cm}
          %  \vspace*{-0.2em}
           $\begin{array}{|c|c|c|c|}
           \hline
             & X & Y & Z\\           \hline
           X & I & \iseq{i}{Z} & \iseq{-i}{Y} \\           \hline
           Y & \iseq{-i}{Z} & I & \iseq{i}{X} \\           \hline
           Z & \iseq{i}{Y} & \iseq{-i}{X} & I\\           \hline
           \end{array}$
            \caption{Pauli Merging}
            \label{fig:pauli-compile}
          \end{wrapfigure}

\Cref{def:canon-form} states that every program $\hat{H}$ typed as $\tjudge{\textcolor{spec}{\bigotimes^n t(2)}}{\hat{H}}{\textcolor{spec}{\quan{F}{\hmx}{\bigotimes^n t(2)}}}$, where $\slen{\textcolor{spec}{\bigotimes^n t(2)}} = n$, can be rewritten as a two-dimensional list. The outer list corresponds to a sum $\sum_j r_j \hat{P}_j$, and each inner element $\hat{P}_j$ is a fixed-length tensor product $\bigotimes_{k=0}^{n-1} \mathpzc{P}_k$, representing a list of $n$ Pauli operations. Each of these inner lists (i.e., tensor factors) has fixed size $n$, matching the type length $\slen{\iota} = n$ in the typing judgment $\tjudge{\iota}{\hat{H}}{\quan{F}{\hmx}{\iota}}$. The $k$-th element of the tensor list corresponds to the operation applied to the $k$-th site.

\ignore{
\Cref{def:canon-form} states that every program $\hat{H}$, typed as $\tjudge{\textcolor{spec}{\bigotimes^n t(2)}}{\hat{H}}{\textcolor{spec}{\quan{F}{\hmx}{\bigotimes^n t(2)}}}$ ($\slen{\textcolor{spec}{\bigotimes^n t(2)}}=n$), can be rewritten as a two-dimensional list, with the elements in the outer list in a sum relation ($\sum_j r_j \hat{P}_j$); each element $\hat{P}_j$ is a fixed length-$n$ list ($\bigotimes_{k=0}^{n-1} \mathpzc{P}_k$), representing $n$ tensors of Pauli operations.
Note that, each element (a tensor list) in the inner list is always a fixed size $n$ that is the same as the size numbers, defined as $\slen{\iota}=n$ in the typing $\tjudge{\iota}{\hat{H}}{\quan{F}{\hmx}{\iota}}$, as the $k$-th unit element in the tensor list represents the the operations applying to the $k$-th site.
}
%Other than the equational rules in \Cref{fig:exp-proofsystem-1}, in the Pauli-string-based $t(2)$ typed system, we have an additional equational rule above, helping the rewrites of Pauli-string-based Hamiltonian $\hat{H}$ to be Pauli canonical.

Another key simplification is Pauli operation sequence applications ($\pau \in \{I, X, Y, Z\}$) in a single site can be merged into a single Pauli operation, i.e., $\pau_1 \circ \pau_2$ can be rewritten as a single Pauli operation. Applying an identity $I$ with a $\pau$ results in $\pau$. 
The application results of other Pauli operations ($\pau_1$ and $\pau_2$) are defined in \Cref{fig:pauli-compile}.
The two Pauli applications represent simplification, where a sequence of Pauli group applications results in a single Pauli operation with a certain real number amplitude.
We have the following canonicalization correctness lemma.

\begin{lemma}[Canonical Compilation Correctness]\label{thm:canonical-good}\rm 
Given $\hat{H}$, typed $\tjudge{\textcolor{spec}{\bigotimes^n t(2)}}{\hat{H}}{\textcolor{spec}{\quan{F}{\hmx}{\bigotimes^n t(2)}}}$, there is $\hat{H'}$, such that $\textcolor{spec}{{\bigotimes^n t(2)}} \vdash \teq{\hat{H}}{\hat{H'}}$ and $\hat{H'}$ is Pauli canonical in \Cref{def:canon-form}.
\end{lemma}

\begin{proof}
Fully mechanized proofs were done via the equational theory in \Cref{sec:eqthm} and the assumption of a $\dag$-canonical program, with the extra lemma stated in \Cref{thm:canonical-good1}.
\end{proof}

\begin{lemma}[Guaranteed Canonicalization]\label{thm:canonical-good1}\rm 
Given a $\dag$-canonical program (\Cref{def:dag-form-canon}) $e=\sum_j \uapp{z_j}{e_j}$, typed as $\tjudge{\textcolor{spec}{\bigotimes^n t(2)}}{e}{\textcolor{spec}{\quan{F}{\hmx}{\bigotimes^n t(2)}}}$, for every $\uapp{z_j}{e_j}$, either $z_j$ is a real number, or we can find another $\uapp{z_k}{e_k}$ in $e$, such that $e_j = \sdag{e_k}$, and $z_j=\sdag{z_k}$.
\end{lemma}

With the \Cref{thm:equiv-sem}, we can show that the Hamiltonian simulation of $\hat{H}$ and its equivalent canonical form $\hat{H'}$ produce the same result.

\subsection{Gadgetization}\label{sec:gadget}

The \emph{gadgetization} step is defined as the transformation of a canonicalized Pauli-string-based Hamiltonian $\hat{H}$ to another Hamiltonian $\hat{H}'$, with each Pauli string in $\hat{H}'$ being two-local --- a $k$-local Pauli string refers to that the number of non-identity Pauli operations (not equal to $I$) is at most $k$.
This step might happen before or after Trotterization. 

The simplest solution is to perform a unitary synthesization \cite{Whitfield_2011} in the synthesization stage (\Cref{sec:decompose}), which were performed by most previous compiler works \cite{10.1145/3632923,cao2024marqsimreconcilingdeterminismrandomness}, as these works hid the step as a subroutine of their Trotterization algorithms.
However, there are quantum algorithms for performing the transformation, such as the perturbative gadget algorithms.
We believe that it is necessary to substantiate it as a step in our certified compiler, and clarify that there is a key component in compiling a Hamiltonian simulation that transforms a higher local Pauli-string-based Hamiltonian to a two-local Hamiltonian, with the note that this step might happen in the synthesization, which is why the step is marked \textcolor{green}{green} in \Cref{fig:compilationprocess}.

\begin{figure}[t]
    \centering
    \includegraphics[width=1\textwidth]{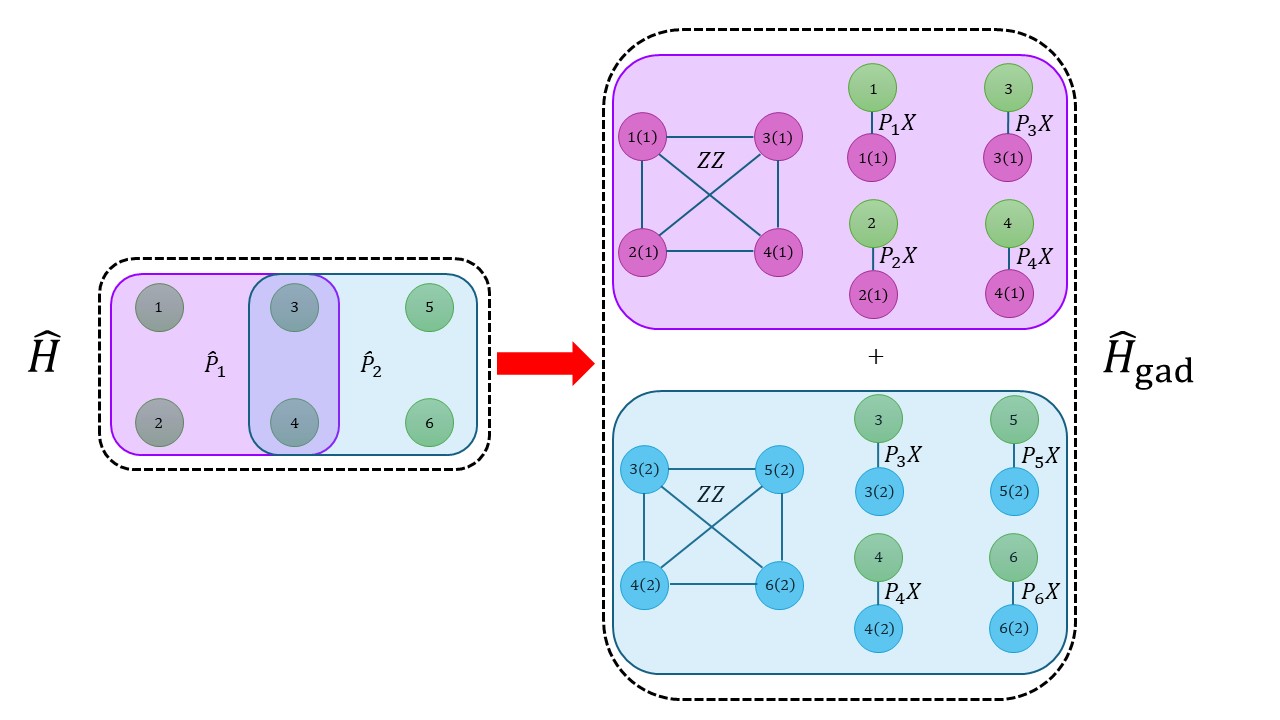}
    \caption{Perturbative gadget for a sum (the dashed areas) of two $4$-local Pauli strings. The \textcolor{teal}{green} numbered nodes are the original qubits; the \textcolor{purple}{purple} (e.g., $3(1)$) and \textcolor{blue}{blue} (e.g., $5(2)$) nodes are ancilla qubits.}
    \label{fig:gadget}
\end{figure}

The rest of the section provides an example demonstrating the \emph{perturbative gadget} algorithm \cite{PhysRevA.77.062329}, which is formalized in \Cref{appx:gadget}. We start with the $\quan{F}{\hmx}{\bigotimes^n t(2)}$ typed canonicalized Pauli-string-based Hamiltonian $\hat{H}=\sum_j r_j \hat{P}_j$, canonicalized in \Cref{sec:compilecanonical}, and produce a two-local Hamiltonian via the algorithm. Assume that we have a $4$-local Hamiltonian with two Pauli strings, as $\hat{H} = r_1 \hat{P}_1 + r_2 \hat{P}_2$; $\hat{P}_1$ and $\hat{P}_2$ are defined on six different Pauli operations $\{\pau_1,...,\pau_6\}$ as follows, demonstrated as left in \Cref{fig:gadget}.

{
\[
\hat{P}_1 = \pau_1 \otimes \pau_2 \otimes \pau_3 \otimes \pau_4 \otimes I \otimes I 
\qquad
\hat{P}_2 =   I \otimes I \otimes \pau_3 \otimes \pau_4 \otimes \pau_5 \otimes \pau_6 
\]
}

As shown in \Cref{fig:gadget}, for each Pauli string $\hat{P}_j$ ($j \in [0,1]$) in a $k$-local Hamiltonian, we assume that $k$ ancilla qubits are provided. To generate the result $H_{\cn{gad}}$ Hamiltonian, we are given $4$ extra qubits for each Pauli string.
For each Pauli string $\hat{P}_j$, the algorithm generates a connection Hamiltonian $\hat{H}^{\cn{anc}}$ to connect the ancilla qubits corresponding to the non-$I$ terms in $\hat{P}_j$, and a behavioral Hamiltonian $\hat{H}^{\cn{v}}$ to capture the constraints in $\hat{P}_j$ from the ancilla qubits back to the original qubits.
To generate the connection Hamiltonians, we produce the following two Hamiltonians for the Pauli strings $\hat{P}_1$ and $\hat{P}_2$:

{\small
\[
\hat{H}^{\cn{anc}}_1 =\sum_{m=1}^3 \sum_{n=m+1}^4 \left(I - Z(m(1)) \circ Z(n(1)) \right)
\qquad
\hat{H}^{\cn{anc}}_2 = \sum_{m=3}^5 \sum_{n=m+1}^6 \left(I - Z(m(2)) \circ Z(n(2)) \right)
\]
}
We use
$\pau(j(k))$ to denote applying the Pauli operation $\pau$ to the $j$-th site in a $k$-th Pauli string $\hat{P}_k$.
Hence
$Z(j(1))$ and $Z(j(2))$ represent the application of $Z$ operation to the ancilla qubit for the $j$-th site in the first and second Pauli string, respectively.
To generate the behavior Hamiltonians, we produce the following two Hamiltonians  for the Pauli strings $\hat{P}_1$ and $\hat{P}_2$:

{\small
\begin{center}
$
\begin{array}{l}
\displaystyle \hat{H}^{\cn{v}}_1 =\sum_{n=1}^4 r'_n \pau_n(n) \circ X(n(1)) \quad \cn{where  } r'_1 = r_1\wedge (n\neq 1 \Rightarrow r'_n = 1)
\\[1.4em]
\displaystyle \hat{H}^{\cn{v}}_2 = \sum_{n=3}^6 r'_n \pau_n(n) \circ X(n(2)) \quad \cn{where  } r'_3 = r_3\wedge (n\neq 3 \Rightarrow r'_n = 1)
\end{array}
$
\end{center}
}
\noindent Similarly, $X((j(1))$ and $X(j(2))$ represent the application of $X$ operation to the ancilla qubit for the $j$-th site in the first and second Pauli strings, respectively.
The result Hamiltonian $\hat{H}_{\cn{gad}}$ is the sum of the four Hamiltonians $\theta_1\hat{H}^{\cn{anc}}_1$, $\theta_2\hat{H}^{\cn{anc}}_2$, $\theta_3\hat{H}^{\cn{v}}_1$ and $\theta_4\hat{H}^{\cn{v}}_2$, with some amplitude values ($\theta_j$ for $j\in[1,4]$). The simulation of $\hat{H}_{\cn{gad}}$ approximates the simulation of the original Hamiltonian $\hat{H}$; the algorithm, the approximation lemma, and the whole compiler pipeline theorem including perturbative gadget are given in \Cref{appx:gadget}.

\subsection{Trotterization}\label{sec:trotter}
%\liyi{Need better examples, and the error rates might have problems.}

Trotterization \cite{Lloyd96} is the step to turn the Hamiltonian simulation operator of a sum of Pauli strings to a sequence of Pauli string simulations, each corresponding to a unitary operation, executable in a quantum computer.
Trotterization is based on the Lie-Trotter product formula (left below) \cite{Lie1880} and the approximated formula, on the right, developed by Suzuki \cite{529425}.

%The former transforms a particle system from a creation-annihilation operation-based description into a Pauli operation-based one. 
%This is significant as the Pauli-operators are easily implemented in quantum computers.
%In \qsnd, we define operations on $t(n,m)$ typed particle systems, while quantum computers are $t(1,2)$ typed, i.e., a two-dimensional particle system with only one spin.
%Hamiltonians can be portrayed as functions. Without considering the particle numbers in a system, the transformation rewrites a function $t(n,m) \to t(n,m)$ ($\quan{F}{\xi}{t(n,m)}$) to be typed $t(1,2) \to t(1,2)$. In addition, as we mentioned in \Cref{sec:sim-gates}, quantum computers require the use of Pauli groups for defining particle behaviors rather than creators and annihilators, i.e., the transformed functions are programmable in Pauli groups because Pauli groups are self-Hermitian; an example is \Cref{sec:boson}.

{\small
\begin{center}
$\cn{exp}(A+B)=\lim_{m\to \infty} (\Pi^m \,\sapp{\cn{exp}({\frac{A}{m}})}{\cn{exp}({\frac{B}{m}})})
\qquad
\qquad
\cn{exp}(\delta(A+B))=\sapp{\cn{exp}(\delta(A))}{\cn{exp}(\delta(B))}+O(\delta^2)
$
\end{center}
}

The Suzuki approximation above suggests that the error rate $O(\delta^2)$ is related to the square of the variable term $\delta$.
Thus, to simulate a quantum system $\eexp{{r \hat{H}}}$ with respect to a time period $r$, one can select a large $m$ to consecutively apply the exponential $\sapp{\cn{exp}({\frac{A}{m}})}{\cn{exp}({\frac{B}{m}})}$ to the state $m$ times, so $\frac{r}{m}$ diminishes such that the error rate $O({(\frac{r}{m})}^2)$ is negligible.
Such tactics can be generalized to an arbitrary number of summation terms.

Through Trotterization, given a time period $r$ and an $n$-qubit Pauli-string-based Hamiltonian $\hat{H}=\sum_{j=0}^{d-1} \theta_j \hat{P}_j$ (with $\hat{P}_j=\bigotimes_{k=0}^{n-1} \mathpzc{P}_{j,k}$), there is $m$ such that $\Pi^m \eexp{\frac{r}{m}\theta_0\hat{P}_0} \circ ...\circ \Pi^m \eexp{\frac{r}{m} \theta_j\hat{P}_j} \circ ...$ approximating $\eexp{{r \hat{H}}}$, where $\Pi^m U = U \circ ... \circ U$, repeating $m$ times of $U$.
Each matrix exponential operation, a.k.a., Hamiltonian simulation operation, generates a unitary, and Trotterization makes all these generated small unitaries in sequence.
The error is related to $r$ and $\theta_j$, detailed below.
%In the context of quantum simulations $\eexp{\theta_j \hat{P}_j }$, $\delta$ refers to the $\sminus \uapp{i}{\theta_j}$.

In \qsnd, we enable two kinds of Trotterization algorithms, the standard algorithm and QDrift \cite{Campbell_2019}.
The standard algorithm uses the Lie-Trotter method. The QDrift method randomizes the sequence of exponential operations, allowing a faster Hamiltonian simulation. 
We first state the general theorem based on an error bound $\epsilon$ (specialized to the two algorithms), then discuss it individually, unveiling the error bounds in different settings.

\begin{lemma}[Trotterization Correctness]\label{thm:trotter-good}\rm 
Given $\hat{H}=\sum_{j=0}^{d-1} \theta_j \hat{P}_j$ and time $r$, typed as $\tjudge{\textcolor{spec}{\bigotimes^n t(2)}}{\hat{H}}{\textcolor{spec}{\quan{F}{\hmx}{\bigotimes^n t(2)}}}$, there is $m$ and $N$, such that, $\eexp{\uapp{r}{\hat{H}}}$'s Trotterization procedure results in $\Pi^N U(d,m,n,r)$, the error is bound by, \\[0.5em]
$
\dabs{\denote{\tjudge{\textcolor{spec}{\bigotimes^{n} t(2)}}{\eexp{\theta\hat{H}}}{\textcolor{spec}{\quan{F}{\umx}{\bigotimes^n t(2)}}}}_0 - \denote{\tjudge{\textcolor{spec}{\bigotimes^n t(2)}}{\Pi_{k=0}^{N-1} U(d,m,n,r)}{\textcolor{spec}{\quan{F}{\umx}{\bigotimes^n t(2)}}}}_0}<\epsilon$.
\\
\end{lemma}

\begin{proof}
To show \Cref{thm:trotter-good} for the standard Trotterization and QDrift, our Coq proof performs induction on $N$, with key error bound lemmas discussed below.
\end{proof}

As a running example, we use the following simple Ising model and show its Trotterized results along with the two algorithms' descriptions.

\[
\hat{H}_{I1} = \sum_j r_h X(j) + \sum_{j} {Z(j)}\circ{Z(j\splus 1)}
\]
%We provide the method of deciding $\alpha_k$, $\tau_k$, and the error $\epsilon$ in \label{thm:trotter-good} for the two methods, respectively.

\myparagraph{The Standard Algorithm.}
In the standard form, for each Pauli string $\theta_j \hat{P}_j$, we construct an application sequence $\Pi^m \eexp{\frac{r}{m} \theta_j \hat{P}_j }$, i.e., we repeat the simulation of $\theta_j \hat{P}_j$ for $m$ times with a scaled time slot $\frac{r}{m}$.
Formally, $U(d,m,n,r)$ is defined as $\eexp{\frac{r}{m} \theta_j \hat{P}_j }$ with $N=d * m$, i.e., $m$ being a splitting number 
of each Pauli string and $d$ being the number of Pauli strings in $\hat{H}$.
Thus, the whole procedure produces an application sequence, explained below.
%The error rate $\epsilon$ is related to $O(\frac{r t^2}{m^2})$.

For example, to simulate the Pauli-string-based Hamiltonian $\hat{H}_{I1}$ with a time slot $r$ and splitting number $m$, we compute the approximation of the exponent $\eexp{r{\hat{H}_{I1}}}$ as follows.

{\footnotesize
\begin{center}
$\hspace{-0.5em}
\begin{array}{l@{\;}c@{\;}l}
{\eexp{r{\hat{H}_{I1}}}}
&=& {{\Pi}^m\left(\eexp{\frac{r}{m}{(\sum_{j} \sapp{Z(j)}{Z(j\splus 1)} + r_h \sum_j X(j))}}\right)}\\
&=& {{\Pi}^m\left(\sapp{\sapp{\eexp{{{\frac{r}{m}}\sapp{Z(0)}{Z(1)}}}}{ \sapp{\eexp{{\frac{r}{m}}{\sapp{Z(1)}{Z(2)}}}}{...}}}{\sapp{\eexp{{\frac{r}{m}r_h}{X(0)}}}{\sapp{\eexp{{\frac{r}{m}r_h}{X(1)}}}{...}}} \right)}
\end{array}
$
\end{center}
}

Let $\hat{H}=\sum_{j=0}^{d\sminus 1} \theta_j \hat{P}_j$ with $N=d * m$, if we list all the trotted Hamiltonian simulation sequences by listing all Pauli strings, we form an $N$-element set as $\hat{H}_k\in \{\frac{\theta_0}{m}\hat{P}_0,\frac{\theta_0}{m}\hat{P}_0,...,\frac{\theta_j}{m}\hat{P}_j,\frac{\theta_j}{m}\hat{P}_j,..., \frac{\theta_{d\sminus 1}}{m}\hat{P}_{d\sminus 1}\}$.
$\hat{H}$ can be viewed as a sum of $N$ different element Hamiltonians $\hat{H}_k$ in the set.
We estimate the error bound $\epsilon$ related to $O(r^2)$, and the exact formula is listed below \cite{PhysRevX.11.011020}.

{\small
\[
\displaystyle
\epsilon = \displaystyle { \frac{r^2}{2} \sum_{k=1}^N \left\lVert 
(\sum_{j=k+1}^N \hat{H}_j) \circ \hat{H}_k - \hat{H}_k \circ (\sum_{j=k+1}^N \hat{H}_j)
\right\rVert}.
\label{eq:trottererr_std}
\tag{5.4.1}
\]
}

The Coq proof relies on a key error bound lemma (\Cref{appx:trotterthm}) based on the error bound in \Cref{eq:trottererr_std}.

\myparagraph{QDrift.}
We manage the Pauli strings in $\hat{H}=\sum_{k=0}^{d\sminus 1} \theta_k \hat{P}_k$ as a set $\{\theta_0\hat{P}_0,..,\theta_{d\sminus 1}\hat{P}_{d\sminus 1}\}$. Let $\lambda = \sum_{k=0}^{d\sminus 1} \slen{\theta_k}$, we probabilistically pick a Pauli string in the set and repeat the process $N$ times; such picks are listed as a sequence $m_0, ...m_k,..., m_{N\sminus 1}$, each $m_k \in [0,d)$. The QDrift procedure produces an application sequence $\Pi_{k=0}^{N-1} \eexp{\frac{r\lambda}{N}\theta_{m_k}\hat{P}_{m_k} }$, i.e., the $U(d,m,n,r)$ is set to $\eexp{\frac{r\lambda }{N}\theta_{m_k}\hat{P}_{m_k} }$.
Here the parameter $m$ in $U(d,m,n,r)$ is not explicitly used but it satisfies $N = d*m$.

We demonstrate the key difference of QDrift with respect to the standard Trotterization using the $\hat{H}_{I1}$ example above. To approximate the exponential $\eexp{r{\hat{H}_{I1}}}$, we randomly choose $N$ Pauli strings in the set $\{r_h X(0), r_h X(1),...,Z(0) \circ Z(1), Z(1) \circ Z(2), ...\}$, and then construct the trottered sequence. One example choice is as follows.
%For example, to simulate $\hat{H}_{I1}$ in \Cref{eq:hi1} with a time slot $t$ through the pick of a splitting number $n$, we compute the exponent $\eexp{\sapp{\hat{H}_{I1}}{t}}$ using the approximation as follows.

{\footnotesize
\begin{center}
$\eexp{r{\hat{H}_{I1}}}
= \sapp{\sapp{\eexp{{\frac{r\lambda}{N}}{\sapp{Z(j)}{Z(j\splus 1)}}}}{\sapp{\eexp{{\frac{r\lambda}{N}}{\sapp{Z(k)}{Z(k\splus 1)}}}}{...}}}{\sapp{\eexp{{\frac{r\lambda}{N}r_h{X(k)}}}}{\sapp{\eexp{{\frac{r\lambda}{N}r_h}{X(0)}}}{...}}}
$ 
\end{center}
}

Let $\hat{H}=\sum_{k=0}^{d-1} \theta_k \hat{P}_k$ and $\lambda = \sum_{k=0}^{d-1} \slen{\theta_k}$, QDrift randomly picks $N$ different Pauli strings for simulation, the error rate $\epsilon$ is related to $O(\frac{r^2}{N})$, with the exact bound listed below \cite{Campbell_2019}.

{\small
\[
\displaystyle
\epsilon = \displaystyle{ \frac{2\lambda^2 r^2}{N}}
\tag{5.4.2}
\label{eq:trottererr_qdrift}
\]
}

The Coq proof relies on a key error bound lemma (\Cref{appx:trotterthm}) based on the error bound in \Cref{eq:trottererr_qdrift}.

\subsection{Synthesization}\label{sec:decompose}

Synthesization applies the application sequence of matrix exponentials of Pauli strings generated from Trotterization to quantum circuits, either digital or analog.
In synthesizing an application sequence $\Pi_{j} \eexp{\theta_j\hat{P}_j}$, it is enough to describe the synthesization of a single matrix exponential term $\eexp{\theta_j \hat{P}_j}$. We discuss the digital and analog-based synthesization below.

%Digital Synthesization Pseudocode
\begin{algorithm}[t]
{\small
\caption{Digital Synthesization of \( U = \exp(-i \theta \hat{P}) \) for a Pauli string \(\hat{P}\)}
\label{fig:digitaldeco}
\begin{algorithmic}[1]
\STATE \textbf{Input:} Pauli string \( \hat{P} = \bigotimes_{k=0}^{n\sminus 1} \mathpzc{P}_{k} \), time period \( \theta \)
\STATE \textbf{Output:} Digital circuit \( U = \exp(-i \theta \hat{P}) \)
\STATE

\STATE Collect sites \( \{m_1, m_2, \ldots, m_l\} \) where \( \mathpzc{P}_{m_j} \neq I \), with \( l \leq n \)

\STATE \( U \leftarrow \cn{Rz}(2\theta)(m_1) \)
\STATE $j = m_1$
\FOR{ \( m_j \in \{m_2, \ldots, m_l\} \)}
    \STATE \( U \leftarrow \cn{CX}(m_{j-1}, m_j) \circ U \circ \cn{CX}(m_{j-1}, m_j) \)
\ENDFOR

\FOR{ \( m_j \in \{m_1, \ldots, m_l\} \)}
    \IF{\( \mathpzc{P}_{m_j} = Y \)}
        \STATE \( U \leftarrow \cn{Rz}\left(\tfrac{\pi}{2}\right)(m_j) \circ \cn{H}(m_j) \circ U \circ \cn{H}(m_j) 
        \circ \cn{Rz}\left(-\tfrac{\pi}{2}\right)(m_j) \)
    \ELSIF{\( \mathpzc{P}_{m_j} = X \)}
        \STATE \( U \leftarrow \cn{H}(m_j) \circ U \circ \cn{H}(m_j) \)
    \ENDIF
\ENDFOR

\STATE \textbf{Return} \( U \)
\end{algorithmic}
}
\end{algorithm}

\myparagraph{Digital Synthesization.}
The digital synthesization synthesized the simulation of a Pauli string, $\eexp{\theta_j \hat{P}_j}$, to elementary quantum gates, e.g., single-qubit gates, such as \cn{H}, \cn{Rz}, \cn{Ry}, and \cn{Rx} gates, as well as two qubit controlled-not gate \cn{CX} ($\cn{CX}(x,y)$ controlling on $x$ and applying $\cn{X}$ gate on $y$).

The procedure focuses on the non-$I$ terms in a Pauli string $\hat{P}=\bigotimes_{j=0}^{n-1} \mathpzc{P}_j$. If $\hat{P}$ has only one non-$I$ term, the compilation generates a rotation gate, listed in \Cref{fig:paulisim}.
In the Bloch qubit sphere interpretation \cite{mike-and-ike}, the rotation gate $\cn{Rx}$, $\cn{Ry}$, $\cn{Rz}$ can be expressed as some Hamiltonian simulations of Pauli matrices ($\cn{Rz}(2 \theta)(j)$ applies a $Z$-axis rotation in the $j$-th qubit), which essentially suggests that the simulation of singleton Pauli $X$, $Y$, and $Z$ terms are the above phase rotation gates.
For example, the gate synthesis of $\eexp{{\frac{r}{n}}{X(j)}}$ generates an $\cn{Rx}$ gate, as $\cn{Rx}(\frac{2r}{n})(j)$. 
%The gate synthesis of the $z$-axis interaction exponent $\eexp{\sapp{\sapp{Z(j)}{Z(j\splus 1)}}{\frac{t}{n}}}$ is introduced in \Cref{sec:overview2}.
%Ultimately, the simulation of the whole $\hat{H}_{I1}$ system is a series of applications of the above two gate patterns.
We describe the algorithm in \Cref{fig:paulisim} for a Pauli string with more than one non-$I$ term, i.e., a $k$-local Pauli string with $k>1$.
 
 \ignore{
\begin{definition}[Digital Synthesization]\label{def:digital-algorithm}\rm 
Let $U=\eexp{\hat{P} \theta}$ and the Pauli string $\hat{P}=\bigotimes_{k=0}^n \alpha(k)$ is $k$-local with $k > 1$, to generate unitary quantum gates, we perform the following.

\begin{enumerate}
\item We collect all the non-$I$ positions as $\{m_1,...,m_l\}$ with $l < n$, we generate the circuit as $U'=\Pi_{d<j} \cn{CX}(m_d,m_j) \circ \cn{RZ}(2\theta)(m_j)\circ \Pi_{j>d} \cn{CX}(m_d,m_j)$, with $d,j \in  \{m_1,...,m_l\}$.

\item With the $U'$ implementation in (1), we re-examine each non-$I$ position as $\{m_1,...,m_l\}$, if $m_j$ is Pauli $Y$, we compose $U'$ with $\cn{RX}(\frac{7\pi}{2})(m_j)\circ U' \circ \cn{RX}(\frac{\pi}{2})(m_j)$, if $m_j$ is Pauli X, we compose $U'$ with $\cn{RX}(\pi)(m_j) \circ \cn{RY}(\frac{\pi}{2})(m_j)\circ U' \circ \cn{RY}(\frac{\pi}{2})(m_j)\circ \cn{RX}(\pi)(m_j)$.
\end{enumerate}
\end{definition}
}

\ignore{
\begin{algorithm}[t]
{\small
\caption{Digital Synthesization of \( U = \exp(-i \theta \hat{P}) \) for a \(k\)-local Pauli string \(\hat{P}\)}
\label{fig:digitaldeco}
\begin{algorithmic}[1]
\STATE \textbf{Input:} Pauli string \( \hat{P} = \bigotimes_{k=0}^{n\sminus 1} \alpha^{(k)} \), time period \( \theta \)
\STATE \textbf{Output:} Digital circuit implementing \( U = \eexp{ \hat{P}\theta} \)
\STATE

\STATE Collect non-identity positions \( \{m_1, m_2, \ldots, m_l\} \) where \( \alpha^{(m_j)} \neq I \), with \( l < n \)

\STATE Construct base circuit:
\STATE \( U' \leftarrow \left(\prod_{d < j} CX(m_d, m_j)\right) \circ RZ(2\theta)(m_j) \circ \left(\prod_{j > d} CX(m_d, m_j)\right) \)
\STATE where \( d, j \in \{m_1, \ldots, m_l\} \)

\FOR{each \( m_j \in \{m_1, \ldots, m_l\} \)}
    \IF{\( \alpha^{(m_j)} = Y \)}
        \STATE \( U \leftarrow RX\left(\tfrac{7\pi}{2}\right)(m_j) \circ U' \circ RX\left(\tfrac{\pi}{2}\right)(m_j) \)
    \ENDIF
    \IF{\( \alpha^{(m_j)} = X \)}
        \STATE \( U \leftarrow RX(\pi)(m_j) \circ RY\left(\tfrac{\pi}{2}\right)(m_j) \circ U' \circ RY\left(\tfrac{\pi}{2}\right)(m_j) \circ RX(\pi)(m_j) \)
    \ENDIF
\ENDFOR

\STATE \textbf{Return} \( U \) as the implemented unitary circuit
\end{algorithmic}
}
\end{algorithm}
}

\begin{figure}[t]
{\small
\begin{center}
$
\begin{array}{c}
{\eexp{{\theta}{X}}}=\cn{RX}(2{\theta})
=
\begin{psmallmatrix}
\cos{{\theta}} & -i\sin{{\theta}}\\
-i\sin{{\theta}} & \cos{{\theta}}
\end{psmallmatrix}
\qquad
{\eexp{{\theta}{Y}}}=\cn{RY}(2{\theta})
=
\begin{psmallmatrix}
\cos{{\theta}} & -\sin{{\theta}}\\
\sin{{\theta}} & \cos{{\theta}}
\end{psmallmatrix}
\\[0.3em]
{\eexp{{\theta}{Z}}}=\cn{RZ}(2{\theta})
=
\begin{psmallmatrix}
\eexp{{\theta}} & 0\\
0 & \cn{exp}{(i {\theta})}
\end{psmallmatrix}
\end{array}
$
\end{center}
}
\vspace*{-0.8em}
  \caption{Pauli matrices simulated to rotation gates. $\eexp{\uapp{\theta}{\pau}}$ can be thought as simulating $\pau$ on time $\theta$.}
  \label{fig:paulisim}
\end{figure}

\begin{figure*}[t]
{\small
\begin{tabular}{c@{$\;\;=\;\;$}c}
\begin{minipage}{0.22\textwidth}
$
\eexp{\theta{\sapp{Z(j)}{Z(j\splus 1)}}}
$
\end{minipage}
&
\begin{minipage}{0.3\textwidth}
{
  \Qcircuit @C=0.5em @R=0.5em {
    \lstick{} & & \ctrl{1}       & \qw  & \ctrl{1}      &\qw\\
    \lstick{} & & \targ         & \gate{\cn{Rz}(2\theta)}       & \targ    & \qw 
    }
}   
\end{minipage}
\end{tabular}
\hspace{1em}
\begin{minipage}{0.48\textwidth}
{\small
$
\eexp{\theta ZZ}=
\begin{psmallmatrix}
\cn{exp}(\sminus i\frac{\theta}{2}) & 0 & 0 & 0\\
 0 & \cn{exp}(i\frac{\theta}{2}) & 0 & 0\\
 0 & 0 & \cn{exp}(i\frac{\theta}{2}) & 0\\
 0 & 0 & 0 & \cn{exp}(\sminus i\frac{\theta}{2})
\end{psmallmatrix}
$
}
  \end{minipage}
}
\vspace*{-0.8em}
\caption{Circuit and unitary for $ZZ$ Interaction. The left circuit applies on $j$ and $j\splus 1$ qubits.}
\label{fig:zz-inter}
\end{figure*}

\begin{figure*}[t]
{\small
\begin{minipage}[b]{0.3\textwidth}
{\centering
{
  \Qcircuit @C=0.5em @R=0.5em {
    \lstick{} & \ctrl{1} & \qw          & \qw                            & \qw      & \ctrl{1} & \qw      \\
    \lstick{} & \targ    & \ctrl{1}     & \qw                            & \ctrl{1} & \targ & \qw \\
    \lstick{} & \qw      & \targ        & \multigate{1}{\eexp{\theta ZZ}} & \targ    & \qw   & \qw \\
    \lstick{} & \qw      & \qw          & \ghost{\eexp{\theta ZZ}}      & \qw      & \qw   & \qw
    }
}}
\subcaption{$\eexp{{\theta}{ZZZZ}}$ gate.}
\label{fig:zzzz-gate}  
\end{minipage}
\begin{minipage}[b]{0.29\textwidth}
{\centering
{
  \Qcircuit @C=0.5em @R=0.5em {
    \lstick{} & \gate{\cn{H}} & \multigate{3}{\eexp{{\theta}{ZZZZ}}} & \gate{\cn{H}} & \qw   \\
    \lstick{} &\gate{\cn{H}}  & \ghost{\eexp{{\theta}{ZZZZ}}}        & \gate{\cn{H}} & \qw  \\
    \lstick{} & \gate{\cn{H}} & \ghost{\eexp{{\theta}{ZZZZ}}}        & \gate{\cn{H}} & \qw  \\
    \lstick{} & \gate{\cn{H}} & \ghost{\eexp{{\theta}{ZZZZ}}}        & \gate{\cn{H}} & \qw
    }
}}
\subcaption{$\eexp{{\theta}{XXXX}}$ gate.}
\label{fig:xxxx-gate}  
\end{minipage}
\begin{minipage}[b]{0.29\textwidth}
{\centering
{
  \Qcircuit @C=0.5em @R=0.5em {
    \lstick{} & \qw & \multigate{3}{\eexp{{\theta}{XXXX}}} & \qw & \qw   \\
    \lstick{} &\qw  & \ghost{\eexp{{\theta}{XXXX}}}        & \qw & \qw  \\
    \lstick{} & \gate{\cn{S}} & \ghost{\eexp{{\theta}{XXXX}}}        & \gate{\sdag{\cn{S}}} & \qw  \\
    \lstick{} & \gate{\cn{S}} & \ghost{\eexp{{\theta}{XXXX}}}        & \gate{\sdag{\cn{S}}} & \qw
    }
}}
\subcaption{$\eexp{{\theta}{XXYY}}$ gate.}
\label{fig:xxyy-gate}  
\end{minipage}
}
\vspace*{-0.5em}
\caption{Synthesized circuits for some long Pauli strings.}
\label{fig:bosongates}
\vspace*{-0.8em}
\end{figure*}
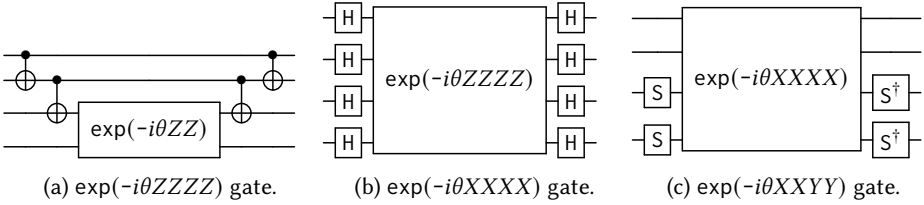

We first show the gate synthesization of the simple $t(2)$ typed Ising model $\hat{H}_{I1}$ in \Cref{sec:trotter}, which produces a sequence of $\eexp{\frac{r}{n} {Z(j)}\circ{Z(j\splus 1)}}$ and $\eexp{\frac{r}{n} r_h X(j)}$ operations (the sequence length is related to $n$). If $n$ is large enough, the operation sequence is a good approximation to $\eexp{r \hat{H}_{I1}}$.
The $\eexp{\frac{r}{n} r_h X(j)}$ operation produces a \cn{Rx} gate \cite{mike-and-ike}, while $\eexp{ \frac{r}{n} {Z(j)}\circ{Z(j\splus 1)}}$ produces a digital circuit as the left in \Cref{fig:zz-inter}, performing two controlled-not gates and an $\cn{Rz}(2\frac{r}{n})$ gate in the middle for the $j$-th and $j\splus 1$-th qubits,
i.e., $\cn{Rz}(2\theta)$ is an $z$-axis phase rotation gate.

The above assumes a two-local Pauli string with non-$I$ terms being $Z$.
If we do not perform a perturbative gadget, we might have $k$-local Pauli strings with $k > 2$.
The general algorithm is described in \Cref{fig:digitaldeco}.

We show an example of a long Pauli string synthesization in \Cref{fig:bosongates}.
Such long Pauli strings might appear in the Hamiltonian simulation of many systems, such as the bosonic system $\hat{H}$ in \Cref{sec:bosonsystem}.
%We show an example of the synthesization of the Hamiltonian simulation of a large $k$-local Pauli string, i.e., the example of the boson system $\hat{H}_T$ in \Cref{sec:bosonsystem}.
%We apply the matrix exponential to the Hamiltonian to perform the simulation with respect to time $r$, i.e., $\eexp{{\hat{H}_T}{r}}$.
After applying Trotterization on the $\hat{H}$ for boson systems, we generate a sequence of Hamiltonian simulations of Pauli strings.
Some possible Pauli string terms might be $4$-local, as $Z \otimes Z \otimes Z \otimes Z$ (abbreviated as $ZZZZ$), $X\otimes X \otimes X \otimes X$ (abbreviated as $XXXX$) and $X\otimes X \otimes Y \otimes Y$ (abbreviated as $XXYY$).
The simulation of the $ZZZZ$ term results in a quadruple $Z$ interaction in \Cref{fig:zzzz-gate}, containing a $ZZ$ interaction application in the middle, shown in \Cref{fig:zz-inter}, with the surrounding of controlled-not gates.
The digital gate synthesization of $\eexp{{\theta}{XXXX}}$ and $\eexp{{\theta}{XXYY}}$ is based on the quadruple $Z$ interaction with reapirs for $X$ and $Y$ conversion, which results in the circuits in \Cref{fig:xxxx-gate,fig:xxyy-gate}.
We show the correctness lemma below.

\begin{lemma}[Ditigal Synthesization Correctness]\label{thm:decompose-good}\rm 
$\eexp{\theta \hat{P}}$ can be synthesized as a sequence of unitary operations as $\Pi_{k}U_k$, based on the elementary gates, $\{\cn{H}, \cn{Rx}, \cn{Ry}, \cn{Rz}, \cn{CX}\}$.
\end{lemma}

\begin{proof}
    Fully mechanized proofs were done by induction on the Pauli canonical forms. By applying \Cref{fig:digitaldeco}, we produce a unitary $U$ for $\eexp{\theta \hat{P}}$. Then, we interpret the two terms as matrix representations in VOQC \cite{VOQC} and equate the two interpretations.
\end{proof}

\myparagraph{Analog Synthesization.}
Many quantum computers provide backends to perform analog quantum computation, i.e., we can view a quantum computer as a Hamiltonian simulator to simulate small Hamiltonians in sequence. Each simulation step in a quantum simulator takes in a positive real number $\theta$ and a Pauli string $\hat{P}$, and produce the simulation result as $\eexp{\theta \hat{P}}$. The procedure is machine-dependent. 
For example, an IBM machine permits the simulation of one or two-local Pauli strings in the set $\{X, Z, ZZ\}$ ($ZZ$ means a two-local Pauli string with only $Z$ Pauli operations). Another machine (Indiana), developed by Richerme \emph{et al.} \cite{richerme2025multimodeglobaldrivingtrapped,Kyprianidis_2024}, permits the simulation of $k$-local Pauli strings in the set $\{Z, X^*\}$, where $X^*$ permits a Pauli string simulation with an arbitrary number of $X$ Paulis.
We show two analog synthesization algorithms in \Cref{fig:analogdeco,fig:inco} to compile our Pauli strings to the IBM and Indiana machines, respectively.

\ignore{
\begin{definition}[Analog Synthesization]\label{def:analog-algorithm}\rm 
Let $U=\eexp{\hat{P} \theta}$ and the Pauli string $\hat{P}=\bigotimes_{k=0}^n \alpha(k)$ is two-local, to generate unitary quantum gates, we collect all the non-$I$ positions as $m_1$ and $m_2$, and performs the following. 

\begin{enumerate}
\item If $\alpha(m_1)$ and $\alpha(m_2)$ and are $Z$, we construct $U'=\eexp{Z(m_1) Z(m_1) \theta}$.
\item If either $\alpha(m_1)$ or $\alpha(m_2)$ and are $Y$, we construct an additional $f(U') = \eexp{X(m_j)\frac{7\pi}{2}} \circ U' \circ \eexp{X(m_j)\frac{\pi}{2}}$, where $U'$ is generated in case (1) and $j \in [1,2]$.
\item If either $\alpha(m_1)$ or $\alpha(m_2)$ and are $X$, we construct an additional $f(U') = \eexp{X(m_j)\pi} \circ \eexp{Y(m_j)\frac{\pi}{2}}\circ U' \circ \eexp{Y(m_j)\frac{\pi}{2}}\circ \eexp{X(m_j)\pi}$, where $U'$ is generated in case(1) and $j \in [1,2]$.
\end{enumerate}
\end{definition}
}

\ignore{
%Analog synthesization Pseudocode
\begin{algorithm}[t]
{\small
\caption{IBM-based Analog Synthesization of \( U = \exp(-i \theta \hat{P}) \) for $2$-Local Pauli Strings}
\label{fig:analogdeco}
\begin{algorithmic}[1]
\STATE \textbf{Input:} $2$-local Pauli string \( \hat{P} = \bigotimes_{k=0}^{n\sminus 1} \mathpzc{P}_{k} \), time period \( \theta \)
\STATE \textbf{Output:} Analog circuit \( U = \exp(-i \theta \hat{P}) \)

\STATE Collect sites \( \{m_1, m_2, \ldots, m_l\} \) where \( \mathpzc{P}_{m_j} \neq I \), with \( l \leq n \)
\IF{$l < 2$}
    \STATE Raise Exception ("Required 2-Local Pauli Strings!")
    \STATE Exit
\ELSIF{$l = 1$}
    \STATE \(  \textbf{Return} \hspace{0.4 em} \eexp{\theta\pau(m_1)} \)
\ENDIF

\STATE \( U \leftarrow \eexp{ \theta Z({m_1}) \circ Z({m_2})} \)

\FOR{ \( m_j \in \{m_1, m_2\} \)}
    \STATE \( H(m_j) = \eexp{\frac{\pi}{4} X(m_j)} \circ  \eexp{\frac{\pi}{4} Z(m_j)} \circ  \eexp{\frac{\pi}{4} X(m_j)} \)
    \IF{\( \mathpzc{P}_{m_j} = Y \)}
        \STATE \( U \leftarrow \eexp{\frac{\pi}{4}Z(m_j)} \circ 
        \eexp{\frac{\pi}{4} X(m_j)} \circ  \eexp{\frac{\pi}{4} Z(m_j)} \circ  \eexp{\frac{\pi}{4} X(m_j)}
        \circ U \circ 
        \eexp{\frac{\pi}{4} X(m_j)} \circ  \eexp{\frac{\pi}{4} Z(m_j)} \circ  \eexp{\frac{\pi}{4} X(m_j)}\circ \pexp{\frac{\pi}{4}Z(m_j)} \)
    \ELSIF{\( \mathpzc{P}_{m_j} = X \)}
        \STATE \( U \leftarrow 
            \eexp{\frac{\pi}{4} X(m_j)} \circ  \eexp{\frac{\pi}{4} Z(m_j)} \circ  \eexp{\frac{\pi}{4} X(m_j)}
        \circ U \circ 
        \eexp{\frac{\pi}{4} X(m_j)} \circ  \eexp{\frac{\pi}{4} Z(m_j)} \circ  \eexp{\frac{\pi}{4} X(m_j)}
        \)
    \ENDIF
\ENDFOR

\STATE \textbf{Return} \( U \)
\end{algorithmic}
}
\end{algorithm}
}

\begin{algorithm}[t]
{\small
\caption{IBM-based Analog Synthesization of \( U = \exp(-i \theta \hat{P}) \) for $2$-Local Pauli Strings}
\label{fig:analogdeco}
\begin{algorithmic}[1]
\STATE \textbf{Input:} $2$-local Pauli string \( \hat{P} = \bigotimes_{k=0}^{n\sminus 1} \mathpzc{P}_{k} \), time period \( \theta \)
\STATE \textbf{Output:} Analog Hamiltonian simulation \( U = \Pi^m\eexp{  r \hat{P}} \)
\STATE
\STATE Identify non-identity sites \( m_1 \) and \( m_2 \) where \( \mathpzc{P}_{m_j} \neq I \)

\STATE \( U \leftarrow \eexp{ \theta Z({m_1}) \circ Z({m_2})} \)

\FOR{ \( m_j \in \{m_1, m_2\} \) }
    \IF{\( \pau_{m_j} \neq Z \)}
            \STATE \textbf{let} $U_1 =  \eexp{\frac{\pi}{4} X(m_j)} \circ  \eexp{\frac{\pi}{4} Z(m_j)} \circ  \eexp{\frac{\pi}{4} X(m_j)}$ \textbf{in}
                 \STATE  \( U \leftarrow  U_1 \circ U \circ U_1\) 
    \ENDIF
    \IF{\( \pau_{m_j} = Y \)}
       \STATE \( U \leftarrow \eexp{ \tfrac{\pi}{4} Z({m_j})} \circ U \circ \eexp{\tfrac{7\pi}{4} Z({m_j})} \)
    \ENDIF
\ENDFOR

\STATE \textbf{Return} \( U \)
\end{algorithmic}
}
\end{algorithm}

To perform analog synthesization for the IBM machine, we assume that a two-local Hamiltonian is achieved by applying a perturbative gadget.
The simulation of a one-local Pauli can be dealt with by the matrix exponential formulas in \Cref{fig:paulisim}.
The Indiana machine permits the simulation of a $k$-local Pauli string, i.e., we do not need perturbative gadget to generate two-local Pauli strings. 

%Analog Synthesization Pseudocode
\begin{algorithm}[t]
{\small
\caption{Indiana Analog Synthesization of \( U = \exp(-i \theta \hat{P}) \)}
\label{fig:inco}
\begin{algorithmic}[1]
\STATE \textbf{Input:} Pauli string \( \hat{P} = \bigotimes_{k=0}^{n\sminus 1} \mathpzc{P}_{k} \), time period \( \theta \)
\STATE \textbf{Output:} Analog circuit \( U = \exp(-i \theta \hat{P}) \)
\STATE

\STATE Collect sites \( \{m_1, m_2, \ldots, m_k\} \) where \( \mathpzc{P}_{m_j} \neq I \), with \( l \leq n \)

\STATE \( \hat{P'} = I \)
\FOR{ \( m_j \in \{m_1, \ldots, m_k\} \)}
    \STATE \( \hat{P'} \leftarrow \hat{P'} \circ X(m_j) \)
\ENDFOR

\STATE  $U \leftarrow \eexp{\uapp{\theta}{\hat{P'} }}$

\FOR{ \( m_j \in \{m_1, \ldots, m_k\} \)}
    \IF{\( \mathpzc{P}_{m_j} = Y \)}
        \STATE \( U \leftarrow \eexp{\frac{\pi}{4} Z(m_j)} \circ U \circ \eexp{\frac{7\pi}{4} Z(m_j)} \)
    \ELSIF{\( \mathpzc{P}_{m_j} = Z \)}
        \STATE \textbf{let} $U_1 =  \eexp{\frac{\pi}{4} X(m_j)} \circ  \eexp{\frac{\pi}{4} Z(m_j)} \circ  \eexp{\frac{\pi}{4} X(m_j)}$ \textbf{in}
                 \STATE  \( U \leftarrow  U_1 \circ U \circ U_1\) 
    \ENDIF
\ENDFOR

\STATE \textbf{Return} \( U \)
\end{algorithmic}
}
\end{algorithm}

As an example, to simulate a simple Ising system $\hat{H}_{I1}$ above, we can directly perform the simulation $\eexp{\frac{r}{n}\sapp{Z(j)}{Z(k)}}$ and $\eexp{\frac{r}{n}r_h {X(j)}}$ for sites $j$ and $k$, in the IBM setting. If $\frac{r}{n}$ is negative, we can instead simulate a positive number $\frac{r}{n}+2m\pi$. The direct Hamiltonian simulation permits a more efficient circuit compilation.
The above analog synthesization step essentially implements the $ZZ$-interaction gate in \Cref{fig:zz-inter}, and also shows the reason for some interaction gates being much more efficient than elementary gates in IBM machines, evidenced in \citet{10.1145/3632923}.

To perform the simulation $\eexp{{\theta}{XXYY}}$ in the Indiana machine, we can construct the simulation procedures (\Cref{fig:inco}) as follows:
We construct the simulation opeartion $\eexp{\theta XXXX}$, with the single qubit simulations $\eexp{\frac{\pi}{4} Z}$ and $\pexp{\frac{7\pi}{4} Z}$ around the last two qubits.

Other quantum machines might have different permitted Pauli string sets, such as Quera machines \cite{quera2024}. Consequently, the machine Hamiltonian analog synthesization step must employ different analog synthesization methods for each quantum machine.
%The Ising system simulation is a toy example. \Cref{sec:boson} presents another example involving the simulation of a bosonic system.
We show the correctness lemma for the two machines below.

\begin{lemma}[Analog Synthesization Correctness]\label{thm:analog-decompose-good}\rm 
$\eexp{\theta \hat{P}}$ can be synthesized as a sequence of unitary operations as $\Pi_{k} U_k$, based on the simulations of Pauli strings permitted in the IBM machine ($\{X, Z, ZZ\}$) and the Indiana machine ($\{Z, X^*\}$).
\end{lemma}

\begin{proof}
    Fully mechanized proofs were done by the same methodology as proving \Cref{thm:decompose-good}. Here, we equate the two matrix interpretations of $\eexp{\theta \hat{P}}$ and the unitaries generated via \Cref{fig:analogdeco,fig:inco}.
\end{proof}

\subsection{The Whole Pipeline Theorem}\label{sec:mainthm}

Below, we show the correctness theorem for compiling a Hamiltonian simulation of a constraint expression $e$, typed as $\tjudge{\iota}{e}{\quan{F}{\hmx}{\iota'}}$, to quantum circuits. We first show the correctness theorem below for the compilation pipeline in \Cref{fig:compilationprocess}, without the perturbative gadget step; such a pipeline is modeled as $\quan{F}{\hmx}{\iota} \vdash (e,r) \gg U : \quan{F}{\umx}{\iota'}$. Such a compilation is split into the four steps above. The theorem is a complete version of \Cref{thm:compile-good-simp}.

\begin{theorem}[Compilation Correctness]\label{thm:compile-good}\rm 
Given $e$, typed as $\tjudge{\iota}{e}{\quan{F}{\hmx}{\iota}}$, and a time period $r$, we compile its simulation $\eexp{\uapp{r}{e}}$ to quantum circuit, as $\quan{F}{\hmx}{\iota} \vdash (e,r)  \gg U : \quan{F}{\umx}{\iota'}$, for every state $\psi$, typed as $\iota \vdash \psi$, we transform it to $\psi'$, typed as $\iota' \vdash \psi'$, let $\psi_1=\denote{\tjudge{\iota}{U}{\quan{F}{\umx}{\iota}}}_0(\psi)$, we transform $\psi_1$ to $\psi'_1$, typed as $\iota' \vdash \psi'_1$; the error is bound by $\dabs{\psi_1' -\denote{\tjudge{\iota'}{\eexp{e r}}{\quan{F}{\umx}{\iota'}}}_0(\psi')}<\epsilon$. 
\end{theorem}

\begin{proof}
    Fully mechanized proofs are done by composing the \Cref{thm:particle-trans,thm:canonical-good,thm:trotter-good,thm:decompose-good,thm:analog-decompose-good}.
\end{proof}

The theorem states that the distance of the compiled circuit $U$ and the semantic definition $\eexp{r e}$ is less than a threshold $\epsilon$, defined as the error bound in Trotterization in \Cref{thm:trotter-good}. The proof combines \Cref{thm:particle-trans,thm:canonical-good,thm:trotter-good} and one of the \Cref{thm:decompose-good,thm:analog-decompose-good}, for compiling the simulation to digital and analog machines, respectively.
The compilation correctness theorem, including the perturbative gadget algorithm, is given in \Cref{appx:gadget}.

\ignore{

We now discuss the target machine Hamiltonian fitting step, mainly for analog-based compilation but can also appear in the circuit-based compilation, i.e., the main goal of the pulse synthesis step in a circuit-based compilation is to fit target machine Hamiltonians \cite{10.1145/3632923,Li_2022,Tacchino_2019}.
As we mentioned in \Cref{sec:sim-gates}, gates are simulated in the underlying quantum machines because they are essentially Hamiltonian simulators.
Instead of performing quantum gate syntheses, one can generate good small Hamiltonians that can be directly implemented as specific target machine Hamiltonians and perform machine-level simulations.
For example, an IBM machine is based on the Heisenberg system \cite{Auerbach1998-jd} and permits the following machine Hamiltonian between any of the two adjacent sites $j$ and $j\splus 1$.

{\small
\[
\hat{H}_{\cn{IBM}} = z_1\sapp{Z(j)}{X(j\splus 1)} + z_2\sapp{Z(j)}{Z(j\splus 1)}+z_3 Z(j) + z_4 X(j\splus 1)
\]
}

As we mentioned in \Cref{sec:sim-gates,sec:qcompile}, quantum gates are simulated based on target machine-level Hamiltonians, and simple quantum circuit generation might not be an effective way of implementing a program.
There are several solutions. First, we can rewrite the quadruple $Z$ interaction for IBM machines to utilize more $ZZ$ or $XZ$ interactions built into the machines, such as in \Cref{fig:zzzz-gate}. 
Second, there are more effective bosonic system approximations \cite{peng2023quantum}, which can reduce the quadruple site interactions to two-site ones, such as $ZZ$ interactions, significantly optimizing the gate numbers in the simulation.
In fact, \citet{Cao_2015} approximated a higher-order site interaction with the linear sum of two or three site interactions.
In addition, other types of quantum machines support more forms of interactions, some of which allow more than two-site interactions \cite{Kyprianidis_2024}; such systems can better utilize quantum computers.

\myparagraph{Energy State Calculation.}\label{sec:groundenergy}
Another key application is the energy state computation based on expectation value calculation in \Cref{sec:background}.
One of the central pieces in energy state computation is to find the exact ground state energy expectation value \(E_0\) and the state $\psi$ holding the energy value.
Based on the variational principle, the expectation value of the Hamiltonian over all possible states $\psi$ is minimized:

%\vspace{-0.3em}
{\small
\begin{center}
$
{\displaystyle E_0 = \min_{\cn{nor}(\psi)} \sapp{\sdag{\cn{nor}(\psi)}}{\sapp{\hat{H}}{\cn{nor}(\psi)}}}
$
\end{center}
}
%\vspace{-0.3em}

The exact ground state can be obtained by minimizing this expectation value over all normalized states.
In \qsnd, we compile the computation to HPC software through classical optimization techniques, such as gradient descent methods, to ``guess'' the minimal expectation value and its associated state of a given Hamiltonian.
Many HPC particle computation platforms use other sophisticated techniques based on the guessing-out strategy. For example, Gamess \cite{540141112,Zahariev2023} and NWChem \cite{10.1063/5.0004997} are HPC software to compute energy states of chemical molecule structures, and users can define the structures in second quantization in the software.
The quantum supremacy of the energy state calculation problems is murky, as some researchers \cite{Lee2023} do not believe in the existence of such supremacy.
Nevertheless, the typical way to calculate in quantum computers is through the variational quantum eigensolver concept \cite{TILLY20221}, having similar procedures as the HPC counterpart.

}

%\input{case}
%\input{compilation}
%\input{evaluation}
%\vspace{-0.3em}
\section{Related Work}
\label{sec:related}

This section reviews related work most relevant to our approach. While many other contributions exist, we focus on those directly connected to the semantics, type systems, and compilation techniques discussed in this paper.

\noindent{\textbf{\emph{Second Quantization.}}}
Second quantization, or occupation number representation, is a mathmetical system used to describe and analyze quantum particle systems.
The key ideas of the methodology were introduced by Paul Dirac \cite{Dirac1988} and were later developed by Pascual Jordan \cite{Jordan1928} and Vladimir Fock \cite{Fock1932,Reed:1975uy}. 
In this approach, the quantum states are represented in the Fock state basis, which is constructed by filling up each single-site state with a certain number of identical particles \cite{Becchi10}. Later, many works were developed for describing second quantization in analyzing different systems, such as fermions in lattice-based systems \cite{Levin:2003ws},
chemical and biological molecule systems \cite{Gori2023}, nuclear physical systems \cite{Johnson_2013,Christiansen2004}.
Recently, researchers developed algebraic formalisms based on second quantization and Lie algebra \cite{Batista_2004,SCHWARZ2021115601}
as a new extension of second quantization, in order to define the transformations from one particle system to another,
such as defining a generalization of Jordan-Wigner Transformation \cite{Batista_2001} to transform different particle systems to $t(2)$ typed particle system, a.k.a., quantum computer systems.
\citet{ying2014quantum} developed a quantum version of recursion and program control flow based on creators and annihilators in second quantization.

The difference between these works and \qsnd is that \qsnd examines the programming aspects of second quantization as we identify its constraint, dynamic, and application semantics, with a certified compiler.

\noindent{\textbf{\emph{Quantum Circuit-based Programming Languages.}}}
\ignore{
Quantum circuit-based language developments are focal.
Q\# \cite{Svore_2018}, Quilc \cite{smith2020opensource}, ScaffCC \cite{JavadiAbhari_2015}, Project Q \cite{Steiger_2018}, Criq \cite{cirq_developers_2023_10247207}, Qiskit \cite{Qiskit2019} are industrial quantum circuit languages, without formal semantics, but based on a standardized denotational quantum circuit semantics.
There are formally verifying quantum circuit programs, including  \qwire~\cite{RandThesis}, \sqir~\cite{PQPC}, and \qbricks~\cite{qbricks}, quantum Hoare logic and its subsequent works \cite{qhoare,qhoreusage,10.1145/3571222}, Qafny \cite{li2024qafny}. These tools have been used to verify a range of quantum algorithms, from Grover's search to quantum phase estimation.
There are works verifying quantum circuit optimizations  (e.g., \voqc~\cite{VOQC}, CertiQ~\cite{Shi2019}), as well as verifying quantum circuit compilation procedures,
including ReVerC~\cite{reverC} and ReQWIRE~\cite{Rand2018ReQWIRERA}.

There are functional language interpretations for quantum circuits \cite{10.1007/978-3-319-89366-2_19,DIAZCARO2019104012,Selinger2013,van_Tonder_2004,Altenkirch,10.1007/978-3-642-40922-6_5,Voichick_2023,Green2013,mingsheng-ying,qml-update,qml-thesis,10.1145/3571204}, with formal semantics.
These languages are mainly designed to encode quantum circuits as functional language constructs, such as quantum pattern matching.
Silq \cite{sliqlanguage} proposed a language to perform automatic uncomputation with denotational semantics for describing circuit behaviors.
ZX-calculus \cite{Coecke_2011,wang2023completeness} turns elements appearing in quantum semantics into a diagram calculus for expressing unitary operations without measurement.
}
Quantum circuit-based language development is a central focus in the field. Several industrial quantum circuit languages have been developed, including Q\#~\cite{Svore_2018}, Quilc~\cite{smith2020opensource}, ScaffCC~\cite{JavadiAbhari_2015}, ProjectQ~\cite{Steiger_2018}, Cirq~\cite{cirq_developers_2023_10247207}, and Qiskit~\cite{Qiskit2019}. While these languages lack formal semantics, they are grounded in a standardized denotational semantics for quantum circuits. In contrast, a number of tools focus on formal verification of quantum circuit programs. These include \qwire~\cite{RandThesis}, \sqir~\cite{PQPC}, and \qbricks~\cite{qbricks}, as well as quantum Hoare logic and its extensions~\cite{qhoare,qhoreusage,10.1145/3571222}, and Qafny~\cite{li2024qafny}. These systems have been used to verify a range of quantum algorithms, including Grover’s search and quantum phase estimation. Verification has also been applied to quantum circuit optimizations, such as in \voqc~\cite{VOQC} and CertiQ~\cite{Shi2019}, and to quantum circuit compilation processes, as demonstrated by ReVerC~\cite{reverC} and ReQWIRE~\cite{Rand2018ReQWIRERA}.

Several functional languages have been developed to model quantum circuits with formal semantics, including~\cite{10.1007/978-3-319-89366-2_19,DIAZCARO2019104012,Selinger2013,van_Tonder_2004,Altenkirch,10.1007/978-3-642-40922-6_5,Voichick_2023,Green2013,mingsheng-ying,qml-update,qml-thesis,10.1145/3571204}. These languages are typically designed to encode quantum circuits as functional constructs, such as through quantum pattern matching. Silq~\cite{sliqlanguage} introduces automatic uncomputation within a denotational semantics framework for describing circuit behavior. Finally, the ZX-calculus~\cite{Coecke_2011,wang2023completeness} provides a diagrammatic language for expressing unitary operations without measurement, offering an alternative to algebraic quantum semantics.

\noindent{\textbf{\emph{Analog Quantum Computing Languages and Pulse Level Programming.}}}
\ignore{
There are a few pulse-level and target machine Hamiltonian programming interfaces that view quantum computers as analog quantum simulators,
such as IBM Qiskit Pulse \cite{10.1145/3505636}, QuEra Bloqade \cite{bloqade2023quera}, and Pasqal
Pulser \cite{Silverio2022pulseropensource} was developed by hardware service providers. 
There are computational quantum physics packages, e.g., QuTip \cite{JOHANSSON20121760}, supporting Hamiltonian simulation on a classical computer.
These works are inspired by many classical analog compilations \cite{10.1145/3373376.3378449,10.1145/2980983.2908116}
SimuQ \cite{10.1145/3632923} provides a formal language on analog quantum computing based on constructing a Hamiltonian with Pauli groups.
\citet{cao2024marqsimreconcilingdeterminismrandomness} provides a better Trotterization algorithm than QDrift.
Compared to these works, \qsnd provides the first verified compilation framework from a user application second quantization Hamiltonian simulation to digital and analog quantum circuits.
}
Several pulse-level and target-machine Hamiltonian programming interfaces treat quantum computers as analog quantum simulators. Notable examples include IBM Qiskit Pulse~\cite{10.1145/3505636}, QuEra Bloqade~\cite{bloqade2023quera}, and Pasqal Pulser~\cite{Silverio2022pulseropensource}, all developed by hardware service providers. In parallel, computational quantum physics packages such as QuTiP~\cite{JOHANSSON20121760} support Hamiltonian simulation on classical computers. These efforts are inspired by a long history of classical analog compilation techniques~\cite{10.1145/3373376.3378449,10.1145/2980983.2908116}. SimuQ~\cite{10.1145/3632923} introduces a formal language for analog quantum computing based on the construction of Hamiltonians using Pauli groups. Recent work by~\citet{cao2024marqsimreconcilingdeterminismrandomness} proposes a Trotterization algorithm that improves upon QDrift.

Compared to these approaches, \qsnd provides the first verified compilation framework that transforms user-level Hamiltonian programs, written in second quantization form, into both digital and analog quantum circuits.

\noindent{\textbf{\emph{Quantum computational Approaches for Analyzing Particle Systems.}}}
Many quantum computation software and algorithms have been proposed for analyzing particle systems.
%Many works utilize HPC systems \cite{Zahariev2023,540141112,10.1063/5.0004997,6651029} to calculate molecule energy states, based on some variational methods.
There are software tools \cite{10.1145/3511715,10.1145/3503222.3507715,POWERS2021100696,PhysRevA.109.042418} for performing particle system Hamiltonian simulation, by viewing it as a quantum algorithm and compiling it to quantum circuits.
Many other works \cite{Cervera_Lierta_2018,Yang_2020,YAMAGUCHI2002343,doi:10.1021/acs.jctc.2c00974} discuss simulating different particle systems in quantum computers.
The variational quantum eigensolver is one of the most well-known algorithms \cite{Tilly_2022,Cerezo2022} to analyze particle systems.
It is mainly used for performing energy state computation, such as estimating the ground energy state of water molecules \cite{Nam2020}.
\section{Conclusion and Future Work}\label{sec:conclusion}
%\vspace{-0.5em}

We present \qsnd, the first verified compilation framework for second-quantization-based Hamiltonian simulation, aiming to define and analyze lattice-based particle systems.
%We generalize quantum computer systems to be a specific quantum particle system, so we can view the compilation of quantum particle systems to quantum computers as mapping from one kind of particle system to another, definable in \qsnd.
We examine the programmability of the second quantization formalism and identify that second quantization programs are constraints describing a system, and the Hamiltonian simulation is the dynamic semantics of the constraint-based system.
We then show the \qsnd syntax, semantics, and type system, a generalization of the qubit-based system, to permit the definition of different particle systems and classical simulation problems.
Via the typing and equational properties, we show the verified compilation procedure of compiling the Hamiltonian simulation for \qsnd programs to quantum circuits.
%Under the \qsnd generalization, we intend to connect the previously unlinked dots, fragmented in different layers, together, so one can utilize \qsnd to define the tasks in these dots, such as user applications, quantum particle system transformations, and target machine specifications.
%Moreover, we enable the definitions of many quantum computation models, such as circuit-based models, analog Hamiltonian simulations, and adiabatic evolution.
%We provide typing information for different particle systems, which can be used to properly program a system and transform between different systems.
%For example, the exponential and logarithm operations can be analyzed by the least fixed point analysis tools, as well as different Hamiltonian analytical tools can be imposed to cut a large Hamiltonian into smaller ones that can be directly implemented in target machine Hamiltonians, such as the example in \Cref{sec:boson}.
%For example, the exponential and logarithm operations can be analyzed by the least fixed point analysis tools, as well as different Hamiltonian analytical tools can be imposed to cut a large Hamiltonian into smaller ones that can be directly implemented in target machine Hamiltonians, such as the example in \Cref{sec:boson}.
%Finally, we plan to fully investigate the commutation and anti-commutation properties of bosons and fermions so that \qsnd can be used to analyze real-world particle systems.

\myparagraph{Future Work.}
We intend to certify many different Hamiltonian-level quantum optimization algorithms and develop tools for analyzing Hamiltonian-based programs.
Another one is to develop different application semantics, especially for computing ground energy states useful in scientific computation, i.e., we can view the application semantics as computing a ground energy state. One of the previous frameworks, which performs the task and naturally extends from Hamiltonian simulation, is adiabatic quantum computation (AQC).
This is a model of quantum computation that solves optimization problems by evolving a quantum system slowly from an initial Hamiltonian $\hat{H}_m$ with a known ground state to a final Hamiltonian $\hat{H}_p$ whose ground state encodes the solution.
The \qsnd application semantics can be naturally extended to AQC for energy state computation.

%\section{Data-Availability Statement}
%The Rocq proofs mentioned in the paper are available on Zenodo: \url{https://zenodo.org/records/15073601}, and will be submitted for evaluation. The artifact is produced for the mechanized proof of the lemmas and theorems in the paper.

\ignore{
By maintaining the system's lowest energy state throughout the process, AQC leverages the adiabatic theorem to find optimal solutions efficiently. The first line below shows the time-dependent Hamiltonian $\hat{H}(r)$ that varies with time $r$ from $\hat{H}_m$ to $\hat{H}_p$. The Schr\"odinger equation on the right would give the exact evolution of this Hamiltonian. The actual solution used in AQC is the sequence on the second line, where each step generates a quantum operation via a Hamiltonian simulation ($\eexp{\uapp{\Delta r}{\hat{H}(r)}}$) for a small time period $\Delta r$. 
In theory, the AQC procedure can be understood as starting with an initial quantum state $\ket{\varphi(0)}$ and applying a sequence of unitary quantum operations, each of which is of the form $\eexp{\uapp{{\Delta r}}{\hat{H}(r)}}$ for a different $t\in\{\Delta r, 2\Delta r, 3 \Delta r, ...,T\}$. 

{\small
\begin{center}
$
\begin{array}{l}
\hat{H}(r)=(1-\frac{r}{T})\hat{H}_m + \frac{r}{T}\hat{H}_p
\qquad
\qquad
i\frac{d}{dt}\ket{\varphi(r)}=\hat{H}(r)\ket{\varphi(r)}
\\[0.4em]
\ket{\varphi(T)}\approx \eexp{\uapp{\Delta r}{\hat{H}(T)}} \circ ...\circ \eexp{\uapp{\Delta r}{\hat{H}(2\Delta r)}}\circ\eexp{\uapp{\Delta r}{\hat{H}(\Delta r)}}\ket{\varphi(0)}
\end{array}
$
\end{center}
}

We conceptually define a time-dependent Hamiltonian $\hat{H}(r)$ that evolves from $\hat{H}(0) = \hat{H}_m$ to $\hat{H}(1) = \hat{H}_p$ as a function of a normalized time parameter $r \in [0,1]$. For example, a linear schedule would use $\hat{H}(r) = (1-r) \hat{H}_m + r \hat{H}_p$. Initially ($r=0$) the system is in the ground state (eigenstate) of $\hat{H}_m$ (the initial superposition state). As $r$ increases, the Hamiltonian gradually incorporates the problem constraints. If the change is slow enough, the quantum adiabatic theorem guarantees that the state of the system will remain in the instantaneous ground state of $\hat{H}(r)$. By the end ($r=1$), the system will be in the ground state of $\hat{H}_p$, which is exactly the state encoding the optimal solution of our problem.
In practice, this continuous AQC is realized through a sequence of many small discrete steps, each applying a slight alteration to the Hamiltonian and consequently a quantum operation on the state. Eventually, after the full schedule, the system's state ideally encodes the solution in its qubit amplitudes. Finally, we measure the qubits in the computational basis. With high probability (if the adiabatic process succeeded and if decoherence and other errors are minimal), the measurement will collapse the system into the basis state corresponding to the lowest energy, i.e., the optimal solution.
}
%AQC replaces the explicit search or brute-force classical process with a physical process. As the Hamiltonian changes, the quantum state naturally ``searches'' the landscape of possibilities by continuously adjusting the amplitudes of basis-ket states, concentrating more weight on low-energy (better) solutions. By the end, the amplitude is ideally concentrated on the optimal solution. This is analogous to a landscape where the system begins as a liquid spread out over all valleys and then slowly freezes into the lowest valley (the global minimum) due to the energy bias introduced by $\hat{H}_p$. The power of this approach is that it exploits quantum parallelism and tunneling: the state can explore many solutions concurrently and avoid getting stuck in local minima under the right conditions, potentially outperforming certain classical algorithms on hard problems.

%\input{semantics}

\section*{Acknowledgments}
This material is based upon work supported by the National Science Foundation under Grant No. OSI-2435255.

%% Bibliography
\bibliography{reference}

%\end{document}

\newpage
\appendix
%\appendix
\section{In Depth Background}\label{app:paulimap}

This section provides additional background information.

\subsection{Quantum Hardware Circuit Operations}\label{sec:background3}

          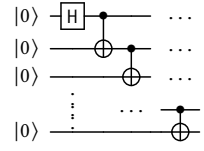
\begin{wrapfigure}{r}{3.5cm}
          %  \vspace*{-0.2em}
            {\qquad
              \footnotesize
              \Qcircuit @C=0.5em @R=0.5em {
                \lstick{\ket{0}} & \gate{\cn{H}} & \ctrl{1} & \qw &\qw & & \dots & \\
                \lstick{\ket{0}} & \qw & \targ & \ctrl{1} & \qw & &  \dots &  \\
                \lstick{\ket{0}} & \qw & \qw   & \targ & \qw & &  \dots &  \\
                & \vdots &   &  &  & & & \\
                & \vdots &  & \dots & & & \ctrl{1} & \qw  \\
                \lstick{\ket{0}} & \qw & \qw & \qw &\qw &\qw & \targ & \qw
              }
            }
            \caption{GHZ Circuit}
            \label{fig:background-circuit-examplea}
          \end{wrapfigure}
          
Computation on a quantum value consists of a series of \emph{quantum operations}, each acting on a subset of qubits in the quantum value. In the standard form, quantum computations are expressed as \emph{circuits}, as in \Cref{fig:background-circuit-examplea}, which depicts a circuit that prepares the Greenberger-Horne-Zeilinger (GHZ) state \cite{Greenberger1989} --- an $n$-qubit entangled value of the form: $\ket{\text{GHZ}^n} = \frac{1}{\sqrt{2}}(\ket{0}^{\otimes n}+\ket{1}^{\otimes n})$, where $\ket{d}^{\otimes n}=\Motimes_{d=0}^{n\sminus 1}\ket{d}$.
In these circuits, each horizontal wire represents a qubit, and boxes on these wires indicate quantum operations, or \emph{gates}. The circuit in \Cref{fig:background-circuit-examplea} uses $n$ qubits and applies $n$ gates: a \emph{Hadamard} (\coqe{H}) gate and $n-1$ \emph{controlled-not} (\coqe{CNOT}) gates. Applying a gate to a quantum value \emph{evolves} it.
Its traditional semantics is expressed by multiplying the value's vector form by the gate's corresponding matrix representation: $n$-qubit gates are $2^n$-by-$2^n$ matrices.
Except for measurement gates, a gate's matrix must be \emph{unitary} and thus preserve appropriate invariants of quantum values' amplitudes.

A \emph{measurement} operation extracts classical information from a quantum value. It collapses the value to a basis state with a probability related to the value's amplitudes (\emph{measurement probability}), e.g., measuring $\frac{1}{\sqrt{2}}(\ket{0} + \ket{1})$ collapses the value to $\ket{0}$ with probability $\frac{1}{2}$,  and likewise for $\ket{1}$, returning classical value $0$ or $1$, respectively. 
A more general form of quantum measurement is \emph{partial measurement}, which measures a subset of qubits in a qubit array;
such operations often have simultaneity effects due to entanglement, \ie{} in a Bell pair $\frac{1}{\sqrt{2}}(\ket{00} + \ket{11})$, measuring one qubit guarantees the same outcome for the other --- if the first bit is measured as $0$, the second bit is too.

\subsection{Connecting Pauli Strings and Second Quantization}\label{app:paulimapa}

There are many interpretations of creators and annihilators in second quantization based on different particle types, which \qsnd classifies them in our type system.
The basic form is the creators and annihilators in \emph{a two-dimensional boson-like particle system} with $t(2)$ typed sites: the lowering $a$ and raising $\sdag{a}$ operations in ladder algebra, with the semantic rules below.
In this system, as shown above, applying a single creator $\sdag{a}$ (having the matrix representation 
$\begin{psmallmatrix}
0 & 0\\
1 & 0
\end{psmallmatrix}$) turning $\ket{0}$ to $\ket{1}$, but it eliminates $\ket{1}$ to be $\zero$,
while applying a single annihilator $a$ (having the matrix representation
$\begin{psmallmatrix}
0 & 1\\
0 & 0
\end{psmallmatrix}$) turning $\ket{1}$ to $\ket{0}$, but it eliminates $\ket{0}$ to be $\zero$.

{\small
\[\hspace{-0.2em}
\begin{array}{c}
\sdag{a}\ket{0}= \ket{1}
\qquad\qquad
\sdag{a}\ket{1}= \zero
\qquad\qquad
a\ket{1}= \ket{0}
\qquad\qquad
a\ket{0}= \zero
\\[0.3em]
\mapp{\sdag{a}}{a}\ket{1}=\ket{1}
\qquad\quad
\mapp{\sdag{a}}{a}\ket{0}=\zero
\qquad\quad
\mapp{a}{\sdag{a}}\ket{0}=\ket{0}
\qquad\quad
\mapp{a}{\sdag{a}}\ket{1}=\zero
\\[0.3em]
{\mapp{\sdag{a}}{a}}\otimes {\mapp{a}{\sdag{a}}}\ket{1}\ket{0}=\ket{1}\ket{0}
\qquad
{\mapp{\sdag{a}}{a}}\otimes {\mapp{a}{\sdag{a}}}\ket{b_1}\ket{b_2}=\zero\;\;\text{if}\;b_1\neq 1 \vee b_2 \neq 0
\end{array}
\]
}

We define the $\mathbb{0}$ and $\mathbb{1}$ operations in \Cref{sec:intro} by composing creators and annihilators,
i.e., $\mapp{\sdag{a}}{a}$, defining $\mathbb{1}$, preserves $\ket{1}$, and $\mapp{a}{\sdag{a}}$, defining $\mathbb{0}$, preserves $\ket{0}$, as shown in the second line above.
As an instance in a two qubit system, $\frac{1}{\sqrt{2}}(\ket{0}\ket{1} + \ket{1}\ket{0})$, applying $\mathbb{1}$ to the second site results the state $\frac{1}{\sqrt{2}}\ket{0}\ket{1}$ \footnote{This is the constraint semantics not the quantum dynamic unitary semantics, so operations are not necessarily unitary.}, by preserving $\ket{1}$ in the second site.
The composition effect can also deal with a composed state of multiple qubits, such as the third line above, with transitions below.

{\small
\begin{center}
$
{{\mapp{\sdag{a}}{a}}\otimes {\mapp{a}{\sdag{a}}}\ket{1}\ket{0}}
=
{{\mapp{\sdag{a}}{a}}\otimes {\mapp{a}{\sdag{a}}} (\ket{1} \otimes \ket{0})}
=
{{\mapp{\sdag{a}}{a}}\ket{1}\otimes {\mapp{a}{\sdag{a}}} \ket{0}}
={\ket{1}\ket{0}}
$
\end{center}
}

We can see that Pauli strings are a special case of tensor strings, as they can be constructed by performing the above two-dimensional lowering and raising operations. The identity operation $I$ below is assumed to be two-dimensional.

{\small
\[
\begin{array}{l@{\qquad\quad}l@{\qquad\quad}l@{\qquad\quad}l}
I \triangleq  
\begin{bmatrix}
1 & 0\\
0 & 1
\end{bmatrix}
&
X \triangleq  
\begin{bmatrix}
0 & 1\\
1 & 0
\end{bmatrix}
&
Y \triangleq  
\begin{bmatrix}
0 & \sminus i\\
i & 0
\end{bmatrix}
&
Z \triangleq  
\begin{bmatrix}
1 & 0\\
0 & \sminus 1
\end{bmatrix}
\\[1em]
I = \sapp{a}{\sdag{a}} + \sapp{\sdag{a}}{a} = \mathbb{0} + \mathbb{1}
&
X = \sdag{a} + a
&
Y = i a - i \sdag{a}
&
Z = \sapp{a}{\sdag{a}} - \sapp{\sdag{a}}{a} =\mathbb{0} -  \mathbb{1}
\end{array}
\]
}

The second line above defines Pauli operations based on the two-dimensional lowering and raising operations above, with the matrix forms of the Pauli operations listed as the first line above.
Here, \(i\) is a complex amplitude value, and \(i a - i \sdag{a}\) is an abbreviation of \(\uapp{i}{a} + (\uapp{(\sminus i)}{\sdag{a}})\). 
Note that the meaning of Pauli groups in terms of second quantization differs from the gates in a quantum circuit, which provides a way of defining particle constraints by describing particle movements as spin rotations across \(X\), \(Y\), and \(Z\) bases.

\subsection{Linear Sums and Quantum Choices}\label{appx:linearsum}

A linear sum in a state connects kets in \Cref{fig:data} represents superposition state; the sum operation \(e_1 + e_2\) is reinterpreted a \emph{quantum choice} operation, in the sense that one can perform either \(e_1\) or \(e_2\) to the state \(\psi\) at the same time,
e.g., it might produce a superposition if \(\psi\) is a state just initialized as a ket.

When second-quantization users employ quantum choices to constrain a lattice-based molecular system, they typically have two intentions.
The \emph{first intention} is 
%similar to the choice selection in the equal-sum example above, and we intend 
to "create" a state going in two different directions simultaneously in a possible superposition.
For example, applying a quantum choice operation $a + \sdag{a}$ on a $t(4)$ particle state $\ket{1}$, with state \emph{normalization}, results in the following:

{\small
\begin{center}
$\frac{1}{\sqrt{2}}\ket{0}+\frac{1}{\sqrt{2}}\ket{2}$
\end{center}
}

In the Hubbard system $\hat{H}$, for each $j$-th site, the sum of the two terms, $\sdag{a}(j) \circ a(j\splus 1)$ and $\sdag{a}(j\splus 1) \circ a(j)$, model all the possibilities of electrons, being relocated to adjacent particle site. 
Normalization means that the norm of the amplitudes in different kets must be summed to $1$, referring to the total probabilities of different parts.
State normalization permits Hamiltonian simulation and quantum computing, while the non-normalized state plays a role in computing expectation values.
In the extended version of \qsnd, we permit a \cn{nor} operation to normalize a state. Quantum simulation is essentially defined as $\uapp{\eexp{\uapp{r}{\hat{H}}}}{\psi(0)}$ for arbitrary state $\psi(0)$, applying $\eexp{\uapp{r}{\hat{H}}}$ to a normalized state $\psi(0)$.
We show a theorem in \Cref{sec:theorem} to link simulations on normalized and non-normalized states.

The \emph{second intention} is to add repulsion, i.e., we view a quantum choice term as some extra force applied to the system to increase or decrease the probability of observing a basis vector state.
Recall that the normalization procedure ensures that a system's total probability is $1$. Conceptually, adding a quantum choice term for certain basis vectors with a positive or negative amplitude increases or reduces the choice of observing the basis vectors, respectively.
For example, the term $z_u \sum_{j} \mathbb{1}(j) \circ  \mathbb{1}(j\splus 1)$, in the Hubbard system (\Cref{sec:hubbard}), usually refers to the repulsion force, reducing the likelihood of the happening of some basis vectors. Let's name the above term as $\hat{H}_u$ and the term $\sminus z_t \sum_{j} (\sapp{\sdag{a}(j)}{ a(j\splus 2)} + \sapp{\sdag{a}(j\splus 2)}{ a(j)})$ as $\hat{H}_t$.
Notice that the complex amplitudes $\sminus z_t$ and $z_u$ have opposite signs.
Applying the whole Hamiltonian to a quantum state $\psi$ results in a quantum state of two pieces: $({\hat{H}_t}{\psi})+({\hat{H}_u}{\psi})$, connected by $+$. A non-zero ket in the state ${\hat{H}_u}{\psi}$ has an amplitude with an opposite sign as the same basis-vector amplitude in ${\hat{H}_t}{\psi}$, which means that the sum of their amplitudes is reduced; thus, the quantum choice term $\hat{H}_u$ acts as a penalty for a specific group of ket states.  
To see how this happens, realize that the Hubbard system describes $t(2)$ typed particles, so there are four cases of the basis vectors of two adjacent particles:

{\small
\begin{center}
$\dualb{}{0}{0} \qquad\qquad\dualb{}{0}{1} \qquad\qquad \dualb{}{1}{0} \qquad\qquad \dualb{}{1}{1}$
\end{center}
}

Assuming that we are dealing with a single particle site, the $\hat{H}_u$ term turns all basis vectors, except $\dualb{}{1}{1}$, to $\zero$, where $\dualb{}{1}{1}$ actually represents two electrons in opposite spins.
Assume that a single $t(2,2)$ typed particle state $\varphi$ contains the ket $\dualb{z}{1}{1}$, applying $\hat{H}_u$ to a state $\varphi$ results in a ket $\dualb{z_u * z}{1}{1}$.
When applying $\hat{H}_t$ to $\varphi$, we expect to have another ket $\dualb{\sminus z_t * z}{1}{1}$, with the same basis vector as above.
The sum of the two sub-parts turns the amplitude of the basis vector $\dualb{}{1}{1}$ to become $z*(z_u-z_t)$, with $\slen{z*(z_u-z_t)}\le\slen{z}$.
This is why the additional repulsion term acts as a penalty on the ket $\dualb{}{1}{1}$ to reduce its probability of happening.

\subsection{Energy Expectation Value Calculation.}\label{sec:groundenergy}

Other than Hamiltonian simulation, another key application is the energy state computation based on expectation value calculation, described above.
One of the central pieces in energy state computation is to find the exact ground state energy expectation value \(E_0\) and the state $\psi$ holding the energy value.
Based on the variational principle, the expectation value of the Hamiltonian over all possible states $\psi$ is minimized:

%\vspace{-0.3em}
{
\begin{center}
$
{\displaystyle E_0 = \min_{\psi} \uapp{\sdag{\psi}}{(\uapp{\hat{H}}{\psi}})}
$
\end{center}
}
%\vspace{-0.3em}

The exact ground state can be obtained by minimizing this expectation value over all possible normalized states.
In \qsnd, we compile the computation to HPC software through classical optimization techniques, such as gradient descent methods, to ``guess'' the minimal expectation value and its associated state of a given Hamiltonian.
Many HPC particle computation platforms use other sophisticated techniques based on the guessing-out strategy. For example, Gamess \cite{540141112,Zahariev2023} and NWChem \cite{10.1063/5.0004997} are HPC software to compute energy states of chemical molecule structures, and users can define the structures in second quantization in the software.
The quantum supremacy of the energy state calculation problems is murky, as some researchers \cite{Lee2023} do not believe in the existence of such supremacy.
Nevertheless, the typical way to calculate in quantum computers is through the variational quantum eigensolver concept \cite{TILLY20221}, having similar procedures as the HPC counterpart.

\subsection{Direct Sum and Quantum State and Hilbert spaces}\label{appx:directsum}

\Cref{sec:staterep} provides a Hilbert space state representation for \qsnd, this section expands the background of constructing high dimensional Hilbert spaces based on the direct sum concept.

Assuming that readers are familiar with \emph{inner product spaces}, vector spaces equipped with an inner product \cite[p.~167]{linear-algebra}.
A finite-dimensional \emph{Hilbert} space is simply a finite-dimensional inner product space \cite[p.~66]{mike-and-ike}.
The definition of Hilbert space is restricted in the infinite-dimensional case, but we are only concerned with finite dimensions in this work.
All of our Hilbert spaces will be denoted with the letter $\Hilb$ and will be over the field $\complex$ of complex numbers.

Dirac Bra-ket notation is a convenient notation for describing vectors and operators in Hilbert spaces.
A vector is written using ket notation, writing ``$\ket{Q}$'' rather than ``$\vec{x}_Q$'' used in other areas of mathematics.
The inner product of any vectors $\ket{\phi}$ and $\ket{\psi}$ is written $\ip{\phi}{\psi}$, and is linear in the second argument $\ket{\psi}$.
The outer product, written $\op{\phi}{\psi}$, is the linear operator defined such that $\ket{\phi}\!\!\ip{\psi}{\psi'} = \ip{\psi}{\psi'} \cdot \ket{\phi}$, where ``$\cdot$'' denotes scalar multiplication.
It is often useful to treat a ``ket'' $\ket{\psi}$ as a \emph{column vector} and a ``bra'' $\bra{\psi} = \ket{\psi}^\dagger$ as a \emph{row vector} (where ``$\dagger$'' denotes the conjugate transpose), and then both the inner and outer product are simply matrix multiplication.
To distinguish a quantum state with computational basis-vector, \qsnd write quantum states $\ket{\psi}$ and $\ket{\phi}$ as $\psi$ and $\phi$, we denote row vectors as $\sdag{\phi}$ and $\sdag{\psi}$. The inner product is written as $\phi \, \sdag{\psi}$.

%The trivial Hilbert space is the set $\{0\}$, a zero-dimensional Hilbert space with only one vector, the zero vector ($\zero$). The set $\complex$ forms a one-dimensional Hilbert space over the complex numbers, spanned by the one-element orthonormal basis $\{1\}$. The inner product of two complex numbers $z_1$ and $z_2$ is $z_1^{*} z_2$.

\begin{figure}[t]
{\small
\begin{align*}
	\ket{\phi} &\triangleq \begin{bmatrix} a \\ b \end{bmatrix} & \ket{\psi} &\triangleq \begin{bmatrix} c \\ d \\ e \end{bmatrix} \\
	\ket{\phi} \oplus \ket{\psi} &= \begin{bmatrix} a \\ b \\ c \\ d \\ e \end{bmatrix} & \ket{\phi} \otimes \ket{\psi} &= \begin{bmatrix} a c \\ a d \\ a e \\ b c \\ b d \\ b e \end{bmatrix} \\
\end{align*}
\[
\begin{aligned}
\begin{aligned}
	E\subcap{a} &\triangleq \begin{bmatrix} f & g & h & i \\ j & k & l & m \\ n & o & p & q \end{bmatrix} \\ E\subcap{b} &\triangleq \begin{bmatrix} r & s \\ t & u \\ v & w \end{bmatrix} \\
\end{aligned}
\qquad
	E\subcap{a} \oplus E\subcap{b} = \begin{bmatrix} f & g & h & i & 0 & 0 \\ j & k & l & m & 0 & 0 \\ n & o & p & q & 0 & 0 \\ 0 & 0 & 0 & 0 & r & s \\ 0 & 0 & 0 & 0 & t & u \\ 0 & 0 & 0 & 0 & v & w \end{bmatrix} \\
\end{aligned}
\]
\vspace{5mm}
\[
	E\subcap{a} \otimes E\subcap{b} = \begin{bmatrix} f r & f s & g r & g  s & h  r & h  s & i  r & i  s \\ f  t & f  u & g  t & g  u & h  t & h  u & i  t & i  u \\ f  v & f  w & g  v & g  w & h  v & h  w & i  v & i  w \\
j r & j s & k r & k  s & l  r & l  s & m  r & m  s \\ j  t & j  u & k  t & k  u & l  t & l  u & m  t & m  u \\ j  v & j  w & k  v & k  w & l  v & l  w & m  v & m  w \\
n r & n s & o r & o  s & p  r & p  s & q  r & q  s \\ n  t & n  u & o  t & o  u & p  t & p  u & q  t & q  u \\ n  v & n  w & o  v & o  w & p  v & p  w & q  v & q  w
	\end{bmatrix}
	\]
}
\caption{Matrix examples for direct sums and tensor products}
\label{fig:dsum-matrix}
\end{figure}

Our algebraic data types in this paper rely on two ways of composing Hilbert spaces: (external) direct sums $\oplus$ and tensor products $\otimes$.
Given two Hilbert spaces $\Hilb\subcap{m}$ and $\Hilb\subcap{n}$ with $\dim(\Hilb\subcap{m}) = M$ and $\dim(\Hilb\subcap{n}) = N$, the direct sum and tensor product are defined so that $\dim(\Hilb\subcap{m} \oplus \Hilb\subcap{n}) = M + N$ and $\dim(\Hilb\subcap{m} \otimes \Hilb\subcap{n}) = M \cdot N$.
These operators apply to Hilbert spaces, as well as vectors and operators within these spaces.
For example, if $\ket{\phi} \in \Hilb\subcap{m}$ and $\ket{\psi} \in \Hilb\subcap{n}$, then $\ket{\phi} \oplus \ket{\psi} \in \Hilb\subcap{m} \oplus \Hilb\subcap{n}$ and $\ket{\phi} \otimes \ket{\psi} \in \Hilb\subcap{m} \otimes \Hilb\subcap{n}$, and if $E\subcap{m} : \Hilb\subcap{m,1} \to \Hilb\subcap{m,2}$ and $E\subcap{n} : \Hilb\subcap{n,1} \to \Hilb\subcap{n,2}$ are linear operators, then $E\subcap{m} \oplus E\subcap{n} : \Hilb\subcap{m,1} \oplus \Hilb\subcap{n,1} \to \Hilb\subcap{m,2} \oplus \Hilb\subcap{n,2}$ is also a linear operator, as well as $E\subcap{m} \otimes E\subcap{n} : \Hilb\subcap{m,1} \otimes \Hilb\subcap{n,1} \to \Hilb\subcap{m,2} \otimes \Hilb\subcap{n,2}$.
%Following standard notational conventions, we sometimes write $\ket{\phi, \psi}$ for $\ket{\phi} \otimes \ket{\psi}$.

The tensor product will be familiar to anyone with a basic understanding of quantum computing.
It is used to describe \emph{joint states}.
If a qubit inhabits a two-dimensional Hilbert space $\Hilb_2$, and a qutrit inhabits a three-dimensional Hilbert space $\Hilb_3$, then the six-dimensional Hilbert space $\Hilb_2 \otimes \Hilb_3$ describes a physical system containing a qubit and a qutrit, possibly entangled.
In the context of programming languages, tensor products correspond nicely to \emph{product types}.

Qubit arrays serve as a sufficient physical model for most of the interesting results in quantum computing.
An $n$-qubit system is described by a $2^n$-dimensional Hilbert space, which is what allows quantum computing to achieve its famed exponential advantages.
For a programmer, however, this can be fairly limiting for expressiveness, effectively restricting one's data types to those whose cardinality is a power of two.
To allow for arbitrary finite data types, the direct sum is useful, serving as an ``\textsc{or}'' where the tensor product serves as ``\textsc{and}.''
The five-dimensional Hilbert space $\Hilb_2 \oplus \Hilb_3$ describes quantum states that are \emph{either} a qubit \emph{or} a qutrit, or a superposition of the two.

These constructions may be more intuitive given matrix representations.
The direct sum of two vectors is their concatenation, while the direct sum of two matrices is their block diagonalization, as shown by the examples in Figure~\ref{fig:dsum-matrix}.

Suppose Hilbert space $\Hilb\subcap{m}$ is spanned by an orthonormal basis $B\subcap{m}$, and Hilbert space $\Hilb\subcap{n}$ is spanned by an orthonormal basis $B\subcap{n}$.
Then $\Hilb\subcap{m} \oplus \Hilb\subcap{n}$ is spanned by an orthonormal basis $\{\ket{u} \oplus 0 : \ket{u} \in B\subcap{m}\} \cup \{ 0 \oplus \ket{v} : \ket{v} \in B\subcap{n}\}$, and $\Hilb\subcap{m} \otimes \Hilb\subcap{n}$ is spanned by an orthonormal basis $\{\ket{u} \otimes \ket{v} : \ket{u} \in B\subcap{m}, \ket{v} \in B\subcap{n} \}$.
The following identities may be useful for understanding operations on these spaces.

\begin{align*}
	\left( \ket{\phi_\zero} \oplus \ket{\psi_\zero} \right) + \left( \ket{\phi_\one} \oplus \ket{\psi_\one} \right) &= \left( \ket{\phi_\zero} + \ket{\phi_\one} \right) \oplus \left( \ket{\psi_\zero} + \ket{\psi_\one} \right) \\
	z \left( \ket{\phi} \oplus \ket{\psi} \right) &= z \ket{\phi} \oplus z \ket{\psi} \\
	\left( \ket{\phi} \oplus \ket{\psi} \right)^\dagger &= \bra{\phi} \oplus \bra{\psi} \\
	\left( \bra{\phi_\zero} \oplus \bra{\psi_\zero} \right) \left( \ket{\phi_\one} \oplus \ket{\psi_\one} \right) &= \ip{\phi_\zero}{\phi_\one} + \ip{\psi_\zero}{\psi_\one} \\
	(E\subcap{m} \oplus E\subcap{n}) \left(\ket{\phi} \oplus \ket{\psi} \right) &= E\subcap{m} \ket{\phi} \oplus E\subcap{n} \ket{\psi} \\
	z \ket{\phi} \otimes \ket{\psi} &= \ket{\phi} \otimes z \ket{\psi} \\
																	&= z \left( \ket{\phi} \otimes \ket{\psi} \right) \\
	\left( \ket{\phi} \otimes \ket{\psi} \right)^\dagger &= \bra{\phi} \otimes \bra{\psi} \\
	\left( \bra{\phi_\zero} \otimes \bra{\psi_\zero} \right) \left( \ket{\phi_\one} \otimes \ket{\psi_\one} \right) &= \ip{\phi_\zero}{\phi_\one} \cdot \ip{\psi_\zero}{\psi_\one} \\
	(E\subcap{m} \otimes E\subcap{n}) \left(\ket{\phi} \otimes \ket{\psi} \right) &= E\subcap{m} \ket{\phi} \otimes E\subcap{n} \ket{\psi}
\end{align*}

%Category theory can help describe the relationship between the direct sum and tensor product.
%Together, they make the category of finite-dimensional vector spaces and linear operators into a \emph{symmetric bimonoidal category} \cite[p.~130]{symmetric-bimonoidal}, also known as a ``rig category.''
%This means that the category of finite-dimensional vector spaces is a monoidal category with respect to both the direct sum and the tensor product, and there are isomorphisms with some well-characterized properties similar to what one would expect from addition and multiplication.
%In particular, there is a distributivity isomorphism from $\Hilb\subcap{a} \otimes \left(\Hilb\subcap{b} \oplus \Hilb\subcap{c}\right)$ to $\Hilb\subcap{a} \otimes \Hilb\subcap{b} \oplus \Hilb\subcap{a} \otimes \Hilb\subcap{c}$, mapping $\ket{a} \otimes \left( \ket{b} \oplus \ket{c} \right)$ to $\ket{a} \otimes \ket{b} \oplus \ket{a} \otimes \ket{c}$.

%Qunity's semantics is defined in terms of norm-non-increasing operators, operators $E$ such that $\mathbb{I} - E^\dagger E$ is a positive operator.
%It is not hard to see that these operators are closed under composition as well as the direct sum and tensor product, meaning that there is also a symmetric bimonoidal category %of finite-dimensional Hilbert spaces and norm-non-increasing operators.
%This is the category on which Qunity's pure denotational (categorical) semantics is defined.

\ignore{
\begin{figure}[t]
\vspace*{-0.5em}
    \includegraphics[width=0.6\textwidth]{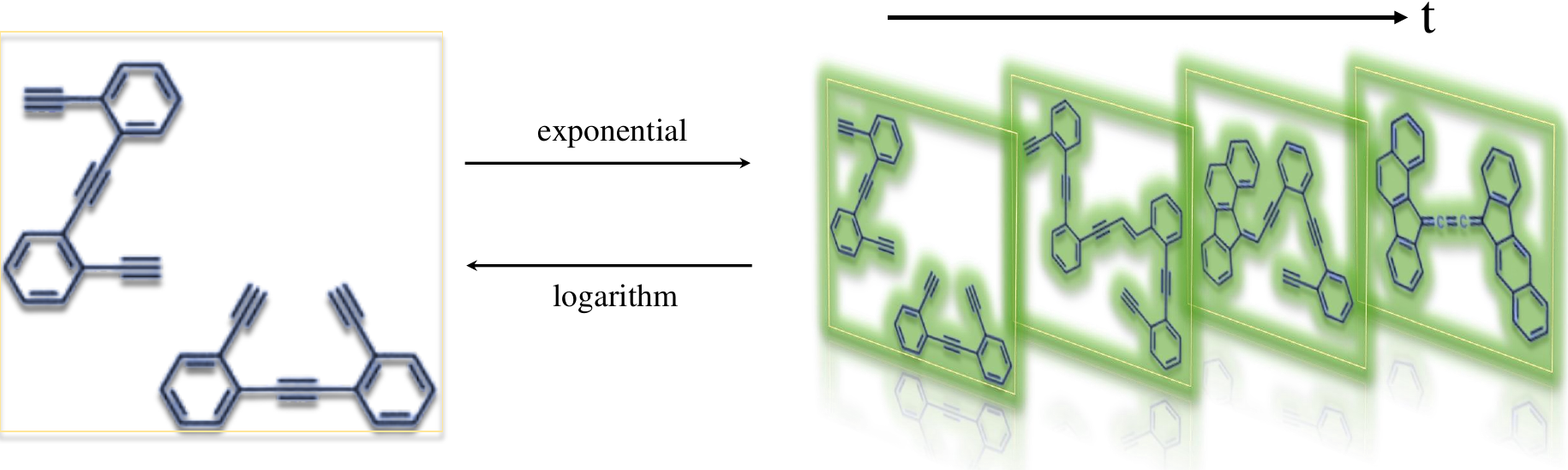}
    \vspace*{-1em}
            \caption{Hamiltonian as a snapshot (left) and quantum simulation (unitary) as movie generation (right); $\cn{exp}$ turns a Hamiltonian to unitary and \cn{log} can perform a reverse.}\label{fig:ham-sim}   
\vspace*{-1.5em}
\end{figure}

\noindent{\textbf{\emph{Two Modes and Two Applications in Second Quantization.}}}
In second quantization, operations can be generalized to Hamiltonian $\hat{H}$, a Hermitian matrix ($\hat{H} = \sdag{\hat{H}}$), with two major applications: energy expectation value calculation and quantum Hamiltonian simulation.
Second quantization for particle systems is how scientific computing users in physics and chemistry think of programs; typically, they view a program having two modes, shown in \Cref{fig:ham-sim}, and the two modes corresponding mainly to two different applications below.

{\small
\begin{center}
$
E=\frac{
\sapp{\sdag{\psi}}{(\sapp{\hat{H}}{\psi})}
}
{\dabs{\psi}^2}
\qquad
\qquad
\psi(t) = \sapp{\eexp{\sapp{\hat{H}}{t}}}{\psi(0)}
$
\end{center}
}
%\vspace*{-0.6em}

One can define a Hamiltonian matrix $\hat{H}$ to describe static particle interaction information in the snapshot mode, capturing a moment that indicates potential particle movements, analogous to the left side of \Cref{fig:ham-sim}.
Defining the snapshot can help users compute the energy states of holding the particle structures in $\hat{H}$, the expectation value calculation (the left one above),
where ground energy state computation refers to finding the minimal expectation value among different choices of $\psi$ (\Cref{sec:qcompile}).
Additionally, one can perform Hamiltonian simulation, the simulation mode, by applying a matrix exponential to $\hat{H}$, as $\eexp{\sapp{\hat{H}}{t}}$ (can be understood as $\eexp{\sapp{t}{\hat{H}}}$), producing a unitary (executable in quantum computers). It is analogous to automatically generating a movie (right in \Cref{fig:ham-sim}), showing the trace particle movements propagating $\hat{H}$ in time $t$. Any circuit-based quantum computation can be interpreted as a form of Hamiltonian simulation.
One could also apply a logarithm to the unitary to get the intermediate snapshot for performing other applications.

Many applications can be formulated based on the two modes, e.g., adiabatic evolution \cite{farhi2000quantum} is a kind of Hamiltonian simulation, but its main purpose is to wait for the simulation to converge to a ground energy gate for eigenvalue computation (\Cref{sec:cases}).
In addition, many physical system users also intend to transform one particle system into another, analogous to compiling a program for another system to find solutions.
More details are in \Cref{sec:qcompile}.

%where they define physical systems in terms of Hamiltonians with applications, such as performing ground energy state computation and Hamiltonian simulation. 
%Lattice-based Hamiltonian matrices (\Cref{sec:background}) describe particle interactions by locating particles in a fixed-site structure formed as a lattice, and Hamiltonian simulation simulates the Hamiltonian's transition behavior for a quantum state over a period of time.
%In this program view, the quantum computer behavior is a procedure of a lattice-based particle system Hamiltonian simulation.
%A trivial generalization is to generalize the Hamiltonian simulator view of quantum computers as a particle system, describable by second quantization, a standard physical formalism for describing quantum particle movements. 

\noindent{\textbf{\emph{Quantum Measurement and Expectation Value Computation.}}}
In a quantum computer, a \emph{measurement} operation typically refers to computational basis measurement to extract classical information from a quantum state by collapsing the state to a basis state with a probability related to the value's amplitudes (\emph{measurement probability}), e.g., measuring $\frac{1}{\sqrt{2}}(\ket{0} + \ket{1})$ collapses the value to $\ket{0}$ with probability $\frac{1}{2}$,  and likewise for $\ket{1}$, returning classical value $0$ or $1$, respectively.
A more general Von Neumann measurement \cite{mike-and-ike} is to define a set of orthonormal basis $\{\psi_i\}$, and define the measurement operation based on two normalized state $\psi$ and $\psi'$ ($\dabs{\psi}=\dabs{\psi'}=1$), as $\slen{\sapp{\sdag{\psi}}{\psi'}}^2$, where $\psi$ is a linear sum of the above orthonormal bases, $\sdag{\psi}$ is a row vector, and $\sapp{\sdag{\psi}}{\psi'}$ computes an inner product.
This means the following: given the state $\psi'$, we measure the probability of projecting it to the state $\psi$.
For example, we can define Hadamard bases $\ket{\pm}=\frac{1}{\sqrt{2}}(\ket{0}\pm\ket{1})$; the probabilities of projecting $\ket{+}$ on to $\ket{0}$ and the vice versa case, $\slen{\sapp{{\bra{0}}}{\ket{+}}}^2$ and $\slen{\sapp{{\bra{+}}}{\ket{0}}}^2$ ($\sdag{\ket{d}} =\bra{d}$), are both half.
In quantum computing, $\psi$ can be achieved by applying a unitary $U$ to initial states $\ket{0}^{\otimes n}$.
if $\psi=U\ket{0}^{\otimes n}$, we can think of the energy computation as  $\slen{\sapp{\sdag{\psi}}{\psi'}}^2=\slen{\sapp{\bra{0}^{\otimes n}}{(\sapp{\sdag{U}}{\psi'})}}^2$,
i.e., we apply the inverse unitary of $U$ to $\psi'$ and see the probability of measuring (computational basis measurement) out of zero.
The above expectation value calculation can be thought of as computing the inner product $\sapp{\sdag{\psi}}{\psi'}$, where $\psi$ is normalized and $\psi'$ is the result of $\sapp{\hat{H}}{\psi}$, not necessary normalized. In \qsnd, we do not directly include a measurement operation but include inner products and normalization operations to permit the definition of expectation value calculation, whereas measurement output operations, i.e., the operations that produce the probability of projecting a quantum state $\psi$ onto $\psi'$,  can be defined on top of these operations.
}

\section{Error Bound Lemmas for Trotterization}\label{appx:trotterthm}

The Coq proof of the standard Trotterization relies on a key error bound lemma below.

\begin{lemma}
[Error bound for Trotterization]\label{thm:theorem_std_trotter}\rm 
Given $\hat{H}=\sum_{i=1}^N \hat{H}_i$ and time $r$, 
 we have error bound as below: 
\[
||\Pi_{i=1}^N \eexp{r \hat{H}_i} - \eexp{r \hat{H}}|| \leq 
\displaystyle { \frac{r^2}{2} \sum_{k=1}^N \left\lVert 
(\sum_{j=k+1}^N \hat{H}_j) \circ \hat{H}_k - \hat{H}_k \circ (\sum_{j=k+1}^N \hat{H}_j)
\right\rVert}.
\]
Here ||.|| denotes the operator norm.
\end{lemma}

The Coq proof scratch of \Cref{thm:theorem_std_trotter} is listed below.

\begin{proof} 
\begin{align*}
& \indent \Pi_{i=1}^N \eexp{r \hat{H}_i} - \eexp{r \hat{H}} \\
& = \Pi_{i=1}^N \eexp{r \hat{H}_i} - \eexp{r \sum_{j=N}^N H_j} \cdot \Pi_{i=1}^{N-1}\eexp{r \hat{H}_i}  \\
& \indent + \eexp{r \sum_{j=N}^N H_j} \cdot \Pi_{i=1}^{N-1}\eexp{r \hat{H}_i} 
 - \eexp{r \sum_{j=N-1}^N H_j} \cdot \Pi_{i=1}^{N-2}\eexp{r \hat{H}_i} \\
& \indent + ... \\
& \indent + \eexp{r \sum_{j=2}^N H_j} \cdot \Pi_{i=1}^{1}\eexp{r \hat{H}_i} 
- \eexp{r \hat{H}} 
\end{align*}

Using the triangle inequality of operator norm, we get 
\begin{align}
& \indent ||\Pi_{i=1}^N \eexp{r \hat{H}_i} - \eexp{r \hat{H}}|| \\
& \leq \sum_{k=1}^N || \eexp{r \sum_{j=k+1}^N \hat{H}_j} \cdot \Pi_{i=1}^k\eexp{r \hat{H}_i} 
 - \eexp{r \sum_{j=k}^N \hat{H}_j} \cdot \Pi_{i=1}^{k-1}\eexp{r \hat{H}_i} || \\
& \leq \sum_{k=1}^N || \eexp{r \sum_{j=k+1}^N \hat{H}_j} \cdot \eexp{r \hat{H}_k} 
 - \eexp{r \sum_{j=k}^N \hat{H}_j} ||  \\
&  \leq \frac{r^2}{2} \sum_{k=1}^N || (\sum_{j=k+1}^N \hat{H}_j) \circ \hat{H}_k
- \hat{H}_k \circ (\sum_{j=k+1}^N \hat{H}_j)||.
\end{align}
Step (2) is from the triangle inequality of operator norm; step (3) 
is because the norm of a unitary matrix is 1; step (4) can be proved by an inductive proof, which shows
\begin{equation}
    ||\eexp{rB} \circ \eexp{rA} - \eexp{r(A+B)}|| \leq \frac{r^2}{2} ||B \circ A - A 
    \circ B||.
\end{equation}
\end{proof}

The Coq proof for QDrift relies on a key error bound lemma below.
Here, we rely on the concept of a superoperator of a quantum state.
Quantum physicists have two main ways to mathematically represent a quantum state: a \emph{pure} state is represented by a state \emph{vector} as in \Cref{fig:data}, while a \emph{mixed} state is represented by a density \emph{operator} (or density \emph{matrix}).
A computation on a pure state is described by a linear operator from the input Hilbert space to the output Hilbert space, while a computation on mixed states is described by a \emph{superoperator}, a sort of higher-order function that acts as a linear operator from density operators to density operators.

\begin{lemma}
[Error bound for QDrift]\label{thm:theorem_qdrift}\rm 
Given $\hat{H}=\sum_{i=0}^{d\sminus 1} \theta_i \hat{H}_i$ and time $r$, where 
$\theta_i$ is positive and is decided by normalizing the terms such that $\max_i||H_i|| = 1$. 
Let $\lambda = \sum_i \theta_i$. We randomly sample $N$ times among the terms 
following the probability distribution that the frequency of $H_i$ in the sampled list is $\theta_i/\lambda$. Let $\Vec{j} = \{j_1, j_2, ..., j_N\}$ denote the order of the sampling where $j_k \in [1, d]$,
and create superoperator $L$ and $L_i$ for the Hamiltonion which satisfies that
for any quantum state $\psi$, with its density operator $\rho=\psi \, \sdag{\psi}$, $L(\rho) = -i(\hat{H}\rho - \rho \hat{H})$
and $L_i(\rho) = -i(\hat{H}_i\rho - \rho \hat{H}_i)$.
Then we can get the following error bound:
\[
||\Pi_{k=1}^N \eexp{\frac{\lambda r}{N} L_{j_k}} - \Pi^N \eexp{\frac{r}{N} L}||_\diamond \leq 
\displaystyle{ \frac{2\lambda^2 r^2}{N}}.
\]
Here $|| . ||_\diamond$ denotes the diamond norm.
\end{lemma}

The Coq proof scratch of \Cref{thm:theorem_qdrift} is listed below.

\begin{proof}

    Based on the definition of superoperators $L \rho = i(\hat{H}\rho - \rho \hat{H})$,  triangle inequality of diamond norm, and $|| \rho ||_\diamond = 1$,
    we can get $|| L ||_\diamond \leq || H || \leq 2\lambda$. 
    
    Similarly, as $\max_i||H_i|| = 1$, we have $|| L_i ||_\diamond \leq || H_i || \leq 2$. 

    In each step of QDrift, a randomly chosen gate is
    $\varepsilon = \sum_i \frac{\theta_i}{\lambda} \eexp{\frac{\lambda r}{N} L_i}$.

    It is compared to 
    $\mu = \eexp{r L/N}$.

    Using the expansion of $\eexp{r A} = \sum_{n=0}^\infty \frac{(i r)^n A^n}{n!}$,
we can get 
$d = || \mu - \varepsilon ||_\diamond \leq
2\sum_{n=2}^{\infty} \frac{1}{n!} (\frac{2\lambda r}{N})^n$.

Using $\sum_{n=2}^{\infty} \frac{x^n}{n!} \leq \frac{x^2}{2} e^x$,
we get $d \leq \frac{2\lambda^2r^2}{N^2} e^{2\lambda r/N} \simeq \frac{2\lambda^2r^2}{N^2}$.

Thus, the total error of the $N$ random sampling satisfies
$D \leq N d \leq  \frac{2\lambda^2r^2}{N} e^{2\lambda r/N} \simeq \frac{2\lambda^2r^2}{N}$.
So we get the claimed error bound.
\end{proof}

\section{Gate Syntheses in Analog Quantum Machines}\label{appx:gatesyn}

At the quantum computer machine level, every gate is implemented as the simulation of a gate Hamiltonian over a certain time period. To define a unitary gate \(U\), theoretically, we need to find its snapshot Hamiltonian by taking the logarithm $\elog{U}$.
Since $U$ is usually small, we typically perform reverse engineering by computing the eigendecomposition of \(U\) and forming a $t(2)$-typed Hamiltonian \(\hat{H}_U\) based on the eigenstates found through eigendecomposition, and simulates the gate based on \(\eexp{r\hat{H}_U}\) with respect to time $r$. We show below the simulation of a Hadamard (\cn{H}) and controlled-not (\cn{CX}) gate, with respect to time period \(\pi\).

{\small
\[
\hat{H}_{\cn{H}} = \frac{1}{2\sqrt{2}} \cdot (X+Z)
\qquad
\cn{H} = \eexp{\pi \hat{H}_{\cn{H}}}
\]
}

{\small
\[
\hat{H}_{\cn{H}} = (\frac{1}{2} \sminus \frac{\sqrt{2}}{4}) \cdot \mathbb{0} + (\frac{1}{2} \sminus \frac{\sqrt{2}}{4} ) \cdot \mathbb{1} + (\sminus \frac{\sqrt{2}}{4}) \cdot X
\qquad
\cn{H} = \eexp{\pi \hat{H}_{\cn{H}}}
\]
}

The above formula is the one-particle Hamiltonian for simulating a Hadamard gate (\cn{H}). To reverse engineer the Hamiltonian, we perform an eigendecomposition on Hadamard gates (\cn{H}), find an eigenstate \(\psi = -\sin(\frac{\pi}{8})\ket{0} + \cos(\frac{\pi}{8})\ket{1}\). We can then construct Hadamard Hamiltonians by the matrix multiplication $\sapp{\psi}{\sdag{\psi}}$, which is the \qsnd Hamiltonian $\hat{H}_{\cn{H}}$ above. The Hadamard gate (\cn{H}) can then be defined as a simulation of the Hamiltonian, as \(\cn{H} = \eexp{\pi \hat{H}_{\cn{H}}}\), with the time period \(\pi\).
Similarly, we can find the two-particle Hamiltonian for a controlled-not gate (\cn{CX}) through eigendecomposition and implement it in \qsnd, as shown below. The gate can then be defined as a simulation of the Hamiltonian as \(\cn{CX} = \eexp{\pi \hat{H}_{\cn{CX}}}\), with the time period \(\pi\).

{\small
\[
\hat{H}_{\cn{CX}} = \mathbb{1} \otimes (\sminus \frac{1}{2} \cdot X + \frac{1}{2} \cdot I)
\qquad
\cn{CX} = \eexp{\pi \hat{H}_{\cn{CX}}}
\]
}

In short, matrix exponential operations turn a Hamiltonian snapshot into a quantum system simulation, while a matrix logarithm operation, modeled in \qsnd, reverses a simulation to a Hamiltonian snapshot. For example, by applying a matrix logarithm on the \cn{H} and \cn{CX} gates, we might find a set of results, and \(\hat{H}_{\cn{H}}\) and \(\hat{H}_{\cn{CX}}\) are in the set \footnote{Matrix logarithms are not unique. \qsnd fixes the uniqueness by choosing one implementation. However, the \qsnd type system and equivalence relations are general enough without relying on the specific implementation. }.

\noindent{\textbf{\emph{Hamiltonian for Better Compilation.}}}
Finding the right target machine Hamiltonian to simulate a gate, i.e., gate synthesis, is usually referred to as control pulse level quantum programming \cite{Li_2022,Tacchino_2019} as control pulses can theoretically be described as \emph{target machine Hamiltonians}, and previous work \cite{10.1145/3632923} has discovered the advantage of compiling system Hamiltonians directly to target machine Hamiltonians by bypassing the series of simulations of individual gates.
To see how this bypassing works, we simulate a simple $1D$ Ising model system, describing phase transition and quantum magnetism behaviors in insulating magnetic systems.

{\small
\[
\hat{H}_{I2} = \sum_{j} \sapp{Z(j)}{Z(j\splus 1)} + z_h \sum_j X(j)
\tag{3.3.2}\label{eq:hi1}
\]
}

The $\hat{H}_{I2}$ equation is a special $1D$ Ising model with uniformed interaction strength, where the first terms in the quantum choice represent the particle interactions, while the second terms represent effects from the external transverse magnetic field, with $z_h$ being the external transverse magnetic field strength.
Simulating the system with respect to a time period $r$ is equivalent to compiling the exponent $\eexp{r\hat{H}_{I2}}$ to quantum computers. The main step involves the Trotterization procedure, a technique that splits $\hat{H}_{I2}$ into small Hamiltonians, executed in sequence, each of which is implementable in a target machine Hamiltonian.
One small target machine Hamiltonian is $\sapp{Z(j)}{Z(j\splus 1)}$, applying on $j$-th and $j\splus 1$-th qubits, with respect to a very small time period $\frac{r}{n}$, i.e., Trotterization tries to find a large number $n$ so the simulation of many small time periods $\frac{r}{n}$ can correctly approximate the simulation of the original system Hamiltonian.
The result of gate synthesizing on the machine Hamiltonian, as to find gates representing $\eexp{\sapp{\frac{r}{n}{Z(j)}}{Z(j\splus 1)}}$, is shown as the left-hand side circuit in \Cref{fig:zz-inter}, performing the circuit program, performing two controlled-not gates and an $\cn{Rz}(2\frac{r}{n})$ gate in the middle for the $j$-th and $j\splus 1$-th qubits,
i.e., $\cn{Rz}(\theta)$ is an $z$-axis phase rotation gate.
As we mentioned above, these quantum gates are eventually simulated as Hamiltonians in pulse-level programming.

On the other hand, instead of performing the sequence of the three gate simulations above,
one can directly implement the small Hamiltonian in an IBM machine, which supports direct pulse level programming and provides the $ZZ$ interaction gate shown as the right in \Cref{fig:zz-inter}; the special gate essentially simulates the target machine Hamiltonian $\sapp{Z(j)}{Z(j\splus 1)}$ in a much more effective way than the above gate synthesis step.
The bypass demonstrates the potential for improving compilation effectiveness.
In many cases, even if one is willing to implement system Hamiltonians as circuit gates, it is still valuable for them to notice the potential optimization done through the above procedure, definable in \qsnd, as many optimization frameworks \cite{VOQC,10.1145/3519939.3523433} decompose a $ZZ$ gate as the three gates above, which is counterproductive.

\section{Perturbative Gadget}\label{appx:gadget}

%Here, we describe the algorithm step by step.

\ignore{

\begin{algorithm}[t]
\caption{Perturbative Gadget}\label{fig:gadget}
\begin{algorithmic}[1]
\STATE \textbf{Input:} A \(k\)-local Hamiltonian \( \hat{H} = \sum_{j=1}^{N} r_j \hat{H}_j \).
\STATE \textbf{Output:} A two-local Hamiltonian \( H^{\cn{gad}} \) simulating the original \(k\)-local Hamiltonian.

\STATE \textbf{Step 1: Introduce Ancilla Qubits}
\FOR{each term \( \hat{H}_j \) in \( \hat{H} \)}
    \STATE Introduce \( k \) ancilla qubits for the \( j \)-th term.
\ENDFOR

\STATE \textbf{Step 2: Construct the Gadget Hamiltonian \( \hat{H}^{\cn{gad}} \)}
\FOR{each term \( \hat{H}_j \) in \( \hat{H} \)}
    \STATE Construct ancilla interaction terms: \( \hat{H}^{\cn{anc}} = \sum_{1 \leq i < j \leq k} \frac{1}{2} \left( I - Z(i)(j) \circ Z(k)(j) \right) \)
    \STATE Construct interaction terms between computational qubits and ancillas: \( \hat{H}^{\cn{v}}_j = \sum_{m=1}^{k} r_{j} \pau_{j,m}(j) \circ X(j)(m) \)
    \STATE Form partial gadget Hamiltonian: \( \hat{H}_j^{\cn{gad}} = \hat{H}^{\cn{anc}} + \hat{H}^{\cn{v}}_j \)
\ENDFOR

\STATE Form full gadget Hamiltonian: \( \hat{H}^{\cn{gad}} = \sum_{j=1}^{N} \hat{H}_j^{\cn{gad}} \)

\STATE \textbf{Step 3: Apply Perturbation Theory}
\STATE Apply perturbation theory to \( \hat{H}^{\cn{gad}} \) to obtain the effective low-energy Hamiltonian approximating \( \hat{H} \).

\STATE \textbf{Step 4: Output}
\STATE Return \( \hat{H}^{\cn{gad}} \) as the two-local Hamiltonian simulating \( \hat{H} \).
\end{algorithmic}
\end{algorithm}

In the algorithm, the $z_{\lambda}$ value selection is a component that depends on the specific applications. Here, we assume that such $z_{\lambda}\in[0,1)$ exists and is a real number. Our correctness statement depends on the $z_{\lambda}$ selections, which are proved in Coq.

\begin{lemma}[Perturbative Gadget Correctness]\label{thm:gadget-good}\rm 
Given $\hat{H}=\sum_j \bigotimes_{k=0}^n \alpha(j,k)$, typed as $\quan{F}{\hmx}{\bigotimes^n t(2)}$, via the algorithm in \Cref{fig:gadget}, we produce $\hat{H'} = \sum_j  (z(j)\hat{H}_1(j) + \hat{H}_2)$, for every $r$ and state $\psi$, we add $k*j$ ancilla $t(2)$ qubits to $\psi$ and transform it to $\psi'$,
Let $\psi_1 = \denote{\tjudge{\bigotimes^n t(2)}{\eexp{\hat{H} r}}{\quan{F}{\umx}{\bigotimes^n t(2)}}}_0(\psi)$ and $\psi'_1=\denote{\tjudge{\bigotimes^n t(2)}{\eexp{\hat{H'} r}}{\quan{F}{\umx}{\bigotimes^n t(2)}}}_0(\psi')$, by transforming $\psi_1$ to $\psi_2$ via the same ancilla qubit addition above as $\psi_2$, we show that $\dabs{\psi_2} - \dabs{\psi'_1}$ is bound by $O(z_{\lambda}^4)$.
\end{lemma}

% Gadget example for r=2 and k=4
\myparagraph{Example: Gadget Hamiltonian for \( r = 2 \), \( k = 4 \)}

We explicitly construct the perturbative gadget Hamiltonian corresponding to a target composite Hamiltonian with \( r = 2 \) terms, each acting on \( k = 4 \) qubits. The full gadget Hamiltonian consists of interaction and ancillary parts, designed to replicate the effect of the target Hamiltonian in a low-energy subspace.

\subsection*{Target Composite Hamiltonian}

The target Hamiltonian is defined as:
\[
H_{\text{comp}} = c_1 H_1 + c_2 H_2,
\]
where each local term \( H_s \) is given by a tensor product of Pauli operators:
\[
H_1 = \sigma_{1,1} \otimes \sigma_{1,2} \otimes \sigma_{1,3} \otimes \sigma_{1,4}, \quad
H_2 = \sigma_{2,1} \otimes \sigma_{2,2} \otimes \sigma_{2,3} \otimes \sigma_{2,4}.
\]

\subsection*{Gadget Hamiltonian}

The corresponding gadget Hamiltonian takes the form:
\[
H_{\text{gad}} = H_1^{\text{anc}} + H_2^{\text{anc}} + V_1 + V_2,
\]
where \( H_s^{\text{anc}} \) encodes the energy penalty terms enforcing consistency constraints, and \( V_s \) introduces perturbative couplings between the target and auxiliary systems.

\subsubsection*{Ancillary Terms}

Each \( H_s^{\text{anc}} \) is defined as:
\[
H_s^{\text{anc}} = \frac{1}{2} \sum_{1 \leq i < j \leq 4} \left( I - Z_{s,i} Z_{s,j} \right),
\]
which we enumerate explicitly as follows:

\textbf{For \( s = 1 \):}
\begin{align*}
H_1^{\text{anc}} = \frac{1}{2} \Big(&
(I - Z_{1,1}Z_{1,2}) +
(I - Z_{1,1}Z_{1,3}) +
(I - Z_{1,1}Z_{1,4}) + \\
& (I - Z_{1,2}Z_{1,3}) +
(I - Z_{1,2}Z_{1,4}) +
(I - Z_{1,3}Z_{1,4}) \Big)
\end{align*}

\textbf{For \( s = 2 \):}
\begin{align*}
H_2^{\text{anc}} = \frac{1}{2} \Big(&
(I - Z_{2,1}Z_{2,2}) +
(I - Z_{2,1}Z_{2,3}) +
(I - Z_{2,1}Z_{2,4}) + \\
& (I - Z_{2,2}Z_{2,3}) +
(I - Z_{2,2}Z_{2,4}) +
(I - Z_{2,3}Z_{2,4}) \Big)
\end{align*}

\subsubsection*{Interaction Terms}

The interaction Hamiltonians \( V_s \) couple the auxiliary and target degrees of freedom. Each is defined as:
\[
V_s = \sum_{j=1}^{4} c_{s,j} \, \sigma_{s,j} \otimes X_{s,j},
\]
where the coefficients \( c_{s,j} \) are:
\[
c_{s,j} =
\begin{cases}
c_s, & \text{if } j = 1, \\
1, & \text{otherwise}.
\end{cases}
\]

Explicitly, we have:

\textbf{For \( s = 1 \):}
\[
V_1 =
c_1 \, \sigma_{1,1} \otimes X_{1,1} +
\sigma_{1,2} \otimes X_{1,2} +
\sigma_{1,3} \otimes X_{1,3} +
\sigma_{1,4} \otimes X_{1,4}
\]

\textbf{For \( s = 2 \):}
\[
V_2 =
c_2 \, \sigma_{2,1} \otimes X_{2,1} +
\sigma_{2,2} \otimes X_{2,2} +
\sigma_{2,3} \otimes X_{2,3} +
\sigma_{2,4} \otimes X_{2,4}
\]
}

Below, we show the perturbative gadget algorithm.

\begin{definition}[Perturbative Gadget]\label{def:pgadget}
Let \( \hat{H}^{\cn{target}}  \) be a \( k \)-local Hamiltonian acting on \( n \) qubits, expressed as

\[
\hat{H}^{\cn{target}}  = \sum_{j=1}^N r_j \hat{P}_j,
\]

where each \( \hat{P}_j \) is a tensor product of \( k \) single-qubit Pauli operators acting on a subset of the \( n \) computational qubits:

\[
\hat{P}_j = \pau_{j,1} \otimes \pau_{j,2} \otimes \cdots \otimes \pau_{j,k}.
\]

To simulate \( \hat{H}^{\cn{target}}  \) using only two-local interactions, we construct a \emph{gadget Hamiltonian} \( \hat{H}^{\cn{gad}} \), acting on the original \( n \) computational qubits together with \( N k \) ancilla qubits. The gadget Hamiltonian is defined as:
\[
\hat{H}_{\cn{gad}} = \sum_{j=1}^{N} \hat{H}_j^{\text{anc}} + \lambda \sum_{j=1}^{N} \hat{H}^{\cn{V}}_j,
\]
where
\[
\hat{H}_j^{\text{anc}}  = \frac{1}{2} \sum_{1 \leq m < n \leq k} \left( I - Z(m(j)) \circ Z(n(j)) \right),
\]
and
\[
\hat{H}^{\cn{v}}_j = \sum_{n=1}^{N} r_{j,n} \, \pau_{j,n}(n) \circ X(n(j)),
\quad \text{with} \quad
r_{j,n} =
\begin{cases}
r_j, & \text{if } n = 1, \\
1, & \text{otherwise}.
\end{cases}
\]

Here, \( X(n(j)) \) and \( Z(n(j)) \) are Pauli operators acting on the \( n \)-th ancilla qubit associated with the \( j \)-th term. This construction ensures that all terms in \( \hat{H}^{\cn{gad}} \) involve only two-qubit interactions.
\end{definition}

\ignore{
\begin{definition}[Gadget Hamiltonian]
Let \( H_{\text{target}} \) be a \( k \)-local Hamiltonian acting on \( n \) qubits, expressed as
\[
H_{\text{target}} = \sum_{s=1}^r c_s H_s,
\]
where each \( H_s \) is a tensor product of \( k \) single-qubit Pauli operators acting on a subset of the \( n \) computational qubits:
\[
H_s = \sigma_{s,1} \otimes \sigma_{s,2} \otimes \cdots \otimes \sigma_{s,k}.
\]
Each local operator is of the form
\[
\sigma_{s,j} = \hat{n}_{s,j} \cdot \vec{\sigma}_{s,j},
\]
where \( \hat{n}_{s,j} \in \mathbb{R}^3 \) is a unit vector, and \( \vec{\sigma}_{s,j} = (X, Y, Z) \) denotes the Pauli vector acting on the relevant qubit.

To simulate \( H_{\text{target}} \) using only two-local interactions, we construct a \emph{gadget Hamiltonian} \( H_{\text{gad}} \), acting on the original \( n \) computational qubits together with \( rk \) ancilla qubits. The gadget Hamiltonian is defined as:
\[
H_{\text{gad}} = \sum_{s=1}^{r} H_s^{\text{anc}} + \sum_{s=1}^{r} V_s,
\]
where
\[
H_s^{\text{anc}} = \frac{1}{2} \sum_{1 \leq i < j \leq k} \left( I - Z_{s,i} Z_{s,j} \right),
\]
and
\[
V_s = \sum_{j=1}^{k} c_{s,j} \, \sigma_{s,j} \otimes X_{s,j},
\quad \text{with} \quad
c_{s,j} =
\begin{cases}
c_s, & \text{if } j = 1, \\
1, & \text{otherwise}.
\end{cases}
\]

Here, \( X_{s,j} \) and \( Z_{s,j} \) are Pauli operators acting on the \( j \)-th ancilla qubit associated with the \( s \)-th term. This construction ensures that all terms in \( H_{\text{gad}} \) involve only two-qubit interactions.
\end{definition}
}

We show below the two lemmas for perturbative gadget, which are mechanized in \rocq.

\begin{lemma}[Eff Hamiltonian Existance]\label{thm:eff}\rm
Let \( \hat{H}^{\cn{target}} \) and \( \hat{H}^{\cn{gad}} \) be defined as in \Cref{def:pgadget}, and suppose the perturbation strength satisfies \( \lambda \leq \lambda_{\max} \), where
\[
\lambda_{\max} = \frac{k - 1}{4} \left( \sum_{j=1}^{N} |r_j| + N(k - 1) \right)^{-1}.
\]
Then, there exists a function \( f(\lambda) = O(\mathrm{poly}(\lambda)) \) and a constant \( \Xi = O(\mathrm{poly}(k)) \) such that the effective Hamiltonian satisfies
\[
\hat{H}^{\cn{eff}}(\hat{H}^{\cn{gad}}, 2n) = \frac{\lambda^k}{\Xi} \, \hat{H}^{\cn{target}}  \otimes \ket{0}\bra{0}^{\otimes rk} + f(\lambda)\, \Pi + O(\lambda^{k+1}),
\]
where \( \Pi \) is the projector onto the support of \( \hat{H}^{\cn{eff}}(\hat{H}^{\cn{gad}}, 2n) \).
\end{lemma}

\begin{lemma}[Perturbative Gadget Correctness]\label{thm:gadgetcorrect}\rm
Let \( \hat{H}^{\cn{target}}  \) and \( \hat{H}^{\cn{gad}} \) be as in \Cref{thm:eff}. Then there exists a perturbation threshold \( \lambda^* \), with
\[
\lambda^* \ll \lambda_{\max} \quad \text{and} \quad \lambda^* \ll \frac{E_1^{\cn{target}} - E_0^{\cn{target}}}{\Xi \| O_{\cn{err}} \|},
\]
such that for all \( \lambda \leq \lambda^* \), the ground states \( \psi_0 = |\psi_0\rangle\langle\psi_0| \) of \( \hat{H}^{\cn{target}}  \) and \( \phi_0 = |\phi_0\rangle\langle\phi_0| \) of \( \hat{H}^{\cn{gad}} \) satisfy:
\[
\left\| \psi_0 - \mathrm{Tr}_{\text{aux}}[\phi_0] \right\|_2 = O(\lambda),
\]
where \( \mathrm{Tr}_{\text{aux}}[\cdot] \) denotes the partial trace over all ancilla qubits.
\end{lemma}

We can now demonstrate the entire pipeline's compilation correctness, including the error introduced by the perturbative gadget algorithm.

\begin{theorem}[Compilation Correctness Including Perturbative Gadget]\label{thm:compile-good1}\rm 
Given $e$, typed as $\tjudge{\iota}{e}{\quan{F}{\hmx}{\iota}}$, and a time period $r$, we compile the simulation $\eexp{\uapp{r}{e}}$ to quantum circuit, via the compilation pipeline $\quan{F}{\hmx}{\iota} \vdash (e,r)  \gg U : \quan{F}{\umx}{\iota'}$, for every state $\psi$, typed as $\iota \vdash \psi$, we transform it to $\psi'$, typed as $\iota' \vdash \psi'$, let $\psi_1=\denote{\tjudge{\iota}{U}{\quan{F}{\umx}{\iota}}}_0(\psi)$, we transform $\psi_1$ to $\psi'_1$, typed as $\iota' \vdash \psi'_1$; thus, $\dabs{\psi_1' -\denote{\tjudge{\iota'}{\eexp{e r}}{\quan{F}{\umx}{\iota'}}}_0(\psi')}<\epsilon+O(\lambda)$. 
\end{theorem}

\section{Other Particle Transformation Methods}\label{sec:jw-trans}

\Cref{sec:transformation} provides a simple particle transformation method. We introduce another transformation method, the Bravyi-Kitaev transformation, which might result in more effective compilation.
%\subsection{Bravyi--Kitaev Mapping}

The Bravyi--Kitaev mapping stores the parity of fermionic modes in a nonlocal Fenwick tree data structure. A Fenwick tree is a partial ordering of binary representations where each node is a copy of its parent with one of the ones changed to zero. This structure allows encoding fermionic operators with Pauli weight scaling as \( O(\log m) \) for \( m \) fermionic modes.

Each fermionic mode corresponds to a qubit associated with a node in the Fenwick tree. Each qubit stores the total parity of the fermionic modes below it in the tree. Fermionic creation and annihilation operators can be constructed using three key subsets of the Fenwick tree: the update set \( U(\alpha) \), parity set \( P(\alpha) \), and flip set \( F(\alpha) \). They act as follows:

\[
a_i \mapsto \frac{1}{2} \left(
\bigotimes_{j \in U(i)} X_j \otimes X_i \otimes \bigotimes_{j \in P(i)} Z_j
+ i \bigotimes_{j \in U(i)} X_j \otimes Y_i \otimes \bigotimes_{j \in P(i) \setminus F(i)} Z_j
\right),
\]

\[
a_i^\dagger \mapsto \frac{1}{2} \left(
\bigotimes_{j \in U(i)} X_j \otimes X_i \otimes \bigotimes_{j \in P(i)} Z_j
- i \bigotimes_{j \in U(i)} X_j \otimes Y_i \otimes \bigotimes_{j \in P(i) \setminus F(i)} Z_j
\right).
\]

\begin{definition}[Update Set \(U(\alpha)\)]
Let the update set \( U(\alpha) \) be such that \(\beta \in U(\alpha)\) if and only if there exists an index \( i_0 \) such that \(\alpha_{i_0} = 0\), and for all \( i > i_0 \), \(\alpha_i = \beta_i\), and for all \( i < i_0 \), \(\beta_i = 1\). Equivalently, \(\beta \in U(\alpha)\) if and only if \(\alpha \prec \beta\) in the Fenwick tree partial order.
\end{definition}

\begin{definition}[Parity Set \(P(\alpha)\)]
The parity set \( P(\alpha) \) contains elements \(\beta \in P(\alpha)\) if and only if there exists an index \( i_0 \) such that \(\alpha_{i_0} = 1\), \(\beta_{i_0} = 0\), for all \( i > i_0\), \(\beta_i = \alpha_i\), and for all \( i < i_0 \), \(\beta_i = 1\). This implies
\[
\prod_{j \in P(i)} q_j = \prod_{j < i} f_j.
\]
\end{definition}

\begin{definition}[Flip Set \(F(\alpha)\)]
The flip set \( F(\alpha) \) contains elements \(\beta \in F(\alpha)\) if and only if there exists an index \( i_0 \) such that \(\beta_{i_0} = 0\), \(\beta_i = \alpha_i\) for all \( i \neq i_0\), and for all \( i < i_0\), \(\alpha_i = 1\). This implies
\[
\prod_{j \in F(i)} q_j = \prod_{j \prec i} f_j.
\]
\end{definition}

\begin{algorithm}
\caption{Generate Update Set \( U(\alpha) \)}
\begin{algorithmic}[1]
\STATE \textbf{Input:} Fermionic index \(\alpha\), number of modes \(n\)
\STATE \textbf{Output:} Update set \( U(\alpha) \)
\STATE \( \text{updateSet} \gets \emptyset \)
\STATE \( l \gets \text{bit-length of } \alpha \)
\FOR{ \( i = 0 \) to \( l - 1 \)}
    \IF{bit \( i \) of \(\alpha\) is 0}
        \STATE Create \(\beta\) by setting bit \( i \) of \(\alpha\) to 1 and bits \( < i \) to 0
        \IF{ \(\beta < n\) }
            \STATE Add \(\beta\) to \(\text{updateSet}\)
        \ENDIF
    \ENDIF
\ENDFOR
\STATE \textbf{Return} \(\text{updateSet}\)
\end{algorithmic}
\end{algorithm}

\begin{algorithm}
\caption{Generate Parity Set \( P(\alpha) \)}
\begin{algorithmic}[1]
\STATE \textbf{Input:} Fermionic index \(\alpha\)
\STATE \textbf{Output:} Parity set \( P(\alpha) \)
\STATE \( \text{paritySet} \gets \emptyset \)
\STATE \( l \gets \text{bit-length of } \alpha \)
\FOR{ \( i = l-1 \) down to \( 0 \) }
    \IF{bit \( i \) of \(\alpha\) is 1}
        \STATE Create \(\beta\) by setting bit \( i \) to 0 and bits \( < i \) to 1
        \STATE Add \(\beta\) to \(\text{paritySet}\)
    \ENDIF
\ENDFOR
\STATE \textbf{Return} \(\text{paritySet}\)
\end{algorithmic}
\end{algorithm}

\begin{algorithm}
\caption{Generate Flip Set \( F(\alpha) \)}
\begin{algorithmic}[1]
\STATE \textbf{Input:} Fermionic index \(\alpha\)
\STATE \textbf{Output:} Flip set \( F(\alpha) \)
\STATE \( \text{flipSet} \gets \emptyset \)
\STATE \( i \gets 0 \)
\WHILE{bit \( i \) of \(\alpha\) is 1}
    \STATE Create \(\beta\) by setting bit \( i \) to 0
    \STATE Add \(\beta\) to \(\text{flipSet}\)
    \STATE \( i \gets i + 1 \)
\ENDWHILE
\STATE \textbf{Return} \(\text{flipSet}\)
\end{algorithmic}
\end{algorithm}

We show an example of transforming the T-J model via the BK-transformation in \Cref{sec:tj}.

\ignore{
\subsection{Bravyi--Kitaev Mapping}\label{appx:gadget}

We believe that it is necessary to substantialize it as a step in our certified compiler, and clarify that there is a key component in compiling a Hamiltonian simulation that transforming a higher local Pauli string based Hamiltonian to a two-local Hamiltonian, with the note that this step might happen in the decomposition, which is why the step is grayed out in \Cref{fig:compilationprocess}.

The rest of the section explains the perturbative gadget algorithm. We start with the canonicalized Pauli string based Hamiltonian $\hat{H}=\sum_j \bigotimes_{k=0}^n \alpha(j,k)$ with type $\quan{F}{\hmx}{\bigotimes^n t(2)}$ produced in \Cref{sec:compilecanonical}, and produce a two-local Hamiltonian via the algorithm .

Recall that a Hamiltonian represents a constraint applying to a quantum particle system. The algorithm views a $k$-local Pauli string as a length $k$ conjuncted constraint manipulating $k$ different sites. It transforms the constraint into a quantum choice, a.k.a., linear sum, of independent constraints applying to the $k$ sites.
To understand this, we look at a similar example of graph coloring. Given an undirected graph $G = (V, E)$ and a set of $n$ colors, we want to construct a tensor string describing the constraint of a vertex having the $k$-th color. For example, if there are totally $16$ colors, we can assume that the particular vertex site is typed as $t(16)$, which will be transformed to four $t(2)$ sites ($\bigotimes^4 t(2)$), so that he $14$-th color can be represented as $\mathbb{1} \otimes \mathbb{1}\otimes \mathbb{1} \otimes \mathbb{0}$, essentially $14$'s binary representation.
The transformed Pauli string of the above term is clearly not two local. To rewrite the tensor string, note that the reason we tensor out the four terms ($\mathbb{1}$ or $\mathbb{0}$) together is that we want to enumerate the $14$-th color. Instead of enumerating colors, one can think of each vertex site as having $16$ color occupation cells; the $k$-th cell is $t(2)$ typed, and the binary basis-vector represents if the $k$-th color occupies the cell. Therefore, the tensor string constraint on the site becomes $\mathbb{1}(14)+\sum_{j\neq 14}\mathbb{0}(j)$, i.e., we constrain the $14$-th cell being turned on and the other cells being off. In this case, we need to place an additional constraint $\cn{sum}_1$ that there is only one cell having a color turning on \cite{Lucas_2014}, which is represented as the Hamiltonian $V_s$ in \Cref{fig:gadget}.

\begin{theorem}
Let \( H_{\text{target}} \) and \( H_{\text{gad}} \) be defined as in Definition~1, and suppose the perturbation strength satisfies \( \lambda \leq \lambda_{\max} \), where
\[
\lambda_{\max} = \frac{k - 1}{4} \left( \sum_{s=1}^{r} |c_s| + r(k - 1) \right)^{-1}.
\]
Then, there exists a function \( f(\lambda) = O(\mathrm{poly}(\lambda)) \) and a constant \( \Xi = O(\mathrm{poly}(k)) \) such that the effective Hamiltonian satisfies
\[
H_{\mathrm{eff}}(H_{\mathrm{gad}}, 2n) = \frac{\lambda^k}{\Xi} \, H_{\text{target}} \otimes \ket{0}\bra{0}^{\otimes rk} + f(\lambda)\, \Pi + O(\lambda^{k+1}),
\]
where \( \Pi \) is the projector onto the support of \( H_{\mathrm{eff}}(H_{\mathrm{gad}}, 2n) \).
\end{theorem}

\begin{lemma}
Let \( H_{\text{target}} \) and \( H_{\text{gad}} \) be as in Theorem~1. Then there exists a perturbation threshold \( \lambda^* \), with
\[
\lambda^* \ll \lambda_{\max} \quad \text{and} \quad \lambda^* \ll \frac{E_1^{\text{target}} - E_0^{\text{target}}}{\Xi \| O_{\text{err}} \|},
\]
such that for all \( \lambda \leq \lambda^* \), the ground states \( \psi_0 = |\psi_0\rangle\langle\psi_0| \) of \( H_{\text{target}} \) and \( \phi_0 = |\phi_0\rangle\langle\phi_0| \) of \( H_{\text{gad}} \) satisfy:
\[
\left\| \psi_0 - \mathrm{Tr}_{\text{aux}}[\phi_0] \right\|_2 = O(\lambda),
\]
where \( \mathrm{Tr}_{\text{aux}}[\cdot] \) denotes the partial trace over all ancilla qubits.
\end{lemma}

\section{A Hubbard Model for Hydrogen Chains}\label{sec:hubbard} % Added a label. Feel free to change/remove

The Hubbard model (system) describes the interactions between elementary particles, explicitly focusing on the electrons having fermionic behavior. 
Here, we focus on the one-dimensional Hydrogen chain, one of the most quintessential systems described using the Hubbard model.
In this system, we can view a site as a $\Motimes^2 \aleph$ typed Hydrogen atom (fermion). \Cref{eq:hubbard} describes the Hamiltonian for this system.

There are two terms in the Hamiltonian. The first accounts for the energy due to the movement (hopping) of electrons, whereas the second term accumulates the energy due to electron repulsion. As the atoms are arranged in a one-dimensional array, the electrons can only move (hop) to the neighboring adjacent atoms, such as moving from $j$-th to $j\splus 1$-th site, in this system.

When simulating the Hubbard system, users typically attempt to manipulate different $z_t$ and $z_u$ values to utilize the Hubbard system for analyzing various particle behaviors.
In simulating the Hydrogen chain  \cite{melo_2021}, we assign a constant to $z_t$ and make $z_u$ dependent on the time periods.
For a period of $T$. $z_u$ will transit from the initial value $z_{u0}$ at time $t = 0$ to the final value $Z_{uf}$ at time $t = T$,
and the equation looks like $ z_u(t) = (1 - \frac{t}{T} )z_{u0} + \frac{t}{T}z_{uf}$.
The naive compilation of the Hubbard system is similar to the one described in \Cref{sec:boson}, except that we also need to include $Z$ terms by enforcing anti-commutation.
To better compile the system, previous researchers \cite{melo_2021} tried to use unconventional quantum state mapping from a $\Motimes^2 \aleph$ typed state to a $\Motimes^2 t(2)$ typed qubit state. For example, for a two-particle system, they utilize an optimized state mapping as follows, where $\ket{k}_0$ and $\ket{k}_1$ are marked as the first and second quantum particles, respectively.

{\small
\begin{center}
$
\begin{array}{l}
\ket{1}_0\ket{1}_1\otimes \ket{0}_0\ket{0}_1\to \ket{0}\ket{0}
\qquad
\ket{1}_0\ket{0}_1\otimes \ket{0}_0\ket{1}_1\to \ket{0}\ket{1}
\\[0.1em]
\ket{0}_0\ket{1}_1\otimes \ket{1}_0\ket{0}_1\to \ket{1}\ket{0}
\qquad
\ket{0}_0\ket{0}_1\otimes \ket{1}_0\ket{1}_1\to \ket{1}\ket{1}
\end{array}
$
\end{center}
}

They assume that the other ket states do not exist.\
With the unconventional mapping, they can rewrite the Hubbard system to an Ising system as:

{\small
\begin{center}
$
\hat{H}_S=\sminus z_t(X \otimes I + I \otimes X) + z_u Z \otimes Z
$
\end{center}
}

The compilation of this new optimized system is similar to the one in \Cref{sec:qcompile}. Through Trotterization, the simulation of the system generates a series of $X$-axis rotation gates $\cn{Rx}$ as well as $\cn{ZZ}$ interaction gates.

\myparagraph{Determining the Parameters \(z_t\) and \(z_u\) in Hubbard System.}

The parameters \(z_t\) and \(z_u\) in the Hubbard system are crucial for accurately describing the physical system. For a $1D$ chain of hydrogen atoms, these parameters can be determined as follows:

1. \textbf{Hopping Integral \(z_t\)}: The hopping integral \(z_t\) represents the kinetic energy associated with an electron hopping from one site to another. It can be estimated using the overlap integral of the atomic orbitals on neighboring sites. For hydrogen atoms, the 1s orbitals are used. The hopping integral can be calculated as:

\[
z_t = \int \varphi_{1s}^*(r - R_i) \hat{H} \varphi_{1s}(r - R_j) \, dr
\]

where \(\varphi_{1s}(r)\) is the 1s orbital wave function, \(\hat{H}\) is the Hamiltonian of the system, and \(R_i\) and \(R_j\) are the positions of the neighboring atoms. In practice, this integral is often approximated using empirical or computational methods, such as density functional theory (DFT).

2. \textbf{On-Site Interaction \(z_u\)}: The on-site interaction \(z_u\) represents the Coulomb repulsion between two electrons occupying the same site. For hydrogen atoms, this can be approximated using the Coulomb integral:

\[
z_u = \int \varphi_{1s}^*(r_1) \varphi_{1s}^*(r_2) \frac{e^2}{|r_1 - r_2|} \varphi_{1s}(r_1) \varphi_{1s}(r_2) \, dr_1 \, dr_2
\]

where \(e\) is the electron charge, and \(\varphi_{1s}(r)\) is the 1s orbital wave function. This integral can also be evaluated using computational techniques, providing an estimate of the electron repulsion energy at each site.
}

\section{A Hubbard Model for Hydrogen Chains}\label{sec:hubbard} % Added a label. Feel free to change/remove

The Hubbard model (system) describes the interactions between elementary particles, specifically focusing on the electrons having fermionic behavior. 
Here, we focus on the one-dimensional Hydrogen chain, one of the most quintessential systems described using the Hubbard system, which consists of many sites of $t^{\aleph}(2)$ typed Hydrogen atoms (fermions). \Cref{sec:hubbard-example} describes a two-site version for this system, as we list the simplified version below.

{\small
\[
\hat{H}_T = \sminus z_t \sum_{j} 
\left(\sapp{\stype{\sdag{a}(j)}{t^{\aleph}(2)}}{ \stype{a(j\splus 1)}{{t^{\aleph}(2)}}} + \sapp{\stype{\sdag{a}(j\splus 1)}{{t^{\aleph}(2)}}}{ \stype{a(j)}{{t^{\aleph}(2)}}}\right) + z_u \sum_{j}  \sapp{\stype{\mathbb{1}(j)}{{t^{\aleph}(2)}}}{\stype{\mathbb{1}(j\splus 1)}{{t^{\aleph}(2)}}}
\]
}

We explain the physical constraints of the two terms in the Hamiltonian. The first accounts for the energy due to the movement (hopping) of electrons, whereas the second term accumulates the energy due to electron repulsion. As the atoms are arranged in a one-dimensional array, the electrons can only move (hop) to the neighboring adjacent atoms, such as moving from $j$-th to $j\splus 1$-th site, in this system.
%The constituent creation-annihilation products capture the contribution from every potential hopping. For instance, $a^{\dagger}_{\sigma}(1)a_{\sigma}(2)$ represents the contribution to kinetic energy by the hopping of the electron with spin $\sigma$ at site $2$ to site $1$. The scale at which these individual contributions affect the total energy is expressed through the parameter $J_{h}$, and it can be determined based on the specific simulation performed. In addition to the kinetic energy, the repulsion between electrons will also affect the total energy. The most significant interaction for an atomic Hydrogen-based system arises from the two electrons occupying the sole orbital. In scenarios where an orbital is half-filled, the Hubbard model assumes the absence of interactions between electrons in a hydrogen atom. This is captured by the second term, where $n_{\sigma}(j) = 0$ when the Hydrogen atom at site $j$ is not occupied by an electron with spin $\sigma$. $U$ scales up the contributions to the appropriate energy scale, which depends on the system of interest.\par

When simulating the Hubbard system, users typically attempt to manipulate different $z_t$ and $z_u$ values to utilize the Hubbard system for analyzing various particle behaviors.
In simulating the Hydrogen chain  \cite{melo_2021}, we assign a constant to $z_t$ and make $z_u$ dependent on the time periods.
For a period of $T$. $z_u$ will transit from the initial value $z_{u0}$ at time $t = 0$ to the final value $Z_{uf}$ at time $t = T$,
and the equation looks like $ z_u(t) = (1 - \frac{t}{T} )z_{u0} + \frac{t}{T}z_{uf}$.

\myparagraph{Determining the Parameters \(z_t\) and \(z_u\) in Hubbard System}
The parameters \(z_t\) and \(z_u\) in the Hubbard system are crucial for accurately describing the physical system. For a $1D$ chain of hydrogen atoms, these parameters can be determined as follows:

1. \textbf{Hopping Integral \(z_t\)}: The hopping integral \(z_t\) represents the kinetic energy associated with an electron hopping from one site to another. It can be estimated using the overlap integral of the atomic orbitals on neighboring sites. For hydrogen atoms, the 1s orbitals are used. The hopping integral can be calculated as:

\[
z_t = \int \psi_{1s}^*(r - R_i) \hat{H} \psi_{1s}(r - R_j) \, dr
\]

where \(\psi_{1s}(r)\) is the 1s orbital wave function, \(\hat{H}\) is the Hamiltonian of the system, and \(R_i\) and \(R_j\) are the positions of the neighboring atoms. In practice, this integral is often approximated using empirical or computational methods, such as density functional theory (DFT).

2. \textbf{On-Site Interaction \(z_u\)}: The on-site interaction \(z_u\) represents the Coulomb repulsion between two electrons occupying the same site. For hydrogen atoms, this can be approximated using the Coulomb integral:

\[
z_u = \int \psi_{1s}^*(r_1) \psi_{1s}^*(r_2) \frac{e^2}{|r_1 - r_2|} \psi_{1s}(r_1) \psi_{1s}(r_2) \, dr_1 \, dr_2
\]

where \(e\) is the electron charge, and \(\psi_{1s}(r)\) is the 1s orbital wave function. This integral can also be evaluated using computational techniques, providing an estimate of the electron repulsion energy at each site.

\begin{figure}[h]
{\small
  \[
  \begin{array}{c}
  \aleph \;\;::=\;\; \uparrow \; \mid \; \downarrow 
        \end{array}
  \]
}
{\small
\begin{mathpar}  
        \inferrule[S-Ten]{}
        { \denote{\inferrule[]{\tjudge{\iota}{e}{\quan{F}{\zeta}{\iota}}\\ \tjudge{\iota'}{e'}{\quan{F}{\zeta}{\iota'}}}{\tjudge{\iota \ttimes \iota'}{e \otimes e'}{\quan{F}{\zeta}{\iota \otimes \iota'}}}}_g (w_1 \otimes w_2) := \denote{\tjudge{\iota}{e}{\quan{F}{\zeta}{\iota}}}_{g}{w_1}\,\textcolor{purple}{\otimes}\,\denote{\tjudge{\iota'}{e'}{\quan{F}{\zeta}{\iota'}}}_{\textcolor{spec}{\funsa{S}{\iota}{w_1}}(g)}{w_2} }
\end{mathpar}
}
{\small
\[
\textcolor{spec}
{
\begin{array}{l}
\funsa{S}{t^{\uparrow}(m)}{\ket{j}}(g_1,g_2) = (g_1 +j, g_2)
\qquad
\funsa{S}{t^{\downarrow}(m)}{\ket{j}}(g_1,g_2) = (g_1, g+j)
\\
\funsa{S}{t(m)}{\ket{j}}(g_1,g_2) = (g_1,g_2)
\qquad
\funsa{S}{\iota \otimes \iota'}{\eta \otimes \eta'}(g) = \funsa{S}{\iota'}{\eta'}(\funsa{S}{\iota}{\eta}(g))
\end{array}
}
\]
}
\vspace*{-1.2em}
  \caption{Extneded \qsnd to spinning fermions.}
  \label{fig:fermionextended}
  \vspace*{-1em}
\end{figure}

\myparagraph{Including Spin Terms.}
In the traditional Hubbard model, there are two types of fermions, spin-up ($\uparrow$) and spin-down ($\downarrow$) fermions.
In constraining the system, we can think of the two kinds of fermions as two different particles and split the fermion type flag $\aleph$ in the \qsnd, into two different type flags $\uparrow$ and $\downarrow$. We then rewrite the semantic rule for \rulelab{S-Ten} to be the one in \Cref{fig:fermionextended}. The main difference is to modify the context $g$ to become a pair of real numbers, one for recording the effects of each type of the two fermions.
We can then rewrite the Hubbard model equation above to the one below.

{\small
\[
\hat{H}_T = \sminus z_t \sum_{j} 
\left(\sapp{{\sdag{a}(j)}}{ {a(j\splus 2)}} + \sapp{{\sdag{a}(j\splus 2)}}{ {a(j)}}\right) + z_u \sum_{\cn{even}(j)}  \sapp{{\mathbb{1}(j)}}{{\mathbb{1}(j\splus 1)}}
\]
}

In the new Hamiltonian, we model the $\uparrow$ and $\downarrow$ fermions to appear alternatively; that is, odd-indexed sites store the $\uparrow$ fermion, and even-indexed sites store the $\downarrow$ fermion. The first term in the Hamiltonian restricts that fermions can only reach same kind fermion sites, e.g., $\sapp{{\sdag{a}(j)}}{ {a(j\splus 2)}}$ means that the disappearing of a fermion in the $j\splus 2$ site results in the fermion appearing in the $j$-th site.
The second term above is a constraint that we restrict the appearing of opposite kind fermions in the adjacent sites, e.g., $\sapp{{\mathbb{1}(j)}}{{\mathbb{1}(j\splus 1)}}$ means that $j$-th and $j\splus 1$-th sites are not likely to have fermions occupied at the same time. Since the term is a sum over $\cn{even}(j)$, the above restriction only constrains the same kind of fermions. 

\ignore{
\section{Quantum Mechanics}\label{appx:qmechan}

\subsection{Hilbert Space \(\mathpzc{H}\)}

A complex vector space with an inner product satisfying:

\begin{itemize}
  \item Conjugate Symmetry: \(\langle \psi | \phi \rangle = \langle \phi | \psi \rangle^\ast\)
  \item Linearity: \(\langle \psi | a\phi + b\chi \rangle = a\langle \psi | \phi \rangle + b\langle \psi | \chi \rangle\), where \(a\) and \(b\) are complex numbers.
  \item Positive Definiteness: \(\langle \psi | \psi \rangle \geq 0\) and the equality holds, if and only of \(|\psi\rangle = 0\).
\end{itemize}

\subsection{Completeness}

Every Cauchy sequence \(\{|\psi_n\rangle\}\) in \(\mathpzc{H}\) converges to an element in \(\mathpzc{H}\). A Cauchy sequence is defined where, for every \(\epsilon > 0\), there exists an \(N\) such that \(|\psi_n - \psi_m| < \epsilon\) for all \(m, n > N\). (For finite-dimensional Hilbert spaces, the completeness requirement is trivially satisfied.)

\subsection{Composite Systems}

\begin{itemize}
  \item Direct Sum (\(\oplus\)): Represents the combination of two independent quantum systems described by Hilbert spaces \(\mathpzc{H}_1\) and \(\mathpzc{H}_2\). The direct sum space \(\mathpzc{H}_1 \oplus \mathpzc{H}_2\) consists of all ordered pairs \(|\psi_1\rangle, |\psi_2\rangle\) where \(|\psi_1\rangle \in \mathpzc{H}_1\) and \(|\psi_2\rangle \in \mathpzc{H}_2\). It represents scenarios where the systems do not interact or influence each other.
  
  \item Tensor Product (\(\otimes\)): Describes the space for composite quantum systems. The tensor product \(\mathpzc{H}_1 \otimes \mathpzc{H}_2\) is formed by combining elements from \(\mathpzc{H}_1\) and \(\mathpzc{H}_2\) such that every pair of vectors \(|\psi_1\rangle \in \mathpzc{H}_1\) and \(|\psi_2\rangle \in \mathpzc{H}_2\) contributes a new vector \(|\psi_1\rangle \otimes |\psi_2\rangle\) to \(\mathpzc{H}_1 \otimes \mathpzc{H}_2\). This space encompasses all possible states of the combined system, including entangled states.
\end{itemize}

\subsection{Operators}

\begin{itemize}
  \item Direct Sum: For operators \(\hat{A}, \hat{B}\), the direct sum \((\hat{A} \oplus \hat{B})\) acts on each component of the direct sum space independently: \((\hat{A} \oplus \hat{B})(|\psi_1\rangle, |\psi_2\rangle) = (\hat{A}|\psi_1\rangle, \hat{B}|\psi_2\rangle)\).
  
  \item Tensor Product: For operators \(\hat{A}, \hat{B}\), the tensor product \((\hat{A} \otimes \hat{B})\) combines the actions of each operator on the respective Hilbert spaces: \((\hat{A} \otimes \hat{B})(|\psi_1\rangle \otimes |\psi_2\rangle) = (\hat{A}|\psi_1\rangle) \otimes (\hat{B}|\psi_2\rangle)\).
  
  \item Hermitian Operators: An operator \(\hat{A}\) is Hermitian if it equals its own adjoint: \(\hat{A} = \hat{A}^\dagger\). This means for all vectors \(|\psi\rangle, |\phi\rangle\) in \(\mathpzc{H}\), we have \(\langle \psi|\hat{A}|\phi \rangle = \langle \hat{A}\psi|\phi \rangle^\ast\). Hermitian operators represent observable quantities in quantum mechanics and have real eigenvalues.
  
  \item Unitary Operators: An operator \(\hat{U}\) is Unitary if its inverse is its adjoint: \(\hat{U}^{-1} = \hat{U}^\dagger\). This implies \(\hat{U}\hat{U}^\dagger = \hat{U}^\dagger\hat{U} = \hat{I}\), where \(\hat{I}\) is the identity operator. Unitary operators preserve the inner product and are used to describe the time evolution and symmetries in quantum systems.
  
  \item Projection Operators: A Projection operator \(\hat{P}\) satisfies \(\hat{P}^2 = \hat{P}\) and \(\hat{P} = \hat{P}^\dagger\). These operators project vectors onto a subspace of the Hilbert space. If \(|\phi\rangle\) is a vector in the space, then \(\hat{P}|\phi\rangle\) is the projection of \(|\phi\rangle\) onto the subspace defined by \(\hat{P}\).
\end{itemize}

}

\ignore{
\section{Quantum Mechanics Representations and Many-particle quantum mechanics}\label{sec:qmechan}

This section introduces some concepts in quantum mechanics.

\subsection{Fundamental Postulates of Quantum Mechanics}

\begin{enumerate}
  \item State Postulate: Every quantum system is completely described by its state vector, which is a unit vector in a complex Hilbert space \(\mathpzc{H}\). The state vector is represented as \(|\psi \rangle\). In mathematical terms: \(|\psi \rangle \in \mathpzc{H}\), with \(\langle \psi | \psi \rangle = 1\).
  
  \item Observable Postulate: Physical observables in quantum mechanics are represented by Hermitian operators (denoted as \(\hat{A}, \hat{B}, \hat{C}, \ldots\)) acting on the Hilbert space \(\mathpzc{H}\). The possible measurement outcomes of an observable correspond to the eigenvalues of its associated Hermitian operator. Mathematically, if \(\hat{A}\) is an observable, and \(\lambda\) is a measurement outcome, then: \(\hat{A} |a \rangle = \lambda |a \rangle\), where \(|a \rangle \in \mathpzc{H}\).
  
  \item Measurement Postulate: If an observable \(\hat{A}\) with non-degenerate eigenvalues is measured in a system in state \(|\psi \rangle\), the probability of obtaining eigenvalue \(\lambda\) is given by \(|\langle a |\psi \rangle|^2\), where \(|a \rangle\) is the eigenvector of \(\hat{A}\) associated with \(\lambda\). After the measurement, the state of the system collapses to \(|a \rangle\). For a normalized state, this is expressed as: \(P(\lambda) = |\langle a |\psi \rangle|^2\), Post-measurement state: \(|\psi' \rangle = \frac{\hat{P}_a |\psi \rangle}{\sqrt{\langle \psi |\hat{P}_a |\psi \rangle}}\) where \(\hat{P}_a = |a \rangle \langle a |\) is the projection operator onto the state \(|a \rangle\).
  
  \item Evolution Postulate: The time evolution of a quantum state is governed by the Schrödinger equation. If \(|\psi(t) \rangle\) describes the state of the system at time \(t\), then its time evolution is given by: \(i \hbar \frac{d}{dt} |\psi(t) \rangle = \hat{H} |\psi(t) \rangle\) where \(\hat{H}\) is the Hamiltonian operator of the system, which corresponds to the total energy observable.
\end{enumerate}

For further details, see \cite{dirac,griffiths,eisberg}.

\subsection{Position Representation}
In the position representation, states are expressed as wave functions $\psi(x)$ in real space, where $x$ denotes the position. The wave function gives the probability amplitude for finding a particle at position $x$. The position operator $\hat{x}$ and the momentum operator $\hat{p}$ are represented as:
\begin{align}
\hat{x} \psi(x) &= x \psi(x), \\
\hat{p} \psi(x) &= -i \hbar \frac{d}{dx} \psi(x).
\end{align}

\subsection{Momentum Representation}
In the momentum representation, states are described as wave functions $\phi(p)$ in momentum space. The operators are represented as:
\begin{align}
\hat{p} \phi(p) &= p \phi(p), \\
\hat{x} \phi(p) &= i \hbar \frac{d}{dp} \phi(p).
\end{align}

\subsection{Energy Representation}
In systems with a well-defined Hamiltonian, states can be represented in terms of energy eigenstates. This representation is particularly useful for solving the Schrödinger equation.

\subsection{Second Quantization (Occupation Number Representation)}
Second quantization, or the occupation number representation, describes systems with variable numbers of indistinguishable particles by the number of particles occupying each possible state, referred to as a mode. The state of a system can be represented as:
\begin{equation}
|n_1, n_2, n_3, \ldots \rangle
\end{equation}
Creation ($\hat{a}^\dagger_i$) and annihilation ($\hat{a}_i$) operators, which satisfy commutation or anticommutation relations for bosons and fermions, respectively, are introduced to change the occupation numbers of the modes.

\subsubsection{Commutation Properties and Actions of Creation and Annihilation Operators}
The creation and annihilation operators play a crucial role in the framework of second quantization, with their actions and commutation or anticommutation relations defined based on the type of particles (bosons or fermions).

\textbf{For Bosons:}
\begin{align}
[\hat{a}_i, \hat{a}_j^\dagger] &= \delta_{ij}, \\
[\hat{a}_i, \hat{a}_j] &= [\hat{a}_i^\dagger, \hat{a}_j^\dagger] = 0.
\end{align}

\textbf{For Fermions:}
\begin{align}
\{\hat{a}_i, \hat{a}_j^\dagger\} &= \delta_{ij}, \\
\{\hat{a}_i, \hat{a}_j\} &= \{\hat{a}_i^\dagger, \hat{a}_j^\dagger\} = 0.
\end{align}

\subsection{Second Quantization in Many-Particle Quantum Mechanics}

Many-particle quantum mechanics underpins the theoretical framework necessary for understanding phenomena across condensed matter physics, chemical systems, and nuclear physics. Traditional wave function approaches scale poorly with the number of particles due to the combinatorial explosion of configuration space. Second quantization addresses these complexities by focusing on occupation numbers and field operators rather than individual particle coordinates, thus offering a scalable approach to studying systems of identical particles.

Second quantization transcends the limitations of first quantization by introducing a more abstract but immensely powerful framework, ideally suited for systems of identical particles, like electrons in a metal or photons in electromagnetic fields.

\subsubsection{Creation and Annihilation Operators}
Central to the formalism are the creation (\(\hat{c}^\dagger_{i,\sigma}\)) and annihilation (\(\hat{c}_{i,\sigma}\)) operators, which respectively add and remove particles from quantum states. These operators are defined for each quantum state labeled by index \(i\) and spin \(\sigma\), following specific algebraic rules:

For fermions (e.g., electrons), the anticommutation relations are:
\begin{align*}
    \{\hat{c}_{i,\sigma}, \hat{c}_{j,\sigma'}^\dagger\} &= \delta_{ij} \circ \delta_{\sigma\sigma'}, \\
    \{\hat{c}_{i,\sigma}, \hat{c}_{j,\sigma'}\} &= \{\hat{c}_{i,\sigma}^\dagger, \hat{c}_{j,\sigma'}^\dagger\} = 0,
\end{align*}
reflecting the Pauli exclusion principle, ensuring no two fermions can occupy the same quantum state.

For bosons, the commutation relations are:
\begin{align*}
    [\hat{b}_{i}, \hat{b}_{j}^\dagger] &= \delta_{ij}, \\
    [\hat{b}_{i}, \hat{b}_{j}] &= [\hat{b}_{i}^\dagger, \hat{b}_{j}^\dagger] = 0,
\end{align*}
allowing any number of bosons to occupy the same state.

\subsubsection{Fock Space and Quantum States}
Fock space, or the state space of many-body systems, is constructed from the vacuum state \(\ket{0}\), which contains no particles. States with particles are built by applying creation operators to the vacuum:
\begin{equation*}
    \ket{n_{1,\uparrow}, n_{1,\downarrow}, n_{2,\uparrow}, \ldots} = \prod_{i} (\hat{c}_i^\dagger)^{n_i} \ket{0},
\end{equation*}
Here, \(n_{i,\sigma}\) are occupation numbers indicating how many particles occupy the state specified by \(i\) and \(\sigma\). For fermions, \(n_{i,\sigma}\) can be either 0 or 1, while for bosons, it can be any non-negative integer.

\subsubsection{Operators in Fock Space}
Within Fock space, physical observables are represented by operators constructed from the creation and annihilation operators. The number operator \(\hat{n}_{i,\sigma} = \hat{c}_{i,\sigma}^\dagger \circ \hat{c}_{i,\sigma}\) counts the number of particles in state \((i, \sigma)\). 

The total number operator for the system is:
\begin{equation*}
    \hat{N} = \sum_{i,\sigma} \hat{n}_{i,\sigma},
\end{equation*}
which sums the occupation numbers across all states.

The kinetic energy of particles hopping between lattice sites or states can be modeled as:
\begin{equation*}
    \hat{T} = \sum_{ij,\sigma} t_{ij} \hat{c}_{i,\sigma}^\dagger \circ \hat{c}_{j,\sigma},
\end{equation*}
where \(t_{ij}\) represents the kinetic energy matrix elements, indicating the probability amplitude for a particle to hop from state \(j\) to state \(i\).

Potential energy terms, often arising from external fields or intrinsic properties, are represented as:
\begin{equation*}
    \hat{V} = \sum_{i,\sigma} V_i \circ \hat{n}_{i,\sigma},
\end{equation*}
where \(V_i\) denotes the potential energy associated with particles in state \(i\).

Interaction terms, particularly relevant in electron systems for capturing Coulomb repulsion, take the form:
\begin{equation*}
    \hat{H}_{\text{int}} = \frac{1}{2} \sum_{ijkl} \sum_{\sigma,\sigma'} U_{ijkl} \circ \hat{c}_{i,\sigma}^\dagger \circ \hat{c}_{j,\sigma'}^\dagger \circ \hat{c}_{l,\sigma'} \circ \hat{c}_{k,\sigma},
\end{equation*}
where \(U_{ijkl}\) are the elements of the interaction matrix, encapsulating the strength and nature of particle-particle interactions within the many-body system.

\subsubsection{Quantum Statistics and Symmetry}
The statistical properties of particles in quantum mechanics — Fermi-Dirac for fermions and Bose-Einstein for bosons — naturally emerge from the anticommutation and commutation relations of the creation and annihilation operators. This elegant unification underlies the power of second quantization, encapsulating the quantum statistics inherently without imposing them externally.

Moreover, the symmetries of physical systems, such as translational, rotational, and spin symmetries, can be systematically explored within this framework by examining the commutation relations of the Hamiltonian with relevant symmetry operators, leading to conservation laws and selection rules essential for understanding quantum dynamics and spectra.

\subsection{Conclusion}
The second quantization formalism provides a potent and versatile approach to describing and analyzing many-body quantum systems. By transcending the limitations of traditional wave function approaches, it lays the groundwork for a deeper understanding of complex quantum phenomena, from superfluidity and superconductivity to quantum magnetism and beyond.
For further details, see \cite{fetter,altland,negele}.
}

\ignore{
\section{Quantum Spin Lattice Models}
Quantum spin lattice models describe interacting spin systems arranged on a lattice. These models are essential for studying magnetic interactions, quantum entanglement, and the statistical mechanics of spin systems.

\subsection{N-dimensional Chain of Sites}
Consider a lattice composed of \(N\) sites, arranged in one or more dimensions. Each site, labeled by an index \(i = 1, 2, \ldots, N\), hosts a quantum spin entity. The local spin at each site can occupy a space spanned by a set of basis states. For example, in a spin-\(1/2\) system, the basis at each site would be \(\{\ket{\uparrow}, \ket{\downarrow}\}\).

\subsection{Fock Space of Spin Lattice Model}
The Fock space, or the Hilbert space, for the entire spin lattice system, is constructed as the tensor product of the Hilbert spaces of individual sites:
\begin{equation*}
    \mathpzc{H} = \bigotimes_{i=1}^N \mathpzc{H}_i,
\end{equation*}
where \(\mathpzc{H}_i\) represents the Hilbert space associated with the \(i\)-th site. For a spin-\(1/2\) lattice, this equates to considering all possible configurations of spins up and down across the entire lattice. This Fock space accommodates all potential states of the system, ranging from all spins aligned (e.g., \(\ket{\uparrow, \uparrow, \ldots, \uparrow}\)) to all possible mixed configurations (e.g., \(\ket{\uparrow, \downarrow, \ldots, \uparrow}\)).

\subsection{Many-Particle Operators in Fock Space}
Operators that act within this Fock space can represent various physical quantities, like total spin, magnetic moment, or spin-spin interaction energies. An example of such an operator is the total spin in the z-direction:
\begin{equation*}
    \hat{S}^z_{\text{total}} = \sum_{i=1}^N \hat{S}_i^z,
\end{equation*}
which sums the z-component of the spin at each site across the entire lattice.

\subsection{Connection between Sigma and Spin Operators}
The spin operators for a spin-\(\frac{1}{2}\) particle are related to the Pauli matrices as follows:
\begin{align*}
    S^x &= \frac{\hbar}{2}\sigma^x, \quad S^y = \frac{\hbar}{2}\sigma^y, \quad S^z = \frac{\hbar}{2}\sigma^z, \\
    \sigma^x &= \begin{pmatrix} 0 & 1 \\ 1 & 0 \end{pmatrix}, \quad
    \sigma^y = \begin{pmatrix} 0 & -i \\ i & 0 \end{pmatrix}, \quad
    \sigma^z = \begin{pmatrix} 1 & 0 \\ 0 & -1 \end{pmatrix}.
\end{align*}
These operators facilitate the description and manipulation of spins within the lattice model.

\subsection{Commutation Relations of Pauli Matrices}
The Pauli matrices satisfy the following commutation relations:
\begin{align*}
    [\sigma^x, \sigma^y] &= 2i\sigma^z, \\
    [\sigma^y, \sigma^z] &= 2i\sigma^x, \\
    [\sigma^z, \sigma^x] &= 2i\sigma^y.
\end{align*}
These relations are fundamental in the algebra of spin operators and play a crucial role in defining the dynamics and interactions within quantum spin systems.

\subsection{Raising and Lowering Operators}
The spin raising (\(\sigma^+\)) and lowering (\(\sigma^-\)) operators are defined as:
\begin{align*}
    \sigma^+ &= \sigma^x + i\sigma^y = \begin{pmatrix} 0 & 2 \\ 0 & 0 \end{pmatrix}, \\
    \sigma^- &= \sigma^x - i\sigma^y = \begin{pmatrix} 0 & 0 \\ 2 & 0 \end{pmatrix}.
\end{align*}
These operators are used to flip the spins: \(\sigma^+\) raises the spin from down to up, while \(\sigma^-\) lowers it from up to down.

\subsection{Properties of Sigma Operators}
The raising and lowering operators \(\sigma^+\) and \(\sigma^-\) satisfy the following properties:
\begin{align*}
    \sigma^+ \ket{\downarrow} &= \ket{\uparrow}, \quad \sigma^+ \ket{\uparrow} = 0, \\
    \sigma^- \ket{\uparrow} &= \ket{\downarrow}, \quad \sigma^- \ket{\downarrow} = 0.
\end{align*}
Furthermore, they relate to \(\sigma^z\) through the commutators:
\begin{align*}
    [\sigma^z, \sigma^+] &= 2\sigma^+, \\
    [\sigma^z, \sigma^-] &= -2\sigma^-.
\end{align*}
}

\ignore{
\subsection{Second Quantization in Spin Lattice Models}
In the context of quantum spin lattice models, second quantization introduces operators that alter the spins at lattice sites. These include the spin raising and lowering operators (\(S_i^+\) and \(S_i^-\)), which are analogous to the creation and annihilation operators but for spins.

\subsubsection{Spin Raising and Lowering Operators}
The operators increase or decrease the spin component along the z-axis, respectively:
\begin{align}
S_i^+ &= S_i^x + i S_i^y, \\
S_i^- &= S_i^x - i S_i^y
\end{align}

\subsubsection{Action on Spins}
The action of these operators on the spin at site \(i\) is:
\begin{align}
S_i^+ |\downarrow_i \rangle &= |\uparrow_i \rangle, \\
S_i^- |\uparrow_i \rangle &= |\downarrow_i \rangle
\end{align}
These equations model the flipping of spins within the lattice.
}

\section{Other Physical Models Definable in \qsnd}\label{sec:othermodels}

Here, we show many other systems that are definable in \qsnd.

\subsection{Ising Model}
The Ising model is one of the simplest quantum spin lattice models. It considers spins that can be in one of two states (up or down) interacting with their nearest neighbors. The Hamiltonian for the quantum Ising model in a transverse field is given by:
\begin{equation}
H = -J \sum_{\langle i,j \rangle} Z(i) \circ Z(j) - h \sum_i X(i)
\end{equation}
where \(J\) represents the interaction strength between neighboring spins, \(h\) is the external magnetic field (which is a constant), and the sum \(\langle i,j \rangle\) runs over all pairs of nearest neighbors.

\subsection{Heisenberg Model}
The Heisenberg model includes interaction in all three spin components. Its Hamiltonian is expressed as:
\begin{equation}
H = -J \sum_{\langle i,j \rangle} (X(i) \circ  X(j) + Y(i) \circ  Y(j) + Z(i) \circ  Z(j))
\end{equation}
In this model, \(J\) represents the interaction strength, and the \(S\) terms are spin-\(\frac{1}{2}\) operators, allowing for complex spin interactions.

\subsection{XY Model}
The XY model restricts interactions to the X and Y axis components. Its Hamiltonian is:
\begin{equation}
H = -J \sum_{\langle i,j \rangle} (X(i) \circ  X(j) + Y(i) \circ  Y(j))
\end{equation}

The quantum spin lattice models are fundamental for understanding quantum phase transitions, magnetic properties, and many-body quantum phenomena. They provide a rich framework for exploring quantum correlations, entanglement, and the effects of quantum fluctuations on macroscopic systems.

\subsection{t-J Model}\label{sec:tj}

%\liyi{central points: show uitilities. CS people like "deatails".}
%\liyi{Show Ham, --> explain the system in terms of operations in Ham. if no different, maybe just point to the previous. }
%\liyi{compute energy for an example. maybe. or do similar things. show 2 particle or 3 particles. }
%\liyi{Show the exp expansion --> how it is done. }
%\liyi{Show the simulation result. Show how the exponential term produces a unitary/circuit. the exp of the system relates to some quantum circuit, and produce some simulation results. }

The t-J model is a pivotal framework in condensed matter physics for exploring high-temperature superconductivity in cuprate materials. Originating from the Hubbard model in the limit of strong on-site Coulomb repulsion, it prohibits double occupancy of lattice sites, thereby allowing for a focused analysis of electron hopping and spin interactions. This approach is key to understanding superconductivity and magnetism in strongly correlated electron systems.
\myparagraph{Hilbert Space}
The Hilbert space \(\mathpzc{H}\) excludes states with double occupancy. Basis states \(|s_1, s_2, ..., s_N \rangle\) represent each site as empty (0), occupied by an electron with spin up (\(\uparrow\)), or spin down (\(\downarrow\)). This space is formally defined as:
\[
\mathpzc{H} = \{ |s_1, s_2, ..., s_N \rangle : s_i \in \{0, \uparrow, \downarrow\} \, \forall \, i \}
\]

\myparagraph{Hamiltonian}
The Hamiltonian \(H\) includes terms for electron mobility and spin interactions:
\[
H = -t \sum_{\langle ij \rangle, \sigma} (a^\dagger(i,\sigma) \circ a(j,\sigma) + a^\dagger(j,\sigma) \circ  a(i,\sigma)) + J \sum_{\langle ij \rangle} (\mathbf{S}(i) \cdot \mathbf{S}(j) - \frac{1}{4}\mathbb{1}(i) \circ  \mathbb{1}(j))
\]

\myparagraph{Hopping Term (\(-t\))}
Allows electrons to move between adjacent sites, governed by:
\[
-t \sum_{\langle ij \rangle, \sigma} (a^\dagger(i,\sigma) \circ  a(j,\sigma) + a^\dagger(j,\sigma) \circ  a(i,\sigma))
\]

\myparagraph{Spin-Spin Interaction Term (\(J\))}
Models antiferromagnetic interactions:
\[
J \sum_{\langle ij \rangle} (\mathbf{S}(i) \cdot \mathbf{S}(i) - \frac{1}{4}\mathbb{1}(i) \circ  \mathbb{1}(j))
\]

\myparagraph{Spin Operators and Pauli Matrices}
\[
S^x(i) = \frac{\hbar}{2} X, \quad S^y(i) = \frac{\hbar}{2} Y, \quad S^z(i) = \frac{\hbar}{2} Z
\]

For further details, see \cite{sachdev,auerbach,anderson}.

\myparagraph{Transformed t-J Model}
The Jordan-Wigner transformation allows us to map the fermionic operators in the t-J model into spin operators. Below is the transformed t-J model for a one-dimensional lattice system.

\myparagraph{Hopping Term}
The fermionic hopping terms in the t-J model are transformed as follows:
\begin{equation}
-t \sum_{\langle i,j \rangle} \left( \exp\left(i \pi \sum_{l=i}^{j-1} n_l\right) S_i^+  \circ  S_j^- + \text{h.c.} \right),
\end{equation}
where \(S_i^+\) and \(S_j^-\) are spin raising and lowering operators at sites \(i\) and \(j\), respectively. The term \(n_l = \frac{1}{2}(I - Z)\) corresponds to the number operator after transformation, and the exponential term arises due to the non-local string of \(Z\) operators, ensuring the preservation of fermionic anticommutation relations.

\myparagraph{Spin-Spin Interaction Term}
The exchange interaction term, crucial in the t-J model, transforms into:
\begin{equation}
J \sum_{\langle i,j \rangle} \left(\frac{1}{2}(S_i^+ \circ  S_j^- + S_i^- \circ  S_j^+) + S_i^z \circ  S_j^z\right),
\end{equation}
where the terms now involve combinations of the spin raising, lowering, and z-component operators, indicative of the transformed spin interactions.

\ignore{
\section{Jordan-Wigner Transformation}\label{sec:jw-tramsform}
The Jordan-Wigner transformation is a critical tool in quantum mechanics and condensed matter physics, allowing for the mapping between spin systems and fermionic systems. This transformation is particularly useful in the study of lattice models where it enables the analysis of spin chains in the language of fermions, facilitating the application of methods and insights from fermionic systems to solve spin-based problems.

The Jordan-Wigner transformation provides a way to convert between spin operators and fermionic creation and annihilation operators. This is crucial in one-dimensional lattice models, enabling the transformation of spin models, like the XY model or the transverse field Ising model, into models of spinless fermions.

Given a one-dimensional lattice with \(N\) sites, the transformation is defined as follows:
\begin{align*}
    c_j &= \left( \prod_{m=1}^{j-1} \sigma_m^z \right) \sigma_j^-, \\
    c_j^\dagger &= \left( \prod_{m=1}^{j-1} \sigma_m^z \right) \sigma_j^+,
\end{align*}
where \(c_j\) and \(c_j^\dagger\) are the fermionic annihilation and creation operators at site \(j\), and \(\sigma_j^+\) (\(\sigma_j^-\)) are the spin raising (lowering) operators at site \(j\).

Incorporating the shorthand \( K_j = \prod_{m < j} \sigma_m^z \), the Jordan-Wigner transformation maps the fermionic operators as follows:
\begin{align*}
c_{j,\sigma} &= K_j \sigma_j^-, \\
c_{j,\sigma}^\dagger &= K_j \sigma_j^+.
\end{align*}
This representation simplifies the expression of fermionic operators in terms of spin variables, a transformation that proves especially useful in one-dimensional lattice systems.

\subsection{Transformed t-J Model}
The Jordan-Wigner transformation allows us to map the fermionic operators in the t-J model into spin operators. Below is the transformed t-J model for a one-dimensional lattice system.

\subsubsection{Hopping Term}
The fermionic hopping terms in the t-J model are transformed as follows:
\begin{equation}
-t \sum_{\langle i,j \rangle} \left( \exp\left(i \pi \sum_{l=i}^{j-1} n_l\right) S_i^+ S_j^- + \text{h.c.} \right),
\end{equation}
where \(S_i^+\) and \(S_j^-\) are spin raising and lowering operators at sites \(i\) and \(j\), respectively. The term \(n_l = \frac{1}{2}(1 - \sigma_l^z)\) corresponds to the number operator after transformation, and the exponential term arises due to the non-local string of \(\sigma^z\) operators, ensuring the preservation of fermionic anticommutation relations.

\subsubsection{Spin-Spin Interaction Term}
The exchange interaction term, crucial in the t-J model, transforms into:
\begin{equation}
J \sum_{\langle i,j \rangle} \left(\frac{1}{2}(S_i^+ S_j^- + S_i^- S_j^+) + S_i^z S_j^z\right),
\end{equation}
where the terms now involve combinations of the spin raising, lowering, and z-component operators, indicative of the transformed spin interactions.
}

\ignore{
\section{Bravyi-Kitaev Transformation}
The Bravyi-Kitaev transformation offers an alternative to the Jordan-Wigner transformation, providing a more efficient mapping of fermionic operators to qubits for certain quantum simulation tasks. This transformation is particularly beneficial for quantum computation, where it can lead to reduced gate counts in quantum algorithms.

\subsection{Fundamentals of the Bravyi-Kitaev Transformation}
The Bravyi-Kitaev transformation redefines fermionic creation and annihilation operators in terms of Pauli operators, in a manner that maintains the anticommutation relations while optimizing the locality of terms. For a system of \( N \) fermionic modes, the transformation can be expressed as:
\begin{equation*}
    \hat{c}_j \rightarrow \frac{1}{2} (\hat{Q}_j \hat{X}_j + i \hat{Q}_j \hat{Y}_j), \quad \hat{c}_j^\dagger \rightarrow \frac{1}{2} (\hat{Q}_j \hat{X}_j - i \hat{Q}_j \hat{Y}_j),
\end{equation*}
where \( \hat{X}_j \), \( \hat{Y}_j \), and \( \hat{Q}_j \) are products of Pauli matrices that encode the fermionic anticommutation relations into qubit operations. The exact form of \( \hat{Q}_j \) depends on the specific fermionic mode \( j \) and the chosen encoding.

}

\myparagraph{Transformed t-J Model via BK transformation}
In the context of the t-J model, the Bravyi-Kitaev transformation enables the efficient simulation of electron dynamics and spin interactions on a quantum computer. The transformed Hamiltonian in the Bravyi-Kitaev framework for the t-J model takes the form:

\[
\begin{array}{l}
    \hat{H}_{\text{t-J}}^{\text{BK}} = -t \sum_{\langle m,j \rangle, \sigma} \left(\frac{1}{2} (\hat{Q}(m) \circ \hat{X}(m) - i \hat{Q}(m) \circ \hat{Y}(n)) \frac{1}{2} (\hat{Q}(j) \circ \hat{X}(j) + i \hat{Q}(j) \circ \hat{Y}(j)) + \text{h.c.}\right) \\
    \qquad\qquad + J \sum_{\langle m,j \rangle} \left(\mathbf{S}(m) \cdot \mathbf{S}(j) - \frac{1}{4} \mathbb{1}(m) \circ \mathbb{1}(j)\right)
    \end{array}
\]

where \( \mathbf{S}(m) \) and \( \mathbb{1}(m) \) are expressed in terms of qubit operators using the Bravyi-Kitaev encoding. The hopping terms (\( t \)) and the spin-spin interaction terms (\( J \)) now involve the manipulation of qubits instead of fermions, which is more natural for quantum computational frameworks.

The Bravyi-Kitaev transformation offers improved scaling of quantum resources compared to the Jordan-Wigner transformation and can be crucial for implementing quantum simulations of condensed matter systems in practice.

The Bravyi-Kitaev transformation provides a significant advantage in quantum simulation by mapping fermionic algebra to qubit operations more efficiently than the Jordan-Wigner transformation. It is especially useful in the context of quantum computing, where it enables more effective resource utilization and can potentially lead to faster quantum algorithms for simulating many-particle systems in condensed matter physics.

\end{document}